\pdfoutput=1

\documentclass[11pt]{scrartcl}
\usepackage[a4paper, total={16cm, 24cm}]{geometry}
\usepackage{natbib}
\usepackage{authblk}

\author[1]{Niclas Boehmer}
\author[2]{Robert Bredereck}
\author[3]{Dominik Peters}
\affil[1]{TU Berlin} 
\affil[1]{TU Clausthal} 
\affil[3]{CNRS, LAMSADE, Universit\'e Paris Dauphine--PSL}

\usepackage[pagebackref]{hyperref}
\hypersetup{
	pdfencoding=auto, 
	psdextra,
	colorlinks=true,
	citecolor=green!50!black,
	linkcolor=red!60!black,
}

\usepackage{amsthm}
\usepackage{amssymb}
\usepackage{amsfonts}
\usepackage{amsmath}
\usepackage{microtype}
\usepackage[nameinlink]{cleveref}
\usepackage{pgfplots}      
\pgfplotsset{compat = 1.3}
\usepackage{caption}
\usepackage{subcaption}
\usepackage{booktabs}
\usepackage{tabularx}
\usepackage[inline]{enumitem}

\usepackage{thmtools}
\usepackage{thm-restate}
\newcommand{\prooflink}[1]{\marginline{\vspace{0.5cm}\footnotesize \hyperlink{restated#1}{\hypertarget{original#1}{[Proof]}}}}
\newcommand{\restatehere}[1]{%
	\marginline{\vspace{0.6cm}\footnotesize \hyperlink{original#1}{\hypertarget{restated#1}{[Main]}}}%
	\csname #1\endcsname*%
}

\usepackage{tikz}
\usetikzlibrary{decorations.pathreplacing}

\makeatletter
\def\moverlay{\mathpalette\mov@rlay}
\def\mov@rlay#1#2{\leavevmode\vtop{%
		\baselineskip\z@skip \lineskiplimit-\maxdimen
		\ialign{\hfil$\m@th#1##$\hfil\cr#2\crcr}}}
\newcommand{\charfusion}[3][\mathord]{
	#1{\ifx#1\mathop\vphantom{#2}\fi
		\mathpalette\mov@rlay{#2\cr#3}
	}
	\ifx#1\mathop\expandafter\displaylimits\fi}
\makeatother

\newcommand{\cupdot}{\charfusion[\mathbin]{\cup}{\cdot}}

\usepackage{mathtools}

\usepackage[backgroundcolor=orange!15!white,bordercolor=orange!80!black,textsize=footnotesize]{todonotes} %
\tikzset{/tikz/notestyleraw/.append style={rounded corners=0pt,inner sep=0.6ex}}
\presetkeys{todonotes}{inline}{}
\usepackage{mathtools}

\DeclarePairedDelimiter\floor{\lfloor}{\rfloor}

\usepackage{xspace}
\usepackage{siunitx}

\newcommand{\seqWi}[1]{\text{Sequential-#1-Winner}\xspace}
\newcommand{\seqLo}[1]{\text{Sequential-#1-Loser}\xspace}
\newcommand{\seqWiAb}[1]{\text{Seq.-#1-Wi.}\xspace}
\newcommand{\seqLoAb}[1]{\text{Seq.-#1-Lo.}\xspace}
\newcommand{\Score}[1]{\text{#1-Score}\xspace}

\newcommand{\seqWiAbS}{\text{Seq.-Wi.}\xspace}
\newcommand{\seqLoAbS}{\text{Seq.-Lo.}\xspace}
\newcommand{\ScoreS}{\text{Score}\xspace}

\renewcommand{\seqWiAb}[1]{\text{Seq.-#1-Winner}\xspace}
\renewcommand{\seqLoAb}[1]{\text{Seq.-#1-Loser}\xspace}
\renewcommand{\seqWiAbS}{\text{Seq.-Winner}\xspace}
\renewcommand{\seqLoAbS}{\text{Seq.-Loser}\xspace}

\newcommand{\TopK}{\textsc{Top-$k$ Determination}\xspace}
\newcommand{\PositionK}{\textsc{Position-$k$ Determination}\xspace}
\newcommand{\Winner}{\textsc{Winner Determination}\xspace}

\newcommand{\vs}{\ensuremath{\mathbf{s}}}
\newcommand{\KT}{\kappa}

\newtheorem{theorem}{Theorem}[section]
\newtheorem{claim}[theorem]{Claim}
\newtheorem{example}[theorem]{Example}
\newtheorem{observation}[theorem]{Observation}
\newtheorem{proposition}[theorem]{Proposition}
\newtheorem{lemma}[theorem]{Lemma}
\newtheorem{corollary}[theorem]{Corollary}
\theoremstyle{definition}
\newtheorem{definition}[theorem]{Definition}
\newtheorem{remark}[theorem]{Remark}

\newcommand{\normphi}{{{\mathrm{norm}\hbox{-}\phi}}}

\newcommand{\rank}{\mathrm{pos}}
\newcommand{\score}{\mathrm{score}}
\newcommand{\rev}{\mathrm{rev}}
\newcommand{\cand}{\mathrm{cand}}

\newenvironment{claimproof}[1][\proofname]
{%
	\proof[#1]%
}
{%
	\endproof%
}

\title{Rank Aggregation Using Scoring Rules}
\date{\vspace{-1.5cm}}

\begin{document}

\maketitle

{\footnotesize\tableofcontents}

\begin{abstract}
	\begin{center}
		\textbf{\textsf{Abstract}} \smallskip
	\end{center}
	To aggregate rankings into a social ranking, one can use scoring systems such as Plurality, Veto, and Borda. 
	We distinguish three types of methods: 
	ranking by score, ranking by repeatedly choosing a winner that we delete and rank at the top, and ranking by repeatedly choosing a loser that we delete and rank at the bottom.
	The latter method captures the frequently studied voting rules Single Transferable Vote (aka Instant Runoff Voting), Coombs, and Baldwin.   	
	In an experimental analysis, we show that the three types of methods produce different rankings in practice. We also provide evidence that sequentially selecting winners is most suitable to detect the ``true'' ranking of candidates.
	For different rules in our classes, we then study the (parameterized) computational complexity of deciding in which positions a given candidate can appear in the chosen ranking. 
	As part of our analysis, we also consider the \textsc{Winner Determination} problem for STV, Coombs, and Baldwin and determine their complexity when there are few voters or candidates. 
\end{abstract}

\section{Introduction}

\emph{Rank aggregation}, the task of aggregating several rankings into a single ranking, sits at the foundation of social choice as introduced by \citet{Arro51a}. Besides preference aggregation, it has numerous important applications, for example in the context of meta-search engines \citep{DKNS01a}, of juries ranking competitors in sports tournaments \citep{Truchon98}, and multi-criteria decision analysis.

One of the best-known methods for aggregating rankings is \citeauthor{Keme59a}'s \citeyearpar{Keme59a} method:
A Kemeny ranking is a ranking that minimizes the average swap distance (Kendall-tau distance) to the input rankings. It is axiomatically attractive \citep{YoLe78a,CaSt13a,bossert2014strategy} and has an interpretation as a maximum likelihood estimator \citep{Youn95a} making it well-suited to epistemic social choice that assumes a ground~truth.

However, Kemeny's method is hard to compute \citep{bartholdi1989voting,hemaspaandra2005complexity} which makes the method problematic to use, especially when there are many candidates to rank (for example, when ranking all applicants to a university). Even if computing the ranking is possible, it is coNP-hard to verify if a ranking is indeed a Kemeny ranking \citep{FiHe22}. Thus, third parties cannot easily audit, interpret, or understand the outcome, making systems based on Kemeny's method potentially unaccountable. This limits its applicability in democratic contexts.

These two drawbacks motivate the search for computationally simpler and more transparent methods for aggregating rankings. There is a significant literature on polynomial-time approximation algorithms for Kemeny's method \citep{coppersmith2006ordering,kenyon2007rank,ailon2008aggregating,van2009deterministic}, but these algorithms are typically not attractive beyond their approximation guarantee. In particular, they would typically not fare well in an axiomatic analysis, and are unlikely to be understood by and appealing to the general public (many are based on derandomization).

Instead, we turn to one of the fundamental tools of social choice: positional scoring rules. These rules transform voter rankings into scores for the candidates. For example, under the \emph{Plurality} scoring rule, every voter gives $1$ point to their top-ranked candidate. Under the \emph{Veto} (or anti-plurality) scoring rule, voters give $-1$ point to their last-ranked candidate and zero points to all others. Under the \emph{Borda} scoring rule, every voter gives $m$ points to their top-ranked candidate, $m-1$ points to their second-ranked candidate, and so on, giving 1 point to their last-ranked candidate. 
We study three ways of using scoring rules to aggregate rankings:
\begin{itemize}
	\item \emph{Score}: We rank the candidates in order of their score, higher-scoring candidates being ranked higher.
	\item \emph{Sequential-Winner}: We take the candidate $c$ with the highest score and rank it top in the aggregate ranking. We then delete $c$ from the input profile, re-calculate the scores, and put the new candidate with the highest score in the second position, and so on.
	\item \emph{Sequential-Loser:} We take the candidate $c$ with the lowest score and rank it last. We then delete $c$, re-calculate the scores, and put the new candidate with the lowest score in the second-to-last position, and so on.
\end{itemize}
Ranking by score is the obvious way of using scoring rules for rankings, and so it has been studied in the social choice literature \citep{Smit73a,LevenglickThesis}. Sequential-Loser captures as special cases the previously studied rules Single Transferable Vote (also known as Instant Runoff Voting, among other names), Coomb's method, and Baldwin's method. These are typically used as voting rules that elect a single candidate, but they can also be understood as rank aggregation methods. On the other hand, despite being quite natural, Sequential-Winner methods appear not to have been formally studied in the literature (to our knowledge).

\subsection{Our Contributions}

\paragraph{Axiomatic Properties (\Cref{sec:axioms})}
Based on the existing literature, we begin by describing some axiomatic properties of the methods in our three families. For example, we check which of the methods are Condorcet or majority consistent, and which are resistant to cloning. We also consider independence properties and state some characterization results.

\paragraph{Simulations (\Cref{sec:simulations})}
To understand how and whether the three families of methods practically differ from each other, and how they relate to Kemeny's method, we perform extensive simulations based on synthetic data (sampled using the Mallows and Euclidean models).
We find that, for Plurality and Borda, ranking by score and Sequential-Loser usually produce very similar results, whereas Sequential-Winner offers a new perspective (which is typically closer to Kemeny's method).
Moreover, we observe that Sequential-Loser rules seem to be particularly well suited to identify the best candidates (justifying their usage as single-winner voting rules), while Sequential-Winner rules are best at avoiding low quality candidates.

\paragraph{Computational Complexity (\Cref{sec:complexity})}
The rules in all three of our families are easy to compute in the sense that their description implies a straightforward algorithm for obtaining an output ranking. However, for the sequential rules there is a subtlety: During the execution of the rule, ties can occur. It matters how these are broken, because candidates could end up in significantly different positions. For high-stakes decisions and in democratic contexts, it would be important to know which output rankings are possible. 

Thus, we study the computational problem of deciding whether a given candidate can end up in a given position. This and related problems have been studied in the literature under the name of \emph{parallel universe tie-breaking}, including theoretical and experimental studies for some of the rules in our families \citep{DBLP:conf/ijcai/ConitzerRX09,brill2012price,mattei2014hard,freeman2015general,wang2019practical}.
We extend the results of that literature and find NP-hardness for all the sequential methods that we study. We show that the problem becomes tractable if the number of candidates is small. In contrast, for several methods we find that the problem remains hard even if the number of input rankings is small. Curiously, for few input rankings, methods based on Plurality, Borda, or Veto each induce a different parameterized complexity class.

\section{Preliminaries}

For $k \in \mathbb{N}$, write $[k] = \{1, \dots, k\}$.

Let $C = \{c_1, \dots, c_m\}$ be a set of $m$ \emph{candidates}. A \emph{ranking} $\succ$ of $C$ is a linear order (irreflexive, total, transitive) of $C$. We write $\mathcal{L}(C)$ for the set of all rankings of $C$.

A \emph{(ranking) profile} $P = (\succ_1, \dots, \succ_n)$ is a list of rankings. We sometimes say that the rankings are \emph{voters}.

For a subset $C' \subseteq C$ of candidates  and ranking ${\succ} \in \mathcal{L}(C)$, we write ${\succ}|_{C'}$ for the ranking obtained by restricting $\succ$ to the set $C'$. For a profile $P$, we write $P|_{C'}$ for the profile obtained by restricting each of its rankings to~$C'$.

A \emph{social preference function}%
\footnote{This terminology is due to \citet{YoLe78a}. The term \emph{social welfare function} from \citet{Arro51a} usually refers to resolute functions that may only output a single ranking.} 
$f$ is a function that assigns to every ranking profile $P$ a non-empty set $f(P) \subseteq \mathcal{L}(C)$ of rankings. Here, $f(P)$ may be a singleton but there can be more than one output ranking in case of ties.
For a ranking~$\succ$, we say that \emph{$f$ selects $\succ$ on $P$} if ${\succ} \in f(P)$.

For a ranking ${\succ} \in \mathcal{L}(C)$ and a candidate $c\in C$, let $\rank(\succ,c) = |\{d \in C : d \succ c \}| + 1$ be the \emph{position}  of $c$ in~$\succ$. For example, if $\rank(\succ,c) = 1$ then $c$ is the most-preferred candidate in~$\succ$.
We write $\cand(\succ,r) \in C$ for the candidate ranked in position $r \in [m]$ in ${\succ} \in \mathcal{L}(C)$.

For a ranking ${\succ} \in \mathcal{L}(C)$, $\rev(\succ)$ denotes the ranking where the candidates are ranked in the opposite order as in $\succ$, i.e., for each $r\in [m]$, $\cand(\succ,r)=\cand(\rev(\succ),m-r+1)$. 
For a profile $P=(\succ_1,\dots, \succ_n)$, we write $\rev(P)=(\rev(\succ_1),\dots ,\rev(\succ_n))$.

For an integer $m\in \mathbb{N}$, a \emph{scoring vector} $\vs^{(m)} = (s_1, \dots, s_m) \in \mathbb{R}^m$ is a list of $m$ numbers.
A \emph{scoring system} is a family of scoring vectors $(\vs^{(m)})_{m\in \mathbb{N}}$ one for each possible number $m$ of candidates.
For the sake of conciseness, we sometimes write $\vs$ instead of $(\vs^{(m)})_{m\in \mathbb{N}}$.
We will mainly focus on three scoring systems: 
\begin{itemize}[itemsep=0ex]
	\item \emph{Plurality} with $\vs^{(m)}=(1,0,\dots,0)$ for each $m\in \mathbb{N}$,
	\item \emph{Veto} with $\vs^{(m)}=(0,\dots,0,-1)$ for each $m\in \mathbb{N}$,
	\item \emph{Borda} with $\vs^{(m)}=(m, m-1, \dots, 1)$ for each $m\in \mathbb{N}$.
\end{itemize}
Given a profile $P$ over $m$ candidates, the \emph{$\vs$-score} of candidate $c \in C$ is $\score_{\vs}(P, c) = \sum_{i \in [n]} \vs^{(m)}_{\rank(\succ_i, c)}$. 
We say that a candidate is an \emph{$\vs$-winner} if it has maximum $\vs$-score, and an \emph{$\vs$-loser} if it has minimum $\vs$-score.
For a scoring system 
$\vs$
we denote by 
$\vs^*$
the scoring system where we reverse each scoring vector and multiply all its entries by $-1$, i.e., for each $m\in \mathbb{N}$ and $i\in [m]$, we have $(\vs^*)^{(m)}_i=-\vs^{(m)}_{m-i+1}$. Note that $(\vs^*)^* = \vs$ for every $\vs$, that $\text{Plurality}^* = \text{Veto}$, that $\text{Veto}^* = \text{Plurality}$, and that $\text{Borda}^*$ is the same as Borda, up to a shift.

For two rankings $\succ_1$ and $\succ_2$, their \emph{swap distance} (or \emph{Kendall-tau distance}) $\KT(\succ_1,\succ_2)$ is  the number of pairs of candidates on whose ordering the two rankings disagree, i.e., $\KT(\succ_1,\succ_2) = |\{ (c,d) \in C \times C : c \succ_1 d \text{ and } d \succ_2 c \}|$. Note that the maximum swap distance between two rankings is $\binom{m}{2}$.
Given a profile $P$, \emph{Kemeny's rule} selects those rankings which minimize the average swap distance to the rankings in $P$, so it selects $\arg\min_{{\succ} \in \mathcal{L}(C)} \sum_{i \in N} \KT(\succ, \succ_i)$. We refer to the selected rankings as \emph{Kemeny rankings}.

\section{Scoring-Based Rank Aggregation}
We now formally define the three families of scoring-based social preference functions that we study. 

\begin{definition}[\Score{\vs}]
	Let $\mathbf s$ be a scoring system. For the social preference function \Score{\vs} on profile $P$, we have ${\succ} \in \Score{\vs}(P)$ if and only if for all $c,d\in C$ with $\score_\vs(c,P) > \score_\vs(d,P)$, we have $c \succ d$.
\end{definition}

\begin{definition}[\seqWi{\vs}; \seqWiAb{\vs}]
	Let $\mathbf s$ be a scoring system. The social preference function \seqWiAb{\vs} is defined recursively as follows: For a profile $P$, we have ${\succ} \in \seqWiAb{\vs}(P)$ if and only if 
	\begin{itemize}
		\item the top choice $c = \cand({\succ}, 1)$ is an $\vs$-winner in $P$,
		\item if $|C| > 1$, then ${\succ}|_{C \setminus \{c\}} \in \seqWiAb{\vs}(P|_{C \setminus \{c\}})$.
	\end{itemize}
\end{definition}

\begin{definition}[\seqLo{\vs};  \seqLoAb{\vs}]
	Let $\mathbf s$ be a scoring system. The social preference function \seqLoAb{\vs} is defined recursively as follows: For a profile $P$, we have ${\succ} \in \seqLoAb{\vs}(P)$ if and only if 
	\begin{itemize}
		\item the bottom choice $c = \cand({\succ}, |C|)$ is an $\vs$-loser in $P$,
		\item if $|C| > 1$, then ${\succ}|_{C \setminus \{c\}} \in \seqLoAb{\vs}(P|_{C \setminus \{c\}})$.
	\end{itemize}
\end{definition}

\begin{example}
	\setlength{\abovedisplayskip}{3pt}
	\setlength{\belowdisplayskip}{3pt}
	Let $P$ be the following ranking profile:
	\begin{align*}
		3\times a\succ b \succ c, \quad 2\times b\succ c \succ a,   \quad 2\times c\succ b \succ a
	\end{align*}
	Then for the three methods based on Plurality, we have: 
	\begin{itemize}[itemsep=0ex]
		\item $\text{Plurality-Score}(P)=\{a\succ b \succ c,\:\, a\succ c \succ b\}$,
		\item $\text{Seq.-Plurality-Wi.}(P)=\{a\succ b \succ c\}$, and 
		\item $\text{Seq.-Plurality-Lo.}(P)=\{b\succ a \succ c, \:\, c\succ a \succ b\}$.
	\end{itemize}
\end{example}
We sometimes view \seqWiAb{\vs} (or \seqLoAb{\vs}) rules as round-based voting rules, where in each round an $\vs$-winner (or an $\vs$-loser) is deleted from the profile and added in the highest (or lowest) position of the ranking that has not yet been filled. 
Notably, if there are multiple $\vs$-winners (or $\vs$-losers) in one round, each selection gives rise to different output rankings.  
\seqLoAb{Plurality} is also known as \emph{STV}, \seqLoAb{Veto} as \emph{Coombs}, and \seqLoAb{Borda} as \emph{Baldwin}.

Sequential-Winner and Sequential-Loser rules are formally closely related:
If a candidate is an $\mathbf{s}$-winner in some profile $P$, then it is an $\mathbf{s}^*$-loser in the reverse profile $\rev(P)$.
Hence, we can conclude the following:
\begin{lemma} \label{lem:equiv}
	Let $\vs$ be a scoring system. Then for each ranking profile $P$ and for every ranking ${\succ} \in \mathcal{L}(C)$, we have: 
	\begin{align*}
	{\succ} &\in  \text{Sequential-}\vs\text{-Winner}(P)\\
	\iff  \rev({\succ}\hspace{-0.7pt}) &\in \text{Sequential-}\vs^*\text{-Loser}(\rev(P)).
\end{align*}
\end{lemma}
For example, this lemma establishes a close connection between \seqWiAb{Veto} and \seqLoAb{Plurality}, as a ranking $\succ$ is selected under \seqWiAb{Veto} on profile $P$ if and only if $\rev(\succ)$ is selected under \seqLoAb{Plurality} on profile $\rev(P)$. This equivalence will prove useful in our axiomatic analysis and in our complexity results.

\section{Axiomatic Properties}\label{sec:axioms}

\begin{table*}[t]
	\centering
	\newcommand{\sat}{\textcolor{green!40!black}{\checkmark}}
	\newcommand{\fails}{%
		\begin{tikzpicture}[scale=0.20]
			\draw[line width=0.75,line cap=round] (0,0) to [bend left=6] (1,1);
			\draw[line width=0.75,line cap=round] (0.2,0.95) to [bend right=3] (0.8,0.05);
	\end{tikzpicture}}
	\resizebox{\textwidth}{!}{%
		\begin{tabular}{lcccccccccc}
			\toprule
			& & \multicolumn{3}{c}{Score} & \multicolumn{3}{c}{Sequential-Winner} & \multicolumn{3}{c}{Sequential-Loser} \\
			\cmidrule(l){3-5}
			\cmidrule(l){6-8}
			\cmidrule(l){9-11}
			& Kemeny    & Plurality & Veto     & Borda      & Plurality & Veto      & Borda     & Plurality & Veto      & Borda     \\
			\midrule
			Independence at the top     & \sat      & \fails    & \fails   & \fails     & \sat      & \sat      & \sat      & \fails    & \fails    & \fails    \\
			Independence at the bottom  & \sat      & \fails    & \fails   & \fails     & \fails    & \fails    & \fails    & \sat      & \sat      & \sat      \\
			\midrule
			Reinforcement               & \sat      & \sat      & \sat     & \sat       & \sat      & \sat      & \sat      & \sat      & \sat      & \sat      \\
			Reinforcement at the top    & \fails    & \sat      & \sat     & \sat       & \sat      & \sat      & \sat      & \fails    & \fails    & \fails    \\
			Reinforcement at the bottom & \fails    & \sat      & \sat     & \sat       & \fails    & \fails    & \fails    & \sat      & \sat      & \sat      \\
			\midrule
			Condorcet winner at top     & \sat      & \fails    & \fails   & \fails     & \fails    & \fails    & \fails    & \fails    & \fails    & \sat      \\
			Copy majority               & \sat      & \fails    & \fails   & \fails     & \sat      & \fails    & \fails    & \fails    & \sat      & \fails    \\
			Independence of clones      & \fails    & \fails    & \fails   & \fails     & \fails    & \fails     & \fails    & \sat      & \fails    & \fails    \\
			\bottomrule
		\end{tabular}
	} %
	\caption{An overview of the axiomatic properties of our studied rules. See \Cref{app:axioms} for definitions.}
	\label{tbl:axioms}
\end{table*}

In this section, we will briefly and informally discuss some axiomatic properties and characterizations of the methods in our three families. A more formal treatment appears in \Cref{app:axioms}. See \Cref{tbl:axioms} for an overview.

A desirable property of a ranking aggregation rule is that if one candidate is deleted from the profile, then the relative rankings of the other candidates does not change (\emph{independence of irrelevant alternatives}, IIA). \citeauthor{Arro51a}'s \citeyearpar{Arro51a} impossibility theorem shows that this property cannot be satisfied by unanimous non-dictatorial rules. \citet{Youn88a} proves that Kemeny's method satisfies a weaker version that he calls \emph{local IIA}: removing the candidate that appears in the first or last position in the Kemeny ranking  does not change the ranking of the other candidates. Splitting this property into its two parts, we can easily see from their definitions that \seqWiAb{\vs} satisfies independence at the top, and \seqLoAb{\vs} satisfies independence at the bottom.

Another influential axiom is known as \emph{consistency} or \emph{reinforcement}. A rule $f$ satisfies reinforcement if whenever some ranking $\succ$ is chosen in two profiles, ${\succ} \in f(P) \cap f(P')$, then it is also chosen if we combine the profiles into one, and in fact $f(P + P') = f(P) \cap f(P')$. All the methods in this paper satisfy reinforcement. Notably, \citet{Youn88a} shows that Kemeny is the only anonymous, neutral, and unanimous rule satisfying reinforcement and local IIA. Focusing on \seqLoAb{\vs}, \citet{FBC14a} define reinforcement at the bottom to mean that if the same candidate $c$ is placed in the last position in the selected ranking in two profiles, then $c$ is also placed in the last position in the selected ranking in the combined profile. They show that independence at the bottom and reinforcement at the bottom characterize \seqLoAb{\vs} rules (under mild additional assumption). Using \Cref{lem:equiv}, a simple adaptation of their proof shows that \seqWiAb{\vs} rules can be similarly characterized by independence at the top and reinforcement at the top. (\Score{\vs} methods do not satisfy similar independence assumptions; they have been characterized by \citet{LevenglickThesis} and \citet{Smit73a}.)

Refining their characterization of \seqLoAb{\vs} rules, \citet{FBC14a} characterize \seqLoAb{Plurality} (aka STV) as the only \seqLoAb{\vs} rule satisfying independence of clones \citep{Tide87a}, \seqLoAb{Veto} (aka Coombs) as the only one that, in case a strict majority of voters have the same ranking, copies that ranking as the output ranking, and \seqLoAb{Borda} (aka Baldwin) as the only one always placing a Condorcet winner in the first position. Using \Cref{lem:equiv}, we can similarly characterize \seqWiAb{Plurality} as the only method in its class that copies a majority ranking.

\section{Simulations}\label{sec:simulations}

We analyze our three families of scoring-based ranking rules for Plurality and Borda on synthetically generated~profiles.

\subsection{Setup}

To deal with ties in the computation of our rules, each time we sample a ranking profile over candidates $C$, we also sample a ranking ${\succ_{\text{tie}}} \in \mathcal{L}(C)$ uniformly at random and break ties according to $\succ_{\text{tie}}$ for all rules. 
To quantify the difference between two rankings ${\succ_1},{\succ_2}\in \mathcal{L}(C)$, we use their normalized swap distance, i.e., their swap distance  $\KT(\succ_1, \succ_2)$ divided by the maximum possible swap distance between two rankings~$\smash{\binom{m}{2}}$.

\paragraph{(Normalized) Mallows} We conduct simulations on profiles generated using the Mallows model \cite{mallows1957non}  (as observed by \citet{DBLP:conf/ijcai/BoehmerBFNS21} real-world profiles seem often to be close to some Mallows profiles). 
This model is parameterized by a dispersion parameter $\phi\in [0,1]$ and a central ranking ${\succ^*}\in \mathcal{L}(C)$.
Then, a profile is assembled by sampling rankings i.i.d. so that the probability of sampling a ranking ${\succ} \in \mathcal{L}(C)$ is proportional to $\phi^{\KT(\succ,\succ^*)}$. 
We use the normalization of the Mallows model proposed by \citet{DBLP:conf/ijcai/BoehmerBFNS21}, which is parameterized by a normalized dispersion parameter $\normphi\in [0,1]$. This parameter is then internally converted to a dispersion parameter $\phi$ such that the expected swap distance between a sampled vote and the central vote is $\normphi\cdot (m(m-1)/4)$. 
Then $\normphi=0$ results in profiles only containing the central vote, and $\normphi=1$ leads to profiles where all rankings are sampled with the same probability, so that on average rankings disagree with the central ranking $\succ^*$ on half of the pairwise comparisons. Choosing $\normphi=0.5$ leads to profiles where rankings on average disagree with $\succ^*$ on a quarter of the pairwise comparisons.

\subsection{Comparison of Scoring-Based Ranking Methods}\label{sub:sim-scorecomp}

\begin{figure*}[t!]  
	\centering                  
	\begin{subfigure}[t]{0.485\textwidth} 
		\centering
		\resizebox{0.9\textwidth}{!}{
\begin{tikzpicture}[every plot/.append style={line width=3.5pt}]

\definecolor{color0}{rgb}{1,0.549019607843137,0}
\definecolor{color1}{rgb}{0.133333333333333,0.545098039215686,0.133333333333333}

\node[text width=5cm] at (10,1) 
{\LARGE dashed green and orange line overlap};

\begin{axis}[
ylabel shift = 1pt,
legend columns=1, 
legend cell align={left},
legend style={
  fill opacity=0.8,
  draw opacity=1,
  draw=none,
  text opacity=1,
  at={(1.45,1)},
  line width=3pt,
  anchor=north,
   /tikz/column 2/.style={
  	column sep=10pt,
  },
  font=\LARGE
},
every tick label/.append style={font=\LARGE}, 
label style={font=\LARGE},
tick align=outside,
tick pos=left,
x grid style={white!69.0196078431373!black},
xlabel={normalized dispersion parameter},
xmin=-0.05, xmax=1.05,
xtick style={color=black},
xtick={-0.2,0,0.2,0.4,0.6,0.8,1,1.2},
xticklabels={−0.2,0.0,0.2,0.4,0.6,0.8,1.0,1.2},
y grid style={white!69.0196078431373!black},
ylabel={normalized swap distance},
ymin=-0.0127055555555556, ymax=0.45,
ytick style={color=black},
ytick={0,0.1,0.2,0.3,0.4},
yticklabels={0,0.1,0.2,0.3,0.4}
]
\addlegendentry{Seq.-Lo. vs Seq.-Wi.}
\addlegendentry{Seq.-Lo. vs Score}
\addlegendentry{Seq.-Wi. vs Score}
\addlegendentry{}
\addlegendentry{}
\addlegendentry{Plurality}
\addlegendentry{Borda}
\addlegendimage{red!54.5098039215686!black}
\addlegendimage{color0}
\addlegendimage{color1}
\addlegendimage{empty legend}
\addlegendimage{empty legend}
\addlegendimage{gray}
\addlegendimage{gray,dashed}
\addplot [semithick, red!54.5098039215686!black]
table {%
0 0.399337777777778
0.1 0.196542222222222
0.2 0.127155555555556
0.3 0.0898022222222222
0.4 0.0732422222222222
0.5 0.0728533333333333
0.6 0.0809555555555556
0.7 0.100706666666667
0.8 0.137706666666667
0.9 0.203808888888889
1 0.254948888888889
};
\addplot [semithick, color0]
table {%
0 0
0.1 1.11111111111111e-05
0.2 0.000217777777777778
0.3 0.00126222222222222
0.4 0.00412888888888889
0.5 0.00987555555555556
0.6 0.0210177777777778
0.7 0.0386311111111111
0.8 0.0687666666666667
0.9 0.115688888888889
1 0.148086666666667
};
\addplot [semithick, color1]
table {%
0 0.399337777777778
0.1 0.196544444444444
0.2 0.127226666666667
0.3 0.0901222222222222
0.4 0.0739
0.5 0.0734
0.6 0.0793777777777778
0.7 0.0946177777777778
0.8 0.122895555555556
0.9 0.169933333333333
1 0.20064
};

\addplot [semithick, red!54.5098039215686!black, dashed]
table {%
0 0
0.1 0
0.2 2.22222222222222e-06
0.3 0.000228888888888889
0.4 0.00208444444444444
0.5 0.00673333333333333
0.6 0.0141844444444444
0.7 0.0247
0.8 0.0430666666666667
0.9 0.0816022222222222
1 0.132086666666667
};
\addplot [semithick, color0, dashed]
table {%
0 0
0.1 0
0.2 4.44444444444444e-06
0.3 0.000131111111111111
0.4 0.00125333333333333
0.5 0.00389333333333333
0.6 0.00778
0.7 0.0135022222222222
0.8 0.0233955555555556
0.9 0.0448422222222222
1 0.0740311111111111
};
\addplot [semithick, color1,dashed]
table {%
0 0
0.1 0
0.2 6.66666666666667e-06
0.3 0.000155555555555556
0.4 0.00116444444444444
0.5 0.00356444444444444
0.6 0.00770222222222222
0.7 0.0131888888888889
0.8 0.0233244444444444
0.9 0.0446622222222222
1 0.07486
};
\end{axis}

\end{tikzpicture}}
		\caption{Pairs of scoring-based methods}\label{fig:general-comparsion}
	\end{subfigure}%
	\hfill
	\begin{subfigure}[t]{0.475\textwidth}
		\centering 
		\resizebox{0.9\textwidth}{!}{
\begin{tikzpicture}[every plot/.append style={line width=3.5pt}]

\definecolor{color0}{rgb}{1,0.549019607843137,0}
\definecolor{color1}{rgb}{0.133333333333333,0.545098039215686,0.133333333333333}
\definecolor{color2}{rgb}{0.117647058823529,0.564705882352941,1}
\definecolor{color3}{rgb}{0.580392156862745,0,0.827450980392157}
\definecolor{color4}{rgb}{0.647058823529412,0.164705882352941,0.164705882352941}

\node[text width=5cm] at (10,1) 
{\LARGE dashed lines largely overlap};

\begin{axis}[
legend columns=1, 
legend cell align={left},
legend style={
  fill opacity=0.8,
  draw opacity=1,
  draw=none,
  text opacity=1,
  at={(1.45,1)},
  line width=3pt,
  anchor=north,
   /tikz/column 2/.style={
  	column sep=10pt,
  },
  font=\LARGE
},
every tick label/.append style={font=\LARGE}, 
label style={font=\LARGE},
tick align=outside,
tick pos=left,
x grid style={white!69.0196078431373!black},
xlabel={normalized dispersion parameter},
xmin=-0.05, xmax=1.05,
xtick style={color=black},
xtick={-0.2,0,0.2,0.4,0.6,0.8,1,1.2},
xticklabels={−0.2,0.0,0.2,0.4,0.6,0.8,1.0,1.2},
y grid style={white!69.0196078431373!black},
ylabel={normalized swap distance},
ymin=-0.0127055555555556, ymax=0.45,
ytick style={color=black},
ytick={0,0.1,0.2,0.3,0.4},
yticklabels={0,0.1,0.2,0.3,0.4}
]
\addlegendentry{Kemeny vs Seq.-Lo.}
\addlegendentry{Kemeny vs. Score}
\addlegendentry{Kemeny vs. Seq.-Wi.}
\addlegendentry{}
\addlegendentry{}
\addlegendentry{Plurality}
\addlegendentry{Borda}

\addlegendimage{color2}
\addlegendimage{color3}
\addlegendimage{color4}
\addlegendimage{empty legend}
\addlegendimage{empty legend}
\addlegendimage{gray}
\addlegendimage{gray,dashed}

\addplot [semithick, color2]
table {%
0 0.399337777777778
0.1 0.196542222222222
0.2 0.127155555555556
0.3 0.0897933333333333
0.4 0.0730955555555556
0.5 0.0722177777777778
0.6 0.0787844444444444
0.7 0.0961088888888889
0.8 0.131748888888889
0.9 0.213624444444444
1 0.29502
};
\addplot [semithick, color3]
table {%
0 0.399337777777778
0.1 0.196544444444444
0.2 0.127226666666667
0.3 0.0901222222222222
0.4 0.0744066666666667
0.5 0.0756444444444444
0.6 0.0865844444444444
0.7 0.11114
0.8 0.159875555555556
0.9 0.254042222222222
1 0.337413333333333
};
\addplot [semithick, color4]
table {%
0 0
0.1 0
0.2 0
0.3 9.77777777777778e-05
0.4 0.00164888888888889
0.5 0.00731111111111111
0.6 0.0203177777777778
0.7 0.0428111111111111
0.8 0.0868422222222222
0.9 0.177104444444444
1 0.268853333333333
};

\addplot [semithick, color2, dashed]
table {%
0 0
0.1 0
0.2 0
0.3 0.00014
0.4 0.00128666666666667
0.5 0.00476666666666667
0.6 0.0110111111111111
0.7 0.0198555555555556
0.8 0.0357266666666667
0.9 0.0645644444444444
1 0.101431111111111
};
\addplot [semithick, color3, dashed]
table {%
0 0
0.1 0
0.2 4.44444444444444e-06
0.3 0.000186666666666667
0.4 0.00177555555555556
0.5 0.00641111111111111
0.6 0.0137022222222222
0.7 0.0245044444444444
0.8 0.0424688888888889
0.9 0.0731933333333333
1 0.107742222222222
};
\addplot [semithick, color4, dashed]
table {%
0 0
0.1 0
0.2 2.22222222222222e-06
0.3 0.000115555555555556
0.4 0.00132222222222222
0.5 0.00498888888888889
0.6 0.0107377777777778
0.7 0.0200577777777778
0.8 0.0359622222222222
0.9 0.0653933333333333
1 0.101522222222222
};
\end{axis}

\end{tikzpicture}}
		\caption{Scoring-based methods vs Kemeny ranking}\label{fig:Kemeny-comparison}
	\end{subfigure}
	\caption{Pairwise average normalized swap distance between rankings produced by different methods for Plurality (solid) and Borda (dashed) on Mallows profiles with $10$ candidates and $100$ voters.}\label{fig:Mal}
\end{figure*}
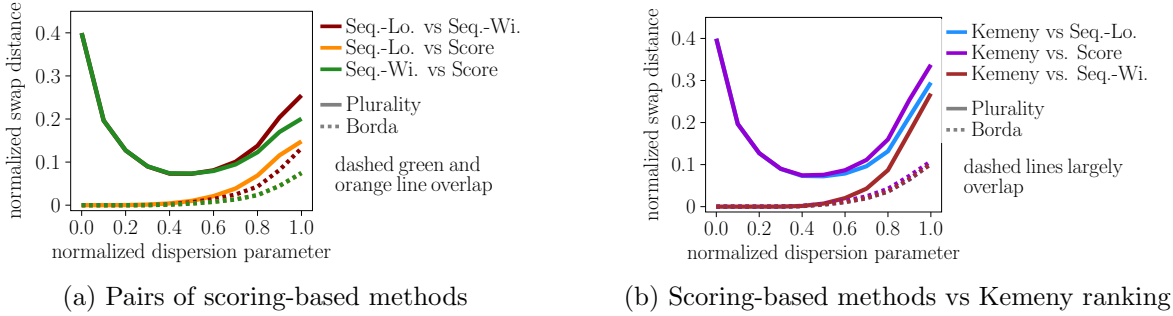

We analyze the average normalized swap distance between the rankings selected by our three families of scoring-based ranking methods on profiles containing $100$ rankings over $10$ candidates.
For this, we sampled $\num{10000}$ profiles for each $\normphi\in \{0,0.1,\dots, 0.9,1\}$ and depict the results in \Cref{fig:general-comparsion}. 
Let us first focus on Plurality: 
We find that the rankings produced by \seqLoAb{Plurality} and \Score{Plurality} are quite similar, whereas the ranking produced by \seqWiAb{Plurality} is substantially different. 
This observation is particularly strong for $\normphi\leq 0.3$: 
In such profiles, all the rankings are similar to each other. 
Accordingly, many candidates initially have a Plurality score of zero, and thus there are many ties in the execution of \Score{Plurality} and \seqLoAb{Plurality} (for the latter, ties occur in more than half of the rounds).
Thus, the rankings computed by the two rules fundamentally depend on the (shared) random tie-breaking order $\succ_{\text{tie}}$.
In contrast, for \seqWiAb{Plurality}, for $\normphi\le 0.3$, no ties in its execution appear. 
In particular, \seqWiAb{Plurality} is thereby able to meaningfully distinguish the weaker candidates on these profiles.

Turning to $\normphi\geq0.3$ (where more candidates have non-zero Plurality score and thus the tie-breaking is no longer as important), \seqLoAb{Plurality} and \Score{Plurality} are still clearly more similar to each other than to \seqWiAb{Plurality}; this indicates that Seq.-Wi. rules indeed add a new perspective to existing scoring-based ranking rules. 

Switching to Borda, the rankings returned by the three methods are quite similar. 
This is intuitive given that Borda scores capture the general strength of candidates in a profile much better than Plurality scores. 
Thus, the Borda score of a candidate also changes less drastically in case some candidate is deleted. 
Increasing $\normphi$, the selected rankings become more different from each other (as profiles get more chaotic, leading to more similar Borda scores of candidates).
Interestingly, for larger values of $\normphi$, \Score{Borda} has the same (small) distance to the other two rules, whereas \seqWiAb{Borda} and \seqLoAb{Borda} are more different.

\subsection{Comparison to Kemeny Ranking}\label{sub:Kemeny-scorecomp}

To assess which method produces the  ``most accurate'' rankings, we compare them to Kemeny's method. 
For \num{10000} profiles for each $\normphi\in \{0,0.1,\dots, 0.9,1\}$, in \Cref{fig:Kemeny-comparison}, we show the average normalized swap distance of the Kemeny ranking to the rankings selected by our rules. 

For Plurality, independently of the value of $\normphi$, \seqWiAb{Plurality} produces the ranking most similar to the Kemeny ranking, then  \seqLoAb{Plurality} and lastly \Score{Plurality}, indicating the advantages of sequential rules. 
What sticks out is that for $\normphi\leq 0.3$, \seqLoAb{Plurality} and  \Score{Plurality} are far away from the Kemeny ranking. 
As discussed above, the reason is that, for both methods, large parts of the ranking are simply determined by the random tie-breaking order in such profiles.  
In contrast, \seqWiAb{Plurality} is not affected, and its output ranking is very close to the Kemeny ranking until $\normphi\leq 0.5$ (when their average normalized distance is only $0.004$). 
For a larger dispersion parameter and in particular for $\normphi\geq 0.7$, the distance from the Kemeny ranking become more similar for our three methods. 
This behavior is intuitive, recalling that for $\normphi=1$, profiles are ``chaotic'', with many  different rankings having comparable quality. 

For Borda, the rankings produced by the three methods are all around the same (small) distance from the Kemeny ranking. This distance increases steadily from $0$ for $\normphi=0$ to around $0.1$ for $\normphi=1$.

\subsection{Further Simulations}  
In \Cref{app:simulations}, we describe the results of further experiments. 
For instance, we analyze in which parts of the computed ranking the considered methods agree or disagree most.
We find that for both Plurality and Borda, for the top positions the Kemeny ranking agrees frequently with the \seqLoAbS rule. For the bottom positions it agrees with the \seqWiAbS rule. 
This suggests that one should use \seqLoAbS for identifying the best candidates and \seqWiAbS for avoiding the worst candidates.
Moreover, \seqWiAbS and \ScoreS agree more commonly on the top half of candidates, whereas \seqLoAbS and \ScoreS agree more commonly on the bottom half of candidates.

We repeat all our experiments on profiles sampled from Euclidean models.
Obtaining similar results, this confirms that our general observations from above also hold for profiles sampled from other distributions. 
We also analyze the influence of the number of voters and candidates on the results, observing that increasing the number of voters leads to an increased similarity of the rankings for Plurality and Borda, whereas increasing the number of candidates leads to an increased similarity for Borda but not for Plurality. 
We also consider additional scoring vectors. 
For instance, we find that for Veto the roles of \seqLoAbS and \seqWiAbS are reversed, which is to be expected, recalling \Cref{lem:equiv}.

\section{Complexity} \label{sec:complexity}

We study various computational problems related to Sequential-Winner and Sequential-Loser rules. 
By breaking ties arbitrarily, it is easy to compute \emph{some} ranking that is selected by such a rule. 
However, in some (high-stakes) applications, it might not be sufficient to simply output some ranking selected by the rule. For instance, some candidate could claim that there also exist other rankings selected by the same rule where that candidate is ranked higher.  
To check such claims, and understand which rankings can be selected in the presence of ties, we need an algorithm that for a given candidate $d$ and position $k$, decides whether $d$ is ranked in position $k$ in some ranking selected by the rule.
Accordingly, we introduce the following computational problem: 

\smallskip
\begin{center}
	{\small 
		\begin{tabularx}{1.0\columnwidth}{ll}
			\toprule
			\multicolumn{2}{c}{\textsc{Position-$k$ Determination} for social preference function $f$} \\
			\midrule
			\textbf{Given:}& \parbox[t]{0.75\columnwidth}{
			 A ranking profile $P$ over candidate set $C$, a designated candidate $d\in C$, and an integer $k\in [|C|]$.
				\vspace*{1mm}} \\%
			\textbf{Question:}& \parbox[t]{0.75\columnwidth}{
				Is there a ranking $\succ$ selected by $f$ on $P$  where $d$ is in position $k$, i.e., ${\succ} \in f(P)$ with $\rank(\succ,d)=k$?} \\ 
			\bottomrule
		\end{tabularx}
	}
\end{center}
\smallskip

Where possible, we will design (parameterized) algorithms that solve this problem.
We also prove hardness results, which will apply even to restricted versions of this problem that are most relevant in practice.
Specifically, we would expect candidates to mainly be interested if they can be ranked highly. 
Thus, we introduce the \TopK problem, where we ask whether a given candidate can be ranked in one of the first $k$ positions.\footnote{If we have an algorithm for \PositionK, we can solve the \TopK problem by using the algorithm for positions $i = 1, \dots, k$. (This is a Turing reduction.)}
Lastly,  the special case of both problems with $k=1$ is of particular importance: The \Winner problem asks whether the designated candidate is ranked in the first position in some ranking selected by the rule. 

For the three Sequential-Loser rules, it is known that their \Winner problem is NP-complete. For STV, this was stated by \citet{DBLP:conf/ijcai/ConitzerRX09}, and for Baldwin and Coombs, this was proven by \citet{mattei2014hard}. We will see that the corresponding \TopK problems for the Sequential-Winner rules are also NP-complete. Thus, since almost all of our problems turn out to be NP-hard, we take a more fine-grained view. 
In particular, we will study the influence of the number $n$ of voters and the number $m$ of candidates on the complexity of our problems. 
This analysis is not only of theoretical interest but also practically relevant, as  in many applications one of the two parameters is considerably smaller than the other (e.g., in political elections $m$ is typically much smaller than $n$, while in applications such as meta-search engines or ranking applicants, $n$ is often much smaller than $m$). 
\Cref{tab:sequentialWinner,table:sequential-loser} provide overviews of our results.

\begin{table*}[t]
	\centering
	\resizebox{1\textwidth}{!}{\begin{tabular}{llll}
		\toprule
			& & $n$ & $m$ \\
		\midrule
		Sequential-Plurality-Loser (STV) & NP-c. (Thm. \ref{th:STV-NP}) & FPT (Obs. \ref{ob:STV-n}) & FPT (Thm. \ref{th:parameter-m}) \\
		Sequential-Veto-Loser (Coombs) & NP-c. (Thm. \ref{th:Coombs-NP}) & W[1]-h. (Thm. \ref{th:Coombs-NP}), XP (Thm. \ref{th:Lo-Ve-K-n})& FPT (Thm. \ref{th:parameter-m})\\
		Sequential-Borda-Loser (Baldwin) & NP-c. (Thm. \ref{th:Baldwin-NP})  & NP-c. for $n = 8$ (Thm. \ref{th:Baldwin-NP}) & FPT (Thm. \ref{th:parameter-m}) \\
		\bottomrule
	\end{tabular}}
	\caption{Our results for Sequential-Loser rules. All hardness results hold for \Winner; all algorithmic results also apply to \PositionK. The unparameterized NP-hardness results in the first column were already stated or proven by \citet{DBLP:conf/ijcai/ConitzerRX09} and \citet{mattei2014hard}} \label{table:sequential-loser}
\end{table*}

\subsection{Parameter Number of Candidates}\label{sub:m}
We start by considering the parameter $m$, the number of candidates. 
It is easy to see that \PositionK for all Sequential-Winner and Sequential-Loser rules is fixed-parameter tractable with respect to $m$ (by iterating over all $m!$ possible output rankings).  
However, it is possible to improve the dependence on the parameter in the running time.

\prooflink{parameterm}
\begin{restatable}{theorem}{parameterm}
	\label{th:parameter-m}
	For every scoring system $\vs$, \PositionK can be solved in 
	\begin{itemize}
		\item $\mathcal{O}(2^m\cdot nm^2)$ time and $\mathcal{O}(m^{k}\cdot nm^2)$ time for  \seqWi{$\vs$}, and 
		\item  $\mathcal{O}(2^m\cdot nm^2)$ time and $\mathcal{O}(m^{m-k}\cdot nm^2)$ time for \seqLo{$\vs$}.
	\end{itemize}
\end{restatable}
\begin{proof}[Proof (algorithm)]
	We present an algorithm for \seqWiAb{$\vs$} (the results for \seqLoAb{$\vs$} directly follow from this by applying \Cref{lem:equiv}).
	We solve the problem via dynamic programming. 
	We call a subset $C'\subseteq C$ of candidates an \emph{elimination set} if there is a selected ranking where the candidates from $C'$ are ranked in the first $|C'|$ positions.
	We introduce a table $T$ with entry $T[C']$ for each subset $C'\subseteq C$ of candidates. 
	$T[C']$ is set to true if $C'$ is an elimination set. 
	We initialize the table by setting $T[\emptyset]$ to true. 
	Now we compute $T$ for each subset $C'\subseteq C$ in increasing order of the size of the subset using the following recurrence relation: 
	We set $T[C']$ to true if there is a candidate $c\in C'$ such that $T[C'\setminus \{c\}]$ is true and $c$ is an $\vs$-winner in $P|_{C\setminus (C'\setminus \{c\})}$.	
	
	After filling the table, we return ``true'' if and only if there is a subset $C'\subseteq C\setminus \{d\}$ with $|C'|=k-1$ such that $T[C']$ is true and $d$ is an $\vs$-winner in $P|_{C\setminus C'}$. 
	By filling the complete table we get a running time in $\mathcal{O}(2^m\cdot nm^2)$. 
	However, it is sufficient to only fill the table for all subsets of size at most $k-1$, resulting in a running time in $\mathcal{O}(m^{k}\cdot nm^2)$.
\end{proof}

\subsection{Sequential Loser} \label{sec:complexity-loser}
We study \seqLoAb{Plurality/Veto/Borda} (aka STV, Coombs, and Baldwin).
The \Winner problem is NP-hard for all three rules. 
\Cref{table:sequential-loser} shows an overview of our results. 
In particular, we get a clear separation of the rules for the number $n$ of voters:
\begin{itemize}
	\item \seqLoAb{Plurality} admits a simple FPT algorithm,
	\item \seqLoAb{Veto} is W[1]-hard but in XP,
	\item \seqLoAb{Borda} is NP-hard for $8$ voters.
\end{itemize}

\subsubsection{Plurality}

\citet{DBLP:conf/ijcai/ConitzerRX09} stated that \Winner for  \seqLoAb{Plurality}~(aka STV) is NP-hard. This result has been frequently cited and used. 
The proof was omitted in the conference paper, and to our knowledge no proof has ever appeared in published work.
To aid future research, we include a simple reduction here.
\begin{table*}[t]
	\centering
		\begin{tabular}{llllll}
			\toprule
			& & $n$ & $k$ & $n+k$ & $m$ \\
			\midrule
			Sequential-Plurality-Winner & NP-c. & W[1]-h., XP & W[1]-h., XP  & FPT  & FPT\\
			Sequential-Veto-Winner & NP-c. & FPT  & W[2]-h., XP  & FPT  & FPT \\
			Sequential-Borda-Winner & NP-c. & NP-h. for $n=8$ & W[1]-h., XP &? &FPT\\
			\bottomrule
	\end{tabular}
	\caption{Our results for Sequential-Winner rules. All hardness results hold for the \TopK problem; all algorithmic results also apply to the general \PositionK problem.} \label{tab:sequentialWinner}
\end{table*}

\begin{theorem} \label{th:STV-NP}
	\Winner for \seqLo{Plurality} (aka STV) is NP-hard. 
\end{theorem}
\begin{proof}
		We reduce from the NP-hard variant of \textsc{Satisfiability} where each clause contains at most three literals and each literal appears exactly twice \cite{DBLP:journals/eccc/ECCC-TR03-049}.
		Let $\varphi$ be a formula fulfilling these restrictions with clause set $F = \{c_1,\dots, c_m \}$ and variable set $X = \{x_1,\dots, x_n \}$. 
		Let $L = X \cup \overline X$ be the set of literals. 
		We construct a ranking profile with candidate set $C = \{ d, w \} \cup F \cup L$, where $d$ is our designated candidate, and the following voters:
	\begin{alignat*}{4}
	 \text{100 voters} \quad & d \succ \dots & \\
	\text{99 voters} \quad & w \succ d \succ \dots & \\
	\text{98 voters} \quad & c_j \succ w \succ d \succ \dots && \forall j\in [m] \\
	\text{60 voters} \quad & \ell \succ \overline{\ell} \succ w \succ d \succ \dots && \forall \ell \in L \\
	\text{2 voters}  \quad & \ell \succ c_j \succ w \succ d \succ \dots \quad && \forall \ell\in L, j\in [m] \text{ where $\ell$ appears in $c_j$}%
	\end{alignat*}
	
	For this ranking profile,  in every execution of \seqLo{Plurality} the first $n$ eliminated candidates must be a subset $L'\subseteq L$ of literals such that for every variable we select either its positive literal or its negative literal (but not both). In other words, $L'$ must satisfy $\ell\in L'\leftrightarrow \overline{\ell}\notin L'$. 
	To see this, note that all literal candidates initially have a Plurality score of  $64$, which is the lowest Plurality score in the profile, and that all other candidates have a higher Plurality score. 
	Thus, in the first round an arbitrary literal $\ell$ of some variable $x$ is eliminated. 
	This increases the Plurality score of the opposite literal $\overline{\ell}$ to over $120$.
	In the second round, we have to eliminate again an arbitrary literal (however,  this time a literal corresponding to a variable different from $x$).
	We repeat this process for $n$ rounds until for each variable exactly one of the corresponding literals has been eliminated. 
	We claim that an execution of \seqLo{Plurality} eliminates $d$ last if and only if the assignment that sets all literals from $L'$ to \emph{true} satisfies~$\varphi$. 
	
	Suppose $\varphi$ is satisfied by some variable assignment $\alpha$, and consider an execution of \seqLo{Plurality} that begins by eliminating the $n$ literals set to true in $\alpha$. After this, the scores of the remaining candidates are:
	\begin{enumerate}[label=(\roman*),leftmargin=0.8cm]
		\item $d$ has $100$ points,
		\item  $w$ has $99$ points,
		\item $c_j$ for $j\in [m]$ has between $100$ and $104$ points (as at least one of the literals occurring in $c_j$ has been eliminated), and 
		\item each literal $\ell\in L$ set to false by $\alpha$ has $124$ points. 
	\end{enumerate}

	In the next round, $w$ is eliminated, reallocating its $99$ points to $d$. 
	Then, in the next $m$ rounds, each clause candidate $c_j$ is eliminated, in each round reallocating its points to $d$. 
	Finally, the remaining literals are eliminated, also each reallocating their points to $d$. 
	Thus, $d$ is the last remaining candidate and ranked in the first position in the selected~ranking.
	
	Let $L'\subseteq L$  be the set of literals eliminated in the first $n$ rounds in some execution of the \seqLo{Plurality} rule (recall that $\ell\in L'\leftrightarrow \overline{\ell}\notin L'$). 
	Suppose that the assignment $\alpha$ setting all literals from $L'$ to true does not satisfy $\varphi$. 
	After the literals from $L'$ have been eliminated, the scores of the remaining candidate are:
	\begin{enumerate}[label=(\roman*),leftmargin=0.8cm]
		\item $d$ has $100$ points,
		\item  $w$ has $99$ points,
		\item $c_j$ for $j\in [m]$ where $\alpha$ satisfies $c_j$ has between $100$ and $104$ points,
		\item $c_j$ for $j\in [m]$ where $\alpha$ does not satisfy $c_j$ has $98$ points, and 
		\item each literal $\ell\in L$ set to false by $\alpha$ has $124$ points. 
	\end{enumerate}
	Thus, in the next round, one of the unsatisfied clauses is eliminated, redistributing its $98$ points to $w$ bringing the score of $w$ to $197$. 
	Because all but $100$ voters prefer $w$ to $d$, the Plurality score of $d$ will never exceed the score of $w$ in consecutive rounds, so $d$ cannot be eliminated last.
\end{proof}

Motivated by this hardness result, we now turn to the problem's parameterized complexity. 
We have already seen in \Cref{th:parameter-m} that the problem is solvable in $\mathcal{O}(2^m\cdot nm^2)$ time. 
Indeed, we show that unless the Exponential Time Hypothesis (ETH)\footnote{The ETH was introduced by \citet{DBLP:journals/jcss/ImpagliazzoPZ01} and states that \textsc{3-SAT} cannot be solved in $2^{o(n)}\cdot \mathrm{poly}(n)$ time where $n$ is the number of variables.} is false, we cannot hope to substantially improve the exponential part of this running time.

\prooflink{eth}
\begin{restatable}{theorem}{eth}
	\label{th:ETH}
	If the ETH is true, then \Winner for \seqLo{Plurality} (aka STV) cannot be solved in $2^{o(m)}\cdot \mathrm{poly}(n,m)$ time.
\end{restatable}

Turning to the number $n$ of voters, we can observe that initially only at most $n$ candidates have a non-zero Plurality score. All other candidates (which are not ranked first in any ranking) will be eliminated immediately, without thereby changing the Plurality scores of other candidates. After these eliminations, we are left with at most $n$ candidates. This makes it easy to see that \PositionK is fixed-parameter tractable with respect to $n$ (by using \Cref{th:parameter-m}).

\begin{observation}\label{ob:STV-n}
	\PositionK for \seqLo{Plurality} (aka STV) is solvable in $\mathcal{O}(2^n\cdot nm^2)$ time.
\end{observation}

\subsubsection{Veto}
We now turn to \seqLoAb{Veto} (aka Coombs). \citet{mattei2014hard} showed that the \Winner problem for this rule is NP-hard.
We give an alternative NP-hardness proof that also implies an ETH-based lower bound for the parameter $m$.
\prooflink{CoombsNP}
\begin{restatable}{theorem}{CoombsNP}
	\Winner for \seqLo{Veto} (aka. Coombs) is NP-complete. If the ETH is true, then the problem cannot be solved in $2^{o(m)}\cdot \mathrm{poly}(n,m)$ time.
\end{restatable}
For the parameter $n$, we show that the problem is W[1]-hard with respect to the number of voters. This is shown via an involved reduction from \textsc{Multicolored Independent Set}.
This result suggests that \seqLoAb{Veto} behaves quite differently from \seqLoAb{Plurality}, even if these two rules might seem ``symmetric'' to each other.

\prooflink{CoombsWhard}
\begin{restatable}{theorem}{CoombsWhard}
\label{th:Coombs-NP}
	\Winner for \seqLo{Veto} (aka. Coombs) is W[1]-hard with respect to the number $n$ of voters. 
\end{restatable}

However, on the positive side,  \Winner and even \PositionK  are solvable in polynomial-time if the number of voters is a constant. 
The intuition behind this result is that for \seqLoAb{Veto}, the ``status'' of an execution is fully captured by the \emph{bottom list} of the  ranking profile, i.e., a list containing the bottom-ranked candidate of each voter. 
Indeed, if we know the current bottom list, we can deduce exactly which candidates have been eliminated thus far.
As there are only $m^n$ many possibilities for the bottom list, dynamic programming yields an XP algorithm for \PositionK. 
\prooflink{vetoXP}
\begin{restatable}{theorem}{vetoXP}
	\label{th:Lo-Ve-K-n}
	\PositionK for \seqLo{Veto} is in XP with respect to the number $n$ of voters. 
\end{restatable}

\subsubsection{Borda}
We conclude by studying \seqLoAb{Borda} (aka Baldwin). \citet{mattei2014hard} proved that \Winner for this rule is NP-hard, adapting an earlier reduction about hardness of manipulation due to \citet{DBLP:journals/ai/DaviesKNWX14}. 
In fact, by giving a construction based on weighted majority graphs and using tools from \citet{bachmeier2019}, we prove that this NP-hardness persists even for only $n=8$ voters. 
This result suggests that the Borda scoring system leads to the hardest computational problems.

\prooflink{baldwinNP}
\begin{restatable}{theorem}{baldwinNP}
	\label{th:Baldwin-NP}
	Let $n \ge 8$ be a fixed even integer.
	Then \Winner for \seqLo{Borda} (aka Baldwin), restricted to instances with exactly $n$ voters, is NP-complete. In addition, if the ETH is true, then the problem cannot be solved in $2^{o(m)}\cdot \mathrm{poly}(m)$ time.
\end{restatable}

\subsection{Sequential Winner} \label{sec:complexity-winner}
In this subsection, we briefly summarize our results for  \seqWiAb{Plurality/Veto/Borda}, which to the best of our knowledge have not been previously studied  (for formal statements and proofs see \Cref{app:complexity-winner}). 
As \Winner is trivial for these rules, we focus on \TopK. 
\Cref{tab:sequentialWinner} displays an overview of our results. 
For all three rules, it turns out that \TopK is NP-hard and W[1]-hard with respect to $k$.
In contrast, for the parameter $n$, the picture is again more diverse: For Borda, we once more get NP-hardness for a constant number of voters ($n=8$), while Plurality and Veto switch their role (we have a fixed-parameter tractable algorithm for Veto and W[1]-hardness for Plurality). 
Recalling the equivalence from \Cref{lem:equiv}, this switch is unsurprising. Indeed, similar reductions are used here as for the corresponding results for Sequential-Loser for the other scoring system.

\section{Future Directions}

There are many directions for future work. In our complexity study, we have focused on the analysis of the space of possible outcomes. However, if we are happy to break ties immediately (e.g. by some fixed order), one could focus on finding the fastest algorithms for computing the output ranking. Interestingly, it is known that computing STV is P-complete \citep{csar2017winner}, so its computation is unlikely to be parallelizable.
An additional challenging open problem will be to determine if the hard problems we have identified become tractable if preferences are structured, for example single-peaked. Note that for single-peaked preferences, it is known that Coombs becomes a Condorcet extension and easy to compute \citep[Prop.~2]{grofman2004coombs}.
Further, ranking candidates by Borda score is known to give a 5-approximation of Kemeny's method \citep{coppersmith2006ordering}. This raises the question whether any other of the rules from our families provide an approximation?
Other specific questions left open by our work are whether ETH lower bounds can be obtained for additional problems, and whether they can be strengthened to SETH bounds. 
Finally, one could try to extend our results to other scoring vectors, and potentially prove dichotomy theorems.

\clearpage
\enlargethispage{20pt}

\clearpage

\appendix

\clearpage

\section{Additional Material for \Cref{sec:axioms}}\label{app:axioms}

In the main body, in \Cref{sec:axioms}, we have given an informal overview of axiomatic properties satisfied by the rules in our three families. In this appendix, we give formal statements of these results. In particular, we will give formal definitions of the relevant axioms.

Let us introduce some additional notation. For two ranking profiles $P = (\succ_1, \dots, \succ_n)$ and $P' = (\succ_1', \dots, \succ_n')$, both defined over the same candidate set, we write $P + P' = (\succ_1, \dots,  \succ_n, \succ_1', \dots, \succ_n')$ for the ranking profile obtained by concatenating the two lists. For an integer $k$, we write $kP = P + \cdots + P$ obtained by concatenating $k$ copies of $P$. For a set $S \subseteq \mathcal{L}(C)$ of rankings, we write $\cand(S, r) = \{ \cand({\succ}, r) : {\succ} \in S \}$ for the set of candidates that appear in position $r$ in at least one of the rankings in $S$. If ${\succ} \in \mathcal{L}(C)$ is a ranking and $\rho : C \to C$ is a permutation of the candidate set, then $\rho(\succ)$ is the ranking where for all pairs $a,b \in C$ of candidates, we have $\rho(a) \mathbin{\rho(\succ)} \rho(b)$ if and only if $a \succ b$. For a set $S \subseteq \mathcal{L}(C)$ of rankings, we write $\rho(S) = \{\rho(\succ) : {\succ} \in S\}$.

Let $f$ be a social preference function, defined for profiles with any number of voters and over all possible candidate sets. (We assume this large domain to be able to state axioms that reason about variable agendas (i.e. different candidate sets) and about variable electorates (i.e., different numbers of voters).) Whenever we do not specify otherwise, in the following axioms we implicitly quantify over all possible finite sets $C$ of candidates.

We begin with some basic axioms.
\begin{itemize}
	\item The rule $f$ is \emph{anonymous} if for all profiles $P = (\succ_1, \dots, \succ_n)$ and all permutations $\sigma : [n] \to [n]$, we have $f(P) = f((\succ_{\sigma(1)}, \dots, \succ_{\sigma(n)}))$. Thus, reordering the rankings does not change the outcome.
	\item The rule $f$ is \emph{neutral} if for all profiles $P = (\succ_1, \dots, \succ_n)$ and all permutations $\rho : C \to C$, we have $f(\rho(P)) = \rho(f(P))$. Thus, a relabeling of candidates leads to the same relabeling of the output.
	\item The rule $f$ is \emph{unanimous} if for all rankings ${\succ} \in \mathcal{L}(C)$ and all profiles $P = (\succ, \dots, \succ)$, where all rankings in $P$ are equal to $\succ$, we have $f(P) = \{{\succ}\}$.
	\item The rule $f$ is \emph{continuous} (sometimes known as the \emph{overwhelming majority} axiom) if for any two profiles $P$ and $P'$ over the same candidate set, there exists an integer $k$ such that $f(P + kP') \subseteq f(P')$.
\end{itemize}
The following are axioms about combining profiles.
\begin{itemize}
	\item The rule $f$ satisfies \emph{reinforcement} if for all profiles $P$ and $P'$ over the same candidate set, we have $f(P + P') = f(P) \cap f(P')$ whenever the intersection is non-empty.
	\item The rule $f$ satisfies \emph{reinforcement at the top} if for all profiles $P$ and $P'$ over the same candidate set, we have $\cand(f(P + P'),1) = \cand(f(P),1) \cap \cand(f(P'),1)$ whenever the intersection is non-empty.
	\item The rule $f$ satisfies \emph{reinforcement at the bottom} if for all profiles $P$ and $P'$ over the same candidate set $C$, we have $\cand(f(P + P'),|C|) = \cand(f(P),|C|) \cap \cand(f(P'),|C|)$ whenever the intersection is non-empty.
\end{itemize}
The following are independence axioms, describing that the output should not change when deleting certain candidates.
\begin{itemize}
	\item The rule $f$ satisfies \emph{independence at the top} if for all profiles $P$ and for all candidates $a \in \cand(f(P), 1)$ that can appear in first position in $f(P)$, we have that for all rankings ${\succ'} \in \mathcal{L}(C \setminus \{a\})$, ${\succ'} \in f(P|_{C\setminus\{a\}})$ if and only if the ranking $\succ''$, obtained by placing $a$ at the top of ranking $\succ'$, is a member of $f(P)$.
	\item The rule $f$ satisfies \emph{independence at the bottom} if for all profiles $P$ and for all candidates $a \in \cand(f(P), |C|)$ that can appear in last position in $f(P)$, we have that for all rankings ${\succ'} \in \mathcal{L}(C \setminus \{a\})$,  ${\succ'} \in f(P|_{C\setminus\{a\}})$ if and only if the ranking $\succ''$, obtained by placing $a$ at the bottom of ranking $\succ'$, is a member of $f(P)$.
\end{itemize}
Next, we introduce axioms that require a rule to follow the view of a majority of voters.
\begin{itemize}
	\item The rule $f$ places \emph{Condorcet winners at top} if for all profiles $P = (\succ_1, \dots, \succ_n)$ where there exists a candidate $a \in C$ such that for all other candidates $b \in C \setminus \{a\}$, a majority of voters prefers $a$ to $b$ (i.e., $|\{ i \in [n] : a \succ_i b\}| > |\{ i \in [n] : b \succ_i a\}|$), we have $\cand(f(P), 1) = \{a\}$. Thus, in profiles where a Condorcet winner exists, all output rankings must place it in the first position.
	\item The rule $f$ \emph{copies a majority ranking} if for all profiles $P= (\succ_1, \dots, \succ_n)$ where there exists a ranking $\succ$ which makes up more than half the profile (i.e., $|\{i \in [n] : {\succ_i} = {\succ}\}| > n/2$), we have $f(P) = \{{\succ}\}$.
\end{itemize}
We will state Tideman's independence of clones property later.

We will now state some axiomatic characterization results, taken or adapted from the literature. The first is a characterization of Kemeny's rule.

\begin{theorem}[\citealp{Youn88a}]
	A social preference function $f$ satisfies anonymity, neutrality, unanimity, reinforcement, independence at the top, and independence at the bottom, if and only if it is Kemeny's rule.
\end{theorem}

Next, there is an existing characterization of Sequential-Loser rules.

\begin{theorem}[\citealp{FBC14a}, Lemma 1]
	\label{thm:seqLoCharacterization}
	A social preference function $f$ satisfies anonymity, neutrality, unanimity, continuity, reinforcement at the bottom, and independence at the bottom, if and only if there exists a scoring system $\vs = (\vs^{(m)})_{m\in \mathbb{N}}$ such that $f$ equals \seqLo{\vs}.
\end{theorem}
\begin{proof}
	This is exactly Lemma 1 from \citet{FBC14a}, except that instead of continuity they use a condition called ``continuity at the bottom'', but this condition is weaker than continuity as we have defined it, because $f(P + kP') \subseteq f(P')$ implies that $\cand(f(P + kP'), |C|) \subseteq \cand(f(P'), |C|)$.
\end{proof}

We will now ``turn around'' \Cref{thm:seqLoCharacterization} to obtain an axiomatic characterization of Sequential-Winner rules. To do so, we will use \Cref{lem:equiv}. For a social preference function $f$, let us write $f^*$ for the social preference function defined as follows:
\[
f^*(P) = \rev(f(\rev(P)))
\quad
\text{for every profile $P$.}
\]
From \Cref{lem:equiv}, we know that if $f$ is \seqWiAb{\vs} then $f^*$ is \seqLoAb{$\vs^*$}, and if $f$ is \seqLoAb{\vs} then $f^*$ is \seqWiAb{$\vs^*$}. One can also easily deduce the following equivalences.
\begin{lemma}
	\label{lem:axiom-equiv}
	Lef $f$ be a social preference function, and let $f^*$ be defined as above. Then the following equivalences hold:
	\begin{itemize}
		\item $f$ satisfies anonymity (resp., neutrality, unanimity, continuity, reinforcement, copying a majority ranking) if and only if $f^*$ satisfies the respective axiom.
		\item $f$ satisfies reinforcement at the top (resp., at the bottom) if and only if $f^*$ satisfies reinforcement at the bottom (resp., at the top).
		\item $f$ satisfies independence at the top (resp., at the bottom) if and only if $f^*$ satisfies independence at the bottom (resp., at the top).
	\end{itemize}
\end{lemma}

\begin{theorem}
	A social preference function $f$ satisfies anonymity, neutrality, unanimity, continuity, reinforcement at the top, and independence at the top, if and only if there exists a scoring system $\vs = (\vs^{(m)})_{m\in \mathbb{N}}$ such that $f$ equals \seqWi{\vs}.
\end{theorem}
\begin{proof}
	It is routine to check that \seqWiAb{\vs} satisfies the mentioned axioms for every scoring system $\vs$. Now let $f$ be a social preference function satisfying these axioms. By \Cref{lem:axiom-equiv}, the social preference function $f^*$ then satisfies anonymity, neutrality, unanimity, continuity, reinforcement at the bottom, and independence at the bottom. By \Cref{thm:seqLoCharacterization}, there exists a scoring system $\vs$ such that $f^*$ equals \seqLoAb{\vs}. By \Cref{lem:equiv}, $f$ equals \seqWiAb{$\vs^*$}, as desired.
\end{proof}

\citet{FBC14a} also provided characterizations of specific rules. For example, they characterize \seqLoAb{Borda} (aka Baldwin) as the only Sequential-Loser rule that places Condorcet winners at the top.

\begin{theorem}[\citealp{FBC14a}, Theorem 3]
	A social preference function $f$ satisfies anonymity, neutrality, unanimity, continuity, reinforcement at the bottom, independence at the bottom, and places Condorcet winners at the top if and only if $f$ equals \seqLo{Borda}.
\end{theorem}

If desired, one could similarly characterize \seqWiAb{Borda} as the only Sequential-Winner rule that places Condorcet losers at the bottom, using \Cref{lem:axiom-equiv}.

\citet{FBC14a} also characterize \seqLoAb{Veto} (aka Coombs) as the only Sequential-Loser rule that copies a majority ranking.

\begin{theorem}[\citealp{FBC14a}, Theorem 2]
	A social preference function $f$ satisfies anonymity, neutrality, unanimity, continuity, reinforcement at the bottom, independence at the bottom, and copies majority rankings if and only if $f$ equals \seqLo{Veto}.
\end{theorem}

Using \Cref{lem:equiv,lem:axiom-equiv}, we can deduce a characterization of \seqWiAb{Plurality} as the only Sequential-Winner rule that copies majority rankings.

\begin{theorem}
	A social preference function $f$ satisfies anonymity, neutrality, unanimity, continuity, reinforcement at the top, independence at the top, and copies majority rankings if and only if $f$ equals \seqWi{Plurality}.
\end{theorem}

Turning towards \citeauthor{Tide87a}'s \citeyearpar{Tide87a} independence of clones property, \citet{FBC14a} adapt the axiom to the ranking context. Given a profile $P$, we say that $C' \subseteq C$ is a \emph{clone set} if in every ranking in $P$, the candidates in $C'$ appear consecutively (i.e., for each ${\succ} \in P$  and each $a \in C \setminus C'$ we either have $a \succ c'$ for all $c' \in C'$ or $c' \succ a$ for all $c' \in C'$.
\begin{itemize}
	\item The rule $f$ satisfies \emph{independence of clones (with top replacement)} if the following property is satisfied. Let $P$ be a profile and let $C'$ be a clone set. Let $P'$ be a profile obtained from $P$ by replacing the set $C'$ by a single candidate $a \not \in C$. Then for all rankings ${\succ} \in \mathcal{L}(C)$, we have ${\succ} \in f(P)$ if and only if ${\succ'} \in f(P')$ where $\succ'$ is the ranking obtained from $\succ$ by deleting the candidates in $C'$ and putting $a$ at the position of the highest-ranked member of $C'$ in $\succ$.
\end{itemize}
\citet{FBC14a} prove that within the family of Sequential-Loser rules, \seqLoAb{Plurality} (aka STV) is characterized by this axiom.%
\footnote{In their statement of the characterization, they appear to have forgotten to include continuity.}
\begin{theorem}[\citealp{FBC14a}, Theorem 1]
	A social preference function $f$ satisfies anonymity, neutrality, unanimity, continuity, reinforcement at the bottom, independence at the bottom, and independence of clones (with top replacement) if and only if $f$ equals \seqLo{Plurality}.
\end{theorem}
By invoking \Cref{lem:equiv,lem:axiom-equiv}, we can again obtain a related characterization of \seqWiAb{Veto}, but using a slightly different version of the clones axioms. In particular, let us define the axiom \emph{independence of clones (with bottom replacement)} exactly as before, except that the definition should end in saying ``putting $a$ at the position of the \emph{lowest}-ranked member of $C'$ in $\succ$.'' 
\begin{theorem}
	A social preference function $f$ satisfies anonymity, neutrality, unanimity, continuity, reinforcement at the top, independence at the top, and independence of clones (with bottom replacement) if and only if $f$ equals \seqWi{Veto}.
\end{theorem}
In the main body (and in particular \Cref{tbl:axioms}), we have taken the ``official'' version of independence of clones to be the one with top replacement.

\section{Additional Material for \Cref{sec:simulations}}\label{app:simulations}

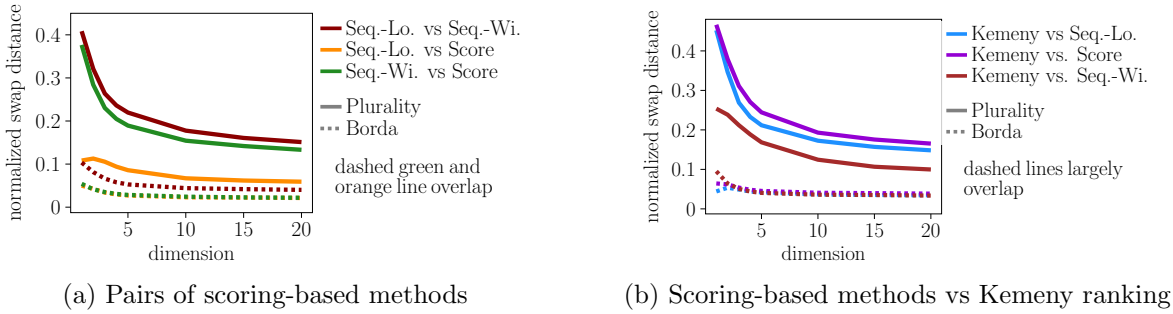
\begin{figure*}[t!]  
	\centering                  
	\begin{subfigure}[t]{0.485\textwidth} 
		\centering
		\resizebox{0.9\textwidth}{!}{
\begin{tikzpicture}[every plot/.append style={line width=3.5pt}]

\definecolor{color0}{rgb}{1,0.549019607843137,0}
\definecolor{color1}{rgb}{0.133333333333333,0.545098039215686,0.133333333333333}

\node[text width=5cm] at (10,1) 
{\LARGE dashed green and orange line overlap};

\begin{axis}[
ylabel shift = 1pt,
legend columns=1, 
legend cell align={left},
legend style={
  fill opacity=0.8,
  draw opacity=1,
  draw=none,
  text opacity=1,
  at={(1.45,1)},
  line width=3pt,
  anchor=north,
   /tikz/column 2/.style={
  	column sep=10pt,
  },
  font=\LARGE
},
every tick label/.append style={font=\LARGE}, 
label style={font=\LARGE},
tick align=outside,
tick pos=left,
x grid style={white!69.0196078431373!black},
xlabel={dimension},
xmin=0.0499999999999999, xmax=20.95,
xtick style={color=black},
y grid style={white!69.0196078431373!black},
ylabel={normalized swap distance},
ymin=-0.0127055555555556, ymax=0.45,
ytick style={color=black},
ytick={0,0.1,0.2,0.3,0.4},
yticklabels={0,0.1,0.2,0.3,0.4}
]
\addlegendentry{Seq.-Lo. vs Seq.-Wi.}
\addlegendentry{Seq.-Lo. vs Score}
\addlegendentry{Seq.-Wi. vs Score}
\addlegendentry{}
\addlegendentry{}
\addlegendentry{Plurality}
\addlegendentry{Borda}
\addlegendimage{red!54.5098039215686!black}
\addlegendimage{color0}
\addlegendimage{color1}
\addlegendimage{empty legend}
\addlegendimage{empty legend}
\addlegendimage{gray}
\addlegendimage{gray,dashed}
\addplot [semithick, red!54.5098039215686!black]
table {%
1 0.40844
2 0.320988888888889
3 0.263737777777778
4 0.236157777777778
5 0.219622222222222
10 0.177633333333333
15 0.160671111111111
20 0.151088888888889
};
\addplot [semithick, color0]
table {%
1 0.108215555555556
2 0.112782222222222
3 0.105455555555556
4 0.0940377777777778
5 0.0858355555555556
10 0.0669977777777778
15 0.0618311111111111
20 0.0591222222222222
};
\addplot [semithick, color1]
table {%
1 0.376575555555556
2 0.284184444444444
3 0.2303
4 0.204831111111111
5 0.18936
10 0.154293333333333
15 0.141888888888889
20 0.133171111111111
};

\addplot [semithick, red!54.5098039215686!black, dashed]
table {%
1 0.103524444444444
2 0.08154
3 0.0662622222222222
4 0.0577088888888889
5 0.05278
10 0.0441488888888889
15 0.0418622222222222
20 0.0401444444444444
};
\addplot [semithick, color0, dashed]
table {%
1 0.0499711111111111
2 0.0426711111111111
3 0.0337977777777778
4 0.0295911111111111
5 0.0271288888888889
10 0.0229088888888889
15 0.0222911111111111
20 0.02134
};
\addplot [semithick, color1,dashed]
table {%
1 0.0541622222222222
2 0.0414733333333333
3 0.0350422222222222
4 0.0307222222222222
5 0.0284377777777778
10 0.0243066666666667
15 0.0226911111111111
20 0.0218355555555556
};
\end{axis}

\end{tikzpicture}}
		\caption{Pairs of scoring-based methods}\label{fig:euc-general-comparsion}
	\end{subfigure}%
	\hfill
	\begin{subfigure}[t]{0.475\textwidth}
		\centering 
		\resizebox{0.9\textwidth}{!}{
\begin{tikzpicture}[every plot/.append style={line width=3.5pt}]

\definecolor{color0}{rgb}{1,0.549019607843137,0}
\definecolor{color1}{rgb}{0.133333333333333,0.545098039215686,0.133333333333333}
\definecolor{color2}{rgb}{0.117647058823529,0.564705882352941,1}
\definecolor{color3}{rgb}{0.580392156862745,0,0.827450980392157}
\definecolor{color4}{rgb}{0.647058823529412,0.164705882352941,0.164705882352941}

\node[text width=5cm] at (10,1) 
{\LARGE dashed lines largely overlap};

\begin{axis}[
legend columns=1, 
legend cell align={left},
legend style={
  fill opacity=0.8,
  draw opacity=1,
  draw=none,
  text opacity=1,
  at={(1.45,1)},
  line width=3pt,
  anchor=north,
   /tikz/column 2/.style={
  	column sep=10pt,
  },
  font=\LARGE
},
every tick label/.append style={font=\LARGE}, 
label style={font=\LARGE},
tick align=outside,
tick pos=left,
x grid style={white!69.0196078431373!black},
xlabel={dimension},
xmin=0.0499999999999999, xmax=20.95,
xtick style={color=black},
y grid style={white!69.0196078431373!black},
ylabel={normalized swap distance},
ymin=-0.0127055555555556, ymax=0.48,
ytick style={color=black},
ytick={0,0.1,0.2,0.3,0.4},
yticklabels={0,0.1,0.2,0.3,0.4}
]
\addlegendentry{Kemeny vs Seq.-Lo.}
\addlegendentry{Kemeny vs. Score}
\addlegendentry{Kemeny vs. Seq.-Wi.}
\addlegendentry{}
\addlegendentry{}
\addlegendentry{Plurality}
\addlegendentry{Borda}

\addlegendimage{color2}
\addlegendimage{color3}
\addlegendimage{color4}
\addlegendimage{empty legend}
\addlegendimage{empty legend}
\addlegendimage{gray}
\addlegendimage{gray,dashed}

\addplot [semithick, color2]
table {%
1 0.452335555555556
2 0.346791111111111
3 0.268931111111111
4 0.232924444444444
5 0.211762222222222
10 0.172577777777778
15 0.157006666666667
20 0.14824
};
\addplot [semithick, color3]
table {%
1 0.466266666666667
2 0.379404444444444
3 0.311062222222222
4 0.270708888888889
5 0.244451111111111
10 0.193042222222222
15 0.17562
20 0.165268888888889
};
\addplot [semithick, color4]
table {%
1 0.253846666666667
2 0.238451111111111
3 0.211815555555556
4 0.188651111111111
5 0.168504444444444
10 0.124522222222222
15 0.106677777777778
20 0.0997422222222222
};

\addplot [semithick, color2, dashed]
table {%
1 0.0441666666666667
2 0.0535933333333333
3 0.0489266666666667
4 0.0440288888888889
5 0.04194
10 0.0367111111111111
15 0.0352355555555556
20 0.0345266666666667
};
\addplot [semithick, color3, dashed]
table {%
1 0.0643644444444444
2 0.0619755555555556
3 0.0539066666666667
4 0.0484111111111111
5 0.0454644444444444
10 0.0408288888888889
15 0.0398422222222222
20 0.0386311111111111
};
\addplot [semithick, color4, dashed]
table {%
1 0.0952644444444444
2 0.0644488888888889
3 0.0506244444444444
4 0.0439244444444444
5 0.0404222222222222
10 0.0358377777777778
15 0.0350977777777778
20 0.0336088888888889
};
\end{axis}

\end{tikzpicture}}
		\caption{Scoring-based methods vs Kemeny ranking}\label{fig:euc-Kemeny-comparison}
	\end{subfigure}
	\caption{Pairwise average normalized swap distance between rankings produced by different methods for Plurality (solid) and Borda (dashed) on Euclidean profiles with $10$ candidates and $100$ voters.} \label{fig:Euc}
\end{figure*}

\subsection{Setup---Euclidean Model}

In addition to ranking profiles generated from the Mallows model, we have also considered profiles generated from the \emph{Euclidean model}. This model is parameterized by the dimension $d \ge 1$. 
To sample a Euclidean profile, for each candidate and voter we sample a point from the $d$-dimensional hypercube $[0,1]^d$ uniformly at random.
In the corresponding profile, each voter ranks the candidates in increasing order of their Euclidean $\ell_2$-distance to the voter.

\subsection{Comparison of Ranking Methods---Euclidean Model}

In the main body, in \Cref{sub:sim-scorecomp} and \Cref{sub:Kemeny-scorecomp}, we have analyzed the relation between rankings selected by the different rules on profiles generated using the Mallows model. 
To verify our results, we reran these experiments on profiles generated using the Euclidean model. 
Specifically, for each dimension $d\in \{1,2,3,4,5,10,15,20\}$, we generated \num{10000} profiles with $100$ voters and $10$ candidates.
We depict the results in \Cref{fig:Euc} (which is analogous to \Cref{fig:Mal}).

\paragraph{Comparison of Scoring-Based Ranking Methods}
We start by analyzing \Cref{fig:euc-general-comparsion}, where we compare the rankings selected by our different scoring-based rules. 
The general trend here is quite similar to Mallows profiles for a large dispersion parameter (see \Cref{fig:general-comparsion}):
For Borda, the agreement of the three methods is much higher than for Plurality, with \Score{Borda} producing rankings close to the other two. 
For Plurality, \seqLoAb{Plurality} and \Score{Plurality} produce again similar results, whereas the ranking produced by \seqWiAb{Plurality} differs more. 
The general level of disagreement between the rules for Plurality is remarkably high here. 
For $d=1$, the difference between \seqWiAb{Plurality} and the other two methods is around $0.4$, which is almost $0.5$ (the expected distance of two rankings drawn uniformly at random). %
Moreover, even for larger $d$, the level of disagreement remains high and is in particular around the level of disagreement for Mallows profiles with parameter $\normphi=1$. 
This is somewhat surprising, as profiles produced by the Mallows model with $\normphi=1$ are ``maximally chaotic'' and thus give the rules only limited information to distinguish the strength of candidates.

\paragraph{Comparison to Kemeny Ranking}
We now turn to the comparison of scoring-based ranking rules to Kemeny's method (\Cref{fig:euc-Kemeny-comparison}). For Plurality, like we have seen in Mallows profiles,  \seqWiAb{Plurality} produces the best results, followed by \seqLoAb{Plurality}, and lastly \Score{Plurality}. 
Considering the influence of the dimension~$d$, 
the difference between the method's distance to the Kemeny ranking is more or less the same for all dimensions with $d=1$ being the only exception: At $d=1$, \seqLoAb{Plurality} and \Score{Plurality} are at normalized swap distance $0.43$ from the Kemeny ranking, whereas \seqWiAb{Plurality} is at distance $0.27$, highlighting again that Euclidean profiles with $d=1$ are particularly challenging and that  \seqWiAb{Plurality} does best.

In contrast, for Borda, the rankings produced by the three methods are all around the same small distance from the Kemeny ranking (mostly independently of the dimension).

\begin{figure*}[t!]  
	\centering                  
	\begin{subfigure}[t]{0.475\textwidth} 
		\centering
		\resizebox{0.9\textwidth}{!}{
\begin{tikzpicture}[every plot/.append style={line width=3.5pt}]

\definecolor{color0}{rgb}{1,0.549019607843137,0}
\definecolor{color1}{rgb}{0.133333333333333,0.545098039215686,0.133333333333333}

\begin{axis}[
ylabel shift = 1pt,
legend columns=1, 
legend cell align={left},
legend style={
  fill opacity=0.8,
  draw opacity=1,
  draw=none,
  text opacity=1,
  at={(1.45,1)},
  line width=3pt,
  anchor=north,
   /tikz/column 2/.style={
  	column sep=10pt,
  },
  font=\LARGE
},
every tick label/.append style={font=\LARGE}, 
label style={font=\LARGE},
tick align=outside,
tick pos=left,
x grid style={white!69.0196078431373!black},
xlabel={position},
xmin=-0.45, xmax=9.45,
xtick style={color=black},
y grid style={white!69.0196078431373!black},
ylabel={average position displacement},
ymin=-0.07515, ymax=1.57815,
ytick style={color=black},
xtick={0,2,4,6,8},
xticklabels={1,3,5,7,9},
ytick={-0.2,0,0.2,0.4,0.6,0.8,1,1.2,1.4,1.6},
yticklabels={−0.2,0.0,0.2,0.4,0.6,0.8,1.0,1.2,1.4,1.6}
]
\addlegendentry{Seq.-Lo. vs Seq.-Wi.}
\addlegendentry{Seq.-Lo. vs Score}
\addlegendentry{Seq.-Wi. vs Score}
\addlegendentry{}
\addlegendentry{}
\addlegendentry{Plurality}
\addlegendentry{Borda}
\addlegendimage{red!54.5098039215686!black}
\addlegendimage{color0}
\addlegendimage{color1}
\addlegendimage{empty legend}
\addlegendimage{empty legend}
\addlegendimage{gray}
\addlegendimage{gray,dashed}
\addplot [semithick, red!54.5098039215686!black]
table {%
0 0.53355
1 0.7955
2 0.98605
3 1.1646
4 1.27845
5 1.33715
6 1.33685
7 1.2593
8 1.1313
9 1.06245
};
\addplot [semithick, color0]
table {%
0 0.548
1 0.75025
2 0.86375
3 0.87785
4 0.83635
5 0.7377
6 0.58085
7 0.3836
8 0.13705
9 0
};
\addplot [semithick, color1]
table {%
0 0
1 0.37
2 0.7599
3 1.0351
4 1.2132
5 1.31915
6 1.32915
7 1.2713
8 1.14705
9 1.07795
};

\addplot [semithick, red!54.5098039215686!black, dashed]
table {%
0 0.22215
1 0.37785
2 0.4335
3 0.43265
4 0.4215
5 0.4231
6 0.4269
7 0.43275
8 0.37015
9 0.21265
};
\addplot [semithick, color0, dashed]
table {%
0 0.21445
1 0.35205
2 0.35795
3 0.3114
4 0.2596
5 0.21285
6 0.1721
7 0.1281
8 0.0525
9 0
};
\addplot [semithick, color1,dashed]
table {%
0 0
1 0.04795
2 0.12455
3 0.17425
4 0.21345
5 0.26175
6 0.31135
7 0.3532
8 0.34125
9 0.20565
};
\end{axis}

\end{tikzpicture}}
		\caption{Pairs of scoring-based methods}\label{fig:posdisp-our}
	\end{subfigure}%
	\qquad 
	\begin{subfigure}[t]{0.475\textwidth}
		\centering   
		\resizebox{0.9\textwidth}{!}{
\begin{tikzpicture}[every plot/.append style={line width=3.5pt}]

\definecolor{color0}{rgb}{1,0.549019607843137,0}
\definecolor{color1}{rgb}{0.133333333333333,0.545098039215686,0.133333333333333}
\definecolor{color2}{rgb}{0.117647058823529,0.564705882352941,1}
\definecolor{color3}{rgb}{0.580392156862745,0,0.827450980392157}
\definecolor{color4}{rgb}{0.647058823529412,0.164705882352941,0.164705882352941}

\node[text width=5cm] at (10,1) 
{\LARGE dashed purple line overlaps with dashed red and blue line};

\begin{axis}[
ylabel shift = 1pt,
legend columns=1, 
legend cell align={left},
legend style={
  fill opacity=0.8,
  draw opacity=1,
  draw=none,
  text opacity=1,
  at={(1.45,1)},
  line width=3pt,
  anchor=north,
   /tikz/column 2/.style={
  	column sep=10pt,
  },
  font=\LARGE
},
every tick label/.append style={font=\LARGE}, 
label style={font=\LARGE},
xtick={0,2,4,6,8},
xticklabels={1,3,5,7,9},
tick align=outside,
tick pos=left,
x grid style={white!69.0196078431373!black},
xlabel={position},
xmin=-0.45, xmax=9.45,
xtick style={color=black},
y grid style={white!69.0196078431373!black},
ylabel={average position displacement},
ymin=-0.07515, ymax=1.57815,
ytick style={color=black},
ytick={-0.2,0,0.2,0.4,0.6,0.8,1,1.2,1.4,1.6},
yticklabels={−0.2,0.0,0.2,0.4,0.6,0.8,1.0,1.2,1.4,1.6}
]
\addlegendentry{Kemeny vs Seq.-Lo.}
\addlegendentry{Kemeny vs. Score}
\addlegendentry{Kemeny vs. Seq.-Wi.}
\addlegendentry{}
\addlegendentry{}
\addlegendentry{Plurality}
\addlegendentry{Borda}

\addlegendimage{color2}
\addlegendimage{color3}
\addlegendimage{color4}
\addlegendimage{empty legend}
\addlegendimage{empty legend}
\addlegendimage{gray}
\addlegendimage{gray,dashed}
\addplot [semithick, color2]
table {%
0 0.14595
1 0.5036
2 0.8544
3 1.11865
4 1.29735
5 1.37995
6 1.37545
7 1.3135
8 1.19765
9 1.1203
};
\addplot [semithick, color3]
table {%
0 0.6269
1 0.9432
2 1.18795
3 1.3525
4 1.45885
5 1.503
6 1.46045
7 1.3768
8 1.23685
9 1.1383
};
\addplot [semithick, color4]
table {%
0 0.59825
1 0.8097
2 0.9357
3 0.97625
4 0.96675
5 0.8845
6 0.80365
7 0.6474
8 0.40605
9 0.14455
};

\addplot [semithick, color2, dashed]
table {%
0 0.0626
1 0.17815
2 0.29525
3 0.3617
4 0.38675
5 0.40375
6 0.4215
7 0.4289
8 0.3658
9 0.2032
};
\addplot [semithick, color3, dashed]
table {%
0 0.2011
1 0.3698
2 0.42375
3 0.4339
4 0.416
5 0.4209
6 0.41795
7 0.4222
8 0.3661
9 0.1973
};
\addplot [semithick, color4,dashed]
table {%
0 0.20735
1 0.3727
2 0.4288
3 0.43645
4 0.4147
5 0.3845
6 0.34975
7 0.30175
8 0.1878
9 0.064
};
\end{axis}

\end{tikzpicture}}	
		\caption{Scoring-based methods vs Kemeny ranking}\label{fig:posdisp-Kemeny}
	\end{subfigure}
	\caption{For pairs of rankings, average position displacement on each position for profiles generated using the Mallows model with $\normphi=0.8$ with $10$ candidates and $100$ voters.} \label{fig:posdisp}
\end{figure*}
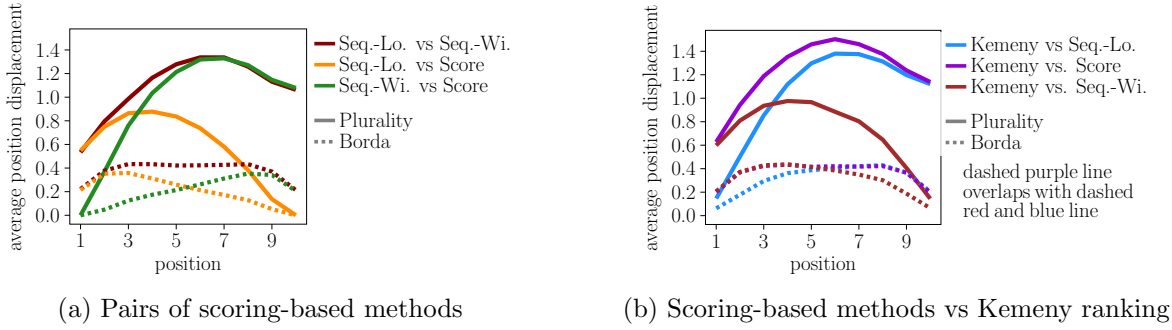
\subsection{Similarity in Different Vote Parts}
To shed some further light on the relation of the different methods, we next analyze in which parts of the computed ranking the considered methods agree or disagree most. 
For this, for two rankings ${\succ},{\succ'}\in \mathcal{L}(C)$, we define the \emph{position displacement} in position $i\in [|C|]$ as 
\[
\tfrac{1}{2} \left( |i-\text{pos}(\succ,\text{cand}(\succ',i))|+|i-\text{pos}(\succ',\text{cand}(\succ,i))|
\right).
\] 
The position displacement quantifies how far away the candidates ranked in position $i$ in one ranking are ranked in the other ranking.
Consider as an example the following two rankings: 
\begin{align*}
	a &\succ b \succ c \succ d \\
	d &\succ' c \succ' a \succ' b
\end{align*}
Then, the position displacement on position $1$ is $\frac{1}{2}\cdot (|1-\text{pos}(\succ,d)|+|1-\text{pos}(\succ',a)|)=\frac{1}{2}\cdot (|1-4|+|1-3|)=\frac{5}{2}$, whereas the position displacement on position $2$ is $\frac{1}{2}\cdot (|2-3|+|2-4|)=\frac{3}{2}$.
In \Cref{fig:posdisp}, we show the average position displacement on $\num{10000}$ profiles with $100$ voters and $10$ candidates sampled from the Mallows model with $\normphi=0.8$.%
\footnote{We rerun this experiment using profiles sampled from the Euclidean model with $d=10$ producing a very similar picture.}
(We chose $\normphi=0.8$ in order to ensure that the tie-breaking rule plays no critical role, while still keeping some structure in the profile.)

First, the general picture for Plurality and Borda is similar in the sense that for all comparisons of rules, the shape of the respective curves is similar. 
Thus, all observations described in the following hold for both Plurality and Borda.
Second, all four methods have a generally higher agreement on the top and bottom positions than on the middle positions.
Third, focusing on the comparison of the different scoring-based methods (\Cref{fig:posdisp-our}), by design \seqWiAbS and \ScoreS always place the same candidate in the first position. 
However, the agreement of the rules remains high in the second position and decreases continuously until position $7$. 
This indicates the intuitive behavior that the more candidates are present in the current round, the higher is the correlation between the scores of the candidates in the initial profile and their score in this round.
Comparing \seqLoAbS to \ScoreS, a reverse effect is present. 
For \seqLoAbS compared to \seqWiAbS, we almost have a symmetric curve with a generally slightly higher agreement on the top than on the bottom.  
Fourth, focusing on the comparison of the scoring-based methods and the Kemeny ranking (\Cref{fig:posdisp-Kemeny}),  
interestingly, on the top positions the Kemeny ranking agrees most with \seqLoAbS while on the bottom it agrees most with \seqWiAbS. 
This suggests that one should use \seqLoAbS for identifying the best candidates and \seqWiAbS for identifying the worst candidates.

\subsection{Number of Ties in Executions of the Rules}
The main motivation for our complexity analysis is that ties might occur in the execution of our rules. 
To better understand whether ties actually occur in practice (and to provide evidence for the explanation we gave in the main body about the behavior of \Score{Plurality} and \seqLoAb{Plurality} for small dispersion parameters), we conducted the following experiment.
We again sampled $\num{10000}$ profiles from the Mallows model, for each $\normphi\in \{0,0.1,\dots, 1\}$. 
For each profile, we executed our rules, as usual breaking ties according to a tie-breaking order $\succ_{\text{tie}}$ that we sampled randomly, and checked in each round whether a tie is present.\footnote{We use the following interpretation of the Score method. We start by computing the scores of the candidates in the initial profile. Subsequently, we eliminate a candidate with the highest number of points and add it at the first position in the selected ranking. However, we do not recompute the scores. In the second round, we eliminate the remaining candidate with the maximum number of points (in the initial profile) and add it in the second position. We repeat this process until all candidates are eliminated.}
The average number of rounds when a tie occurred is shown in \Cref{fig:ties}. 

\begin{figure*}[t!]  
	\centering                  
	\begin{subfigure}[t]{0.45\textwidth} 
		\centering
		\resizebox{0.9\textwidth}{!}{
\begin{tikzpicture}[every plot/.append style={line width=3.5pt}]

\definecolor{color0}{rgb}{1,0.549019607843137,0}
\definecolor{color1}{rgb}{0.133333333333333,0.545098039215686,0.133333333333333}

\begin{axis}[
ylabel shift = 1pt,
legend columns=1, 
legend cell align={left},
legend style={
	fill opacity=0.8,
	draw opacity=1,
	draw=none,
	text opacity=1,
	at={(1.25,1)},
	line width=3pt,
	anchor=north,
	/tikz/column 2/.style={
		column sep=10pt,
	},
	font=\LARGE
},
legend image post style={line width =3pt},
every tick label/.append style={font=\LARGE}, 
label style={font=\LARGE},
tick align=outside,
tick pos=left,
x grid style={white!69.0196078431373!black},
xlabel={normalized dispersion parameter},
xmin=-0.05, xmax=1.05,
xtick style={color=black},
xtick={-0.2,0,0.2,0.4,0.6,0.8,1,1.2},
xticklabels={−0.2,0.0,0.2,0.4,0.6,0.8,1.0,1.2},
y grid style={white!69.0196078431373!black},
ylabel={average number of rounds with tie},
ymin=-0.47365, ymax=9.94665,
ytick style={color=black}
]
\addplot [semithick, red!54.5098039215686!black]
table {%
0 8
0.1 5.4792
0.2 4.2857
0.3 3.2373
0.4 2.3061
0.5 1.6928
0.6 1.4306
0.7 1.3786
0.8 1.5189
0.9 1.7882
1 2.0056
};
\addlegendentry{Seq.-Lo.}
\addplot [semithick, color0]
table {%
0 0
0.1 0
0.2 0
0.3 0.0026
0.4 0.0269
0.5 0.1103
0.6 0.2514
0.7 0.482
0.8 0.8209
0.9 1.3222
1 1.5881
};
\addlegendentry{Seq.-Wi.}
\addplot [semithick, color1]
table {%
0 9.473
0.1 5.9603
0.2 3.7098
0.3 2.3344
0.4 1.4456
0.5 0.9963
0.6 0.869
0.7 0.9181
0.8 1.1343
0.9 1.5426
1 1.8357
};
\addlegendentry{Score}
\end{axis}

\end{tikzpicture}}
		\caption{Plurality}\label{fig:tie-plural}
	\end{subfigure}%
	\qquad
	\begin{subfigure}[t]{0.45\textwidth}
		\centering    
		\resizebox{0.9\textwidth}{!}{
\begin{tikzpicture}[every plot/.append style={line width=3.5pt}]

\definecolor{color0}{rgb}{1,0.549019607843137,0}
\definecolor{color1}{rgb}{0.133333333333333,0.545098039215686,0.133333333333333}

\begin{axis}[
ylabel shift = 1pt,
legend columns=1, 
legend cell align={left},
legend style={
	fill opacity=0.8,
	draw opacity=1,
	draw=none,
	text opacity=1,
	at={(1.25,1)},
	line width=3pt,
	anchor=north,
	/tikz/column 2/.style={
		column sep=10pt,
	},
	font=\LARGE
},
legend image post style={line width =3pt},
every tick label/.append style={font=\LARGE}, 
label style={font=\LARGE},
tick align=outside,
tick pos=left,
x grid style={white!69.0196078431373!black},
xlabel={normalized dispersion parameter},
xmin=-0.05, xmax=1.05,
xtick style={color=black},
xtick={-0.2,0,0.2,0.4,0.6,0.8,1,1.2},
xticklabels={−0.2,0.0,0.2,0.4,0.6,0.8,1.0,1.2},
y grid style={white!69.0196078431373!black},
ylabel={average number of rounds with tie},
ymin=-0.032845, ymax=0.689745,
ytick style={color=black},
ytick={-0.1,0,0.1,0.2,0.3,0.4,0.5,0.6,0.7},
yticklabels={−0.1,0.0,0.1,0.2,0.3,0.4,0.5,0.6,0.7}
]
\addplot [semithick, red!54.5098039215686!black]
table {%
0 0
0.1 0
0.2 0
0.3 0.0017
0.4 0.0091
0.5 0.0355
0.6 0.0828
0.7 0.1422
0.8 0.253
0.9 0.4457
1 0.6569
};
\addlegendentry{Seq.-Lo.}
\addplot [semithick, color0]
table {%
0 0
0.1 0
0.2 0
0.3 0.0009
0.4 0.009
0.5 0.0346
0.6 0.074
0.7 0.1461
0.8 0.2495
0.9 0.4347
1 0.6513
};
\addlegendentry{Seq.-Wi.}
\addplot [semithick, color1]
table {%
0 0
0.1 0
0.2 0
0.3 0.0009
0.4 0.0036
0.5 0.0112
0.6 0.0206
0.7 0.0378
0.8 0.0567
0.9 0.1153
1 0.2026
};
\addlegendentry{Score}
\end{axis}

\end{tikzpicture}}
		\caption{Borda}\label{fig:tie-borda}
	\end{subfigure}
	\caption{Average number of rounds in which a tie occurs for \seqWiAb{Plurality}, \seqLoAb{Plurality}, and \Score{Plurality}on profiles with $10$ candidates and $100$ voters sampled from the Mallows model.}\label{fig:ties}
\end{figure*}
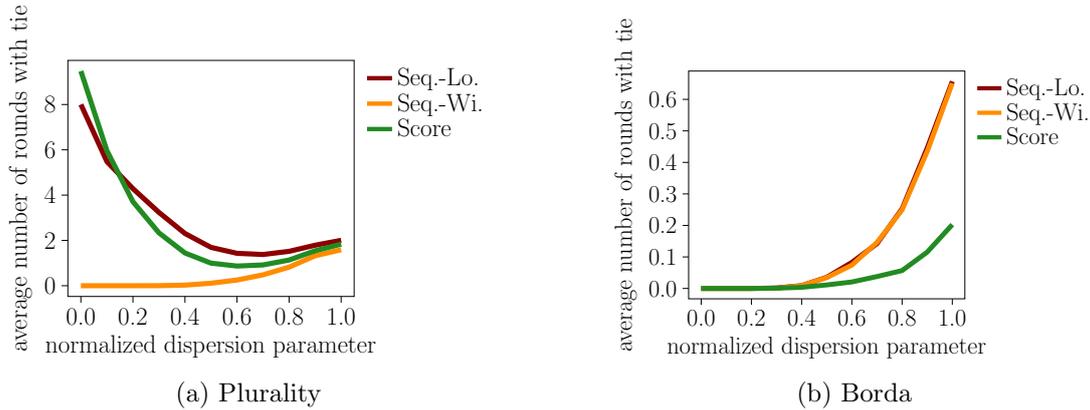  

For Plurality, we previously observed that on Mallows profiles with a small dispersion parameter, \seqLoAb{Plurality} and \Score{Plurality} produce very similar rankings. We mentioned that this could be explained by a large number of ties (that get resolved via the same tie-breaking order). This observation is clearly confirmed here (see \Cref{fig:tie-plural}).
In contrast, for \seqWiAb{Plurality} there are only very few (or no) ties for a small dispersion parameter, while more ties appear if the dispersion parameter is increased.
In general, for each choice of the normalized dispersion parameter,  \seqWiAb{Plurality}  produces the smallest number of ties among the three methods based on Plurality.
For $\normphi=1$, for all three rules, there are ties in around $2$ rounds on average (we were slightly surprised that this number is so high). 
Overall, these results show the importance of tie-breaking for ranking methods based on Plurality scores. 
They also provide an argument in favor of \seqWiAb{Plurality} over the other Plurality-based rules, as this rule leaves fewer decisions to the tie-breaking rule.

For Borda (\Cref{fig:tie-borda}), \Score{Borda} produces the least number of ties.
\seqWiAb{Borda} and \seqLoAb{Borda} have more ties, at very similar rates to each other.
The number of rounds with ties increases with a higher dispersion parameter. 
However, even for $\normphi=1$, there are only very few ties for all Borda-based rules: on average around $0.6$ rounds with a tie for \seqWiAb{Borda} and \seqLoAb{Borda}, and only $0.2$ rounds for \Score{Borda}. 
(Compare this to 2 rounds for the Plurality-based rules.)

For Euclidean models, the number of rounds with a tie mostly does not vary with the dimension, and there are many fewer ties than for the Mallows model. 
In particular, for \seqWi{Plurality}, the average number of rounds  with a tie is around $0.7$, for \Score{Plurality} it is around $1$, and for \seqLo{Plurality} it is around $1.4$.
Again we see that \seqWi{Plurality} produces the fewest ties among the Plurality-based rules. 
Moreover, for Borda the number of ties is again much lower than for Plurality, with \Score{Borda} producing the fewest ties, while \seqWiAb{Borda} and \seqLoAb{Borda} give  results very similar to each other. 
Specifically, for \Score{Borda} the average number of rounds with a tie is around $0.05$, and for \seqWiAb{Borda} and \seqLoAb{Borda} it is around $0.2$.

\subsection{Influence of Profile Size}
So far, we have focused on profiles with $n=100$ voters and $m=10$ candidates. Now, we examine the influence of the size of our profile on the results. 
First, we will analyze the influence of varying the number of voters and second, we will analyze the influence of varying the number of candidates.

\begin{figure*}[t!]  
	\centering                  
	\begin{subfigure}[t]{0.45\textwidth} 
		\centering
		\resizebox{0.9\textwidth}{!}{
\begin{tikzpicture}[every plot/.append style={line width=3.5pt}]

\definecolor{color0}{rgb}{1,0.549019607843137,0}
\definecolor{color1}{rgb}{0.133333333333333,0.545098039215686,0.133333333333333}
\definecolor{color2}{rgb}{0.117647058823529,0.564705882352941,1}
\definecolor{color3}{rgb}{0.580392156862745,0,0.827450980392157}
\definecolor{color4}{rgb}{0.647058823529412,0.164705882352941,0.164705882352941}

\begin{axis}[
ylabel shift = 1pt,
legend columns=1, 
legend cell align={left},
legend style={
	fill opacity=0.8,
	draw opacity=1,
	draw=none,
	text opacity=1,
	at={(1.45,1)},
	line width=3pt,
	anchor=north,
	/tikz/column 2/.style={
		column sep=10pt,
	},
	font=\LARGE
},
legend image post style={line width =3pt},
every tick label/.append style={font=\LARGE}, 
label style={font=\LARGE},
tick align=outside,
tick pos=left,
x grid style={white!69.0196078431373!black},
xlabel={number of voters},
xmin=1.25, xmax=523.75,
xtick style={color=black},
y grid style={white!69.0196078431373!black},
ylabel={normalized swap distance},
ymin=0.0116888888888889, ymax=0.272755555555556,
ytick style={color=black},
ytick={0,0.05,0.1,0.15,0.2,0.25,0.3},
yticklabels={0.00,0.05,0.10,0.15,0.20,0.25,0.30}
]
\addplot [semithick, red!54.5098039215686!black]
table {%
	25 0.209091111111111
	50 0.176697777777778
	100 0.137026666666667
	200 0.09988
	300 0.0796622222222222
	400 0.0665444444444444
	500 0.0575066666666667
};
\addlegendentry{Seq.-Lo. vs Seq.-Wi.}
\addplot [semithick, color0]
table {%
	25 0.0731844444444444
	50 0.0781577777777778
	100 0.0686577777777778
	200 0.0529177777777778
	300 0.0424911111111111
	400 0.0359866666666667
	500 0.0309711111111111
};
\addlegendentry{Seq.-Lo. vs Score}
\addplot [semithick, color1]
table {%
	25 0.189186666666667
	50 0.156766666666667
	100 0.122542222222222
	200 0.0896644444444444
	300 0.0725177777777778
	400 0.06102
	500 0.0534911111111111
};
\addlegendentry{Seq.-Wi. vs Score}
\addplot [semithick, color2, dashed]
table {%
	25 0.2291
	50 0.180926666666667
	100 0.132515555555556
	200 0.0911977777777778
	300 0.0713822222222222
	400 0.0585977777777778
	500 0.0505755555555556
};
\addlegendentry{Kemeny vs Seq.-Lo.}
\addplot [semithick, color3, dashed]
table {%
	25 0.256835555555556
	50 0.209431111111111
	100 0.160022222222222
	200 0.113977777777778
	300 0.0905355555555555
	400 0.0748555555555556
	500 0.0646666666666667
};
\addlegendentry{Kemeny vs Score}
\addplot [semithick, color4, dashed]
table {%
	25 0.157564444444444
	50 0.124166666666667
	100 0.0872888888888889
	200 0.0556733333333333
	300 0.0409377777777778
	400 0.0319022222222222
	500 0.0254244444444444
};
\addlegendentry{Kemeny vs Seq.-Wi.}
\end{axis}

\end{tikzpicture}}
		\caption{Mallows model for $\normphi=0.8$}\label{fig:n-mallows}
	\end{subfigure}%
	\qquad
	\begin{subfigure}[t]{0.45\textwidth}
		\centering    
		\resizebox{0.9\textwidth}{!}{
\begin{tikzpicture}[every plot/.append style={line width=3.5pt}]

\definecolor{color0}{rgb}{1,0.549019607843137,0}
\definecolor{color1}{rgb}{0.133333333333333,0.545098039215686,0.133333333333333}
\definecolor{color2}{rgb}{0.117647058823529,0.564705882352941,1}
\definecolor{color3}{rgb}{0.580392156862745,0,0.827450980392157}
\definecolor{color4}{rgb}{0.647058823529412,0.164705882352941,0.164705882352941}

\begin{axis}[
ylabel shift = 1pt,
legend columns=1, 
legend cell align={left},
legend style={
	fill opacity=0.8,
	draw opacity=1,
	draw=none,
	text opacity=1,
	at={(1.45,1)},
	line width=3pt,
	anchor=north,
	/tikz/column 2/.style={
		column sep=10pt,
	},
	font=\LARGE
},
legend image post style={line width =3pt},
every tick label/.append style={font=\LARGE}, 
label style={font=\LARGE},
tick align=outside,
tick pos=left,
x grid style={white!69.0196078431373!black},
xlabel={number of voters},
xmin=1.25, xmax=523.75,
xtick style={color=black},
y grid style={white!69.0196078431373!black},
ylabel={normalized swap distance},
ymin=0.0472888888888889, ymax=0.273155555555556,
ytick style={color=black},
ytick={0,0.05,0.1,0.15,0.2,0.25,0.3},
yticklabels={0.00,0.05,0.10,0.15,0.20,0.25,0.30}
]
\addplot [semithick, red!54.5098039215686!black]
table {%
	25 0.220215555555556
	50 0.198222222222222
	100 0.177215555555556
	200 0.163326666666667
	300 0.157622222222222
	400 0.153148888888889
	500 0.152908888888889
};
\addlegendentry{Seq.-Lo. vs Seq.-Wi.}
\addplot [semithick, color0]
table {%
	25 0.0620977777777778
	50 0.0683911111111111
	100 0.0673555555555556
	200 0.0631933333333333
	300 0.0604866666666667
	400 0.0592733333333333
	500 0.0579355555555556
};
\addlegendentry{Seq.-Lo. vs Score}
\addplot [semithick, color1]
table {%
	25 0.201851111111111
	50 0.178204444444444
	100 0.154082222222222
	200 0.138155555555556
	300 0.131895555555556
	400 0.126537777777778
	500 0.125395555555556
};
\addlegendentry{Seq.-Wi. vs Score}
\addplot [semithick, color2, dashed]
table {%
	25 0.238066666666667
	50 0.202075555555556
	100 0.171926666666667
	200 0.151322222222222
	300 0.144068888888889
	400 0.138262222222222
	500 0.135684444444444
};
\addlegendentry{Kemeny vs Seq.-Lo.}
\addplot [semithick, color3, dashed]
table {%
	25 0.259951111111111
	50 0.224248888888889
	100 0.192984444444444
	200 0.170031111111111
	300 0.161297777777778
	400 0.154637777777778
	500 0.151828888888889
};
\addlegendentry{Kemeny vs Score}
\addplot [semithick, color4, dashed]
table {%
	25 0.152522222222222
	50 0.136191111111111
	100 0.122413333333333
	200 0.112457777777778
	300 0.108091111111111
	400 0.105184444444444
	500 0.104837777777778
};
\addlegendentry{Kemeny vs Seq.-Wi.}
\end{axis}

\end{tikzpicture}}
		\caption{Euclidean model for $d=10$}\label{fig:n-Euclidean}
	\end{subfigure}
	\caption{Pairwise average normalized swap distance between the Kemeny method, \seqWiAb{Plurality}, \seqLoAb{Plurality}, and \Score{Plurality} rankings on profiles with $10$ candidates and a varying number of voters. }\label{fig:n-comparision}
\end{figure*}
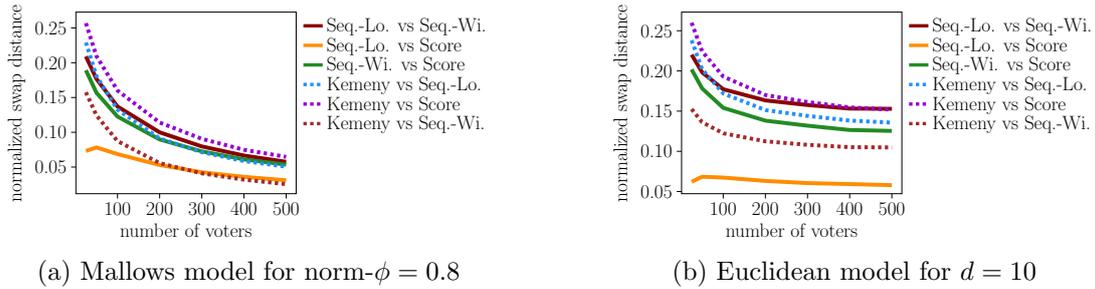  

\paragraph{Varying the number of voters}
In \Cref{fig:n-comparision}, we depict the pairwise difference of our ranking methods for Plurality and the Kemeny ranking, for profiles with $10$ candidates and a varying number of voters.
The profiles are generated using the Mallows model with $\normphi=0.8$ (\Cref{fig:n-mallows}) and the Euclidean model with $d=10$ (\Cref{fig:n-Euclidean}). 
For each $n\in \{25,50,100,200,300,400,500\}$, we generated $\num{10000}$ profiles.
For both generation models, the different methods become pairwise more similar with a higher  number of voters.
For Mallows the distance between the rankings decreases steeply, while for the Euclidean model the decrease (after $n=100$) is slower. 
Generally speaking, increasing the number of voters gives us additional information about the strengths of the candidates and reduces the probability of artifacts. 
For Mallows profiles, there exists a clear ordering of the candidates in terms of their strengths (namely the central order), and additional voters clarify this situation. 
For Euclidean profiles, candidates are less clearly distinguishable; explaining why the four approaches do not all ``converge'' to the same ranking as the number of voters increases (unlike for Mallows). 
For both models, the ordering of pairs of our three scoring-based ranking methods in terms of their similarity is independent of the number of voters (the same is also the case for their ordering with respect to their similarity to the Kemeny ranking).

The above described general trends are also present if we use Borda instead of Plurality.

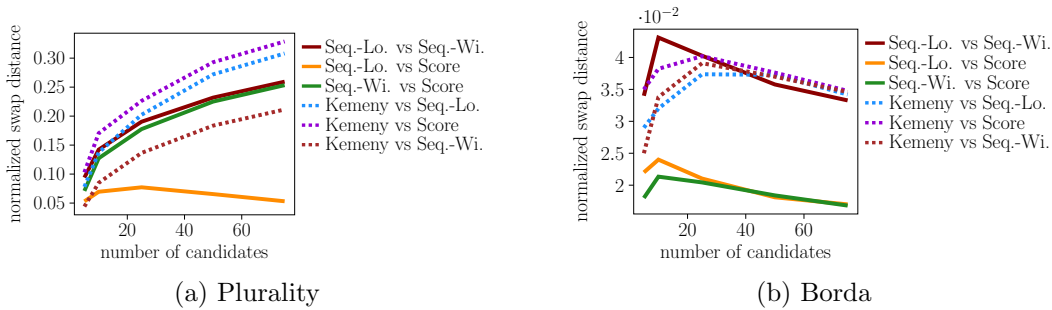
\begin{figure*}[t!]  
	\centering                  
	\begin{subfigure}[t]{0.45\textwidth} 
		\centering
		\resizebox{0.9\textwidth}{!}{
\begin{tikzpicture}[every plot/.append style={line width=3.5pt}]

\definecolor{color0}{rgb}{1,0.549019607843137,0}
\definecolor{color1}{rgb}{0.133333333333333,0.545098039215686,0.133333333333333}
\definecolor{color2}{rgb}{0.117647058823529,0.564705882352941,1}
\definecolor{color3}{rgb}{0.580392156862745,0,0.827450980392157}
\definecolor{color4}{rgb}{0.647058823529412,0.164705882352941,0.164705882352941}

\begin{axis}[
ylabel shift = 1pt,
legend columns=1, 
legend cell align={left},
legend style={
	fill opacity=0.8,
	draw opacity=1,
	draw=none,
	text opacity=1,
	at={(1.45,1)},
	line width=3pt,
	anchor=north,
	/tikz/column 2/.style={
		column sep=10pt,
	},
	font=\LARGE
},
legend image post style={line width =3pt},
every tick label/.append style={font=\LARGE}, 
label style={font=\LARGE},
tick align=outside,
tick pos=left,
x grid style={white!69.0196078431373!black},
xlabel={number of candidates},
xmin=1.5, xmax=78.5,
xtick style={color=black},
y grid style={white!69.0196078431373!black},
ylabel={normalized swap distance},
ymin=0.0297612612612613, ymax=0.343013513513514,
ytick style={color=black},
ytick={0,0.05,0.1,0.15,0.2,0.25,0.3,0.35},
yticklabels={0.00,0.05,0.10,0.15,0.20,0.25,0.30,0.35}
]
\addplot [semithick, red!54.5098039215686!black]
table {%
5 0.094
10 0.142222222222222
25 0.190366666666667
50 0.231918367346939
75 0.259427027027027
};
\addlegendentry{Seq.-Lo. vs Seq.-Wi.}
\addplot [semithick, color0]
table {%
5 0.053
10 0.0695555555555556
25 0.0772
50 0.0655102040816326
75 0.0530990990990991
};
\addlegendentry{Seq.-Lo. vs Score}
\addplot [semithick, color1]
table {%
5 0.071
10 0.127333333333333
25 0.177633333333333
50 0.22505306122449
75 0.253398198198198
};
\addlegendentry{Seq.-Wi. vs Score}
\addplot [semithick, color2, dashed]
table {%
5 0.078
10 0.137333333333333
25 0.2024
50 0.272057142857143
75 0.308057657657658
};
\addlegendentry{Kemeny vs Seq.-Lo.}
\addplot [semithick, color3, dashed]
table {%
5 0.103
10 0.170444444444444
25 0.226333333333333
50 0.292979591836735
75 0.328774774774775
};
\addlegendentry{Kemeny vs Score}
\addplot [semithick, color4, dashed]
table {%
5 0.044
10 0.0848888888888889
25 0.136566666666667
50 0.183583673469388
75 0.21141981981982
};
\addlegendentry{Kemeny vs Seq.-Wi.}
\end{axis}

\end{tikzpicture}}
		\caption{Plurality}\label{fig:m-plural}
	\end{subfigure}%
	~ 
	\begin{subfigure}[t]{0.45\textwidth} 
		\centering    
		\resizebox{0.9\textwidth}{!}{
\begin{tikzpicture}[every plot/.append style={line width=3.5pt}]

\definecolor{color0}{rgb}{1,0.549019607843137,0}
\definecolor{color1}{rgb}{0.133333333333333,0.545098039215686,0.133333333333333}
\definecolor{color2}{rgb}{0.117647058823529,0.564705882352941,1}
\definecolor{color3}{rgb}{0.580392156862745,0,0.827450980392157}
\definecolor{color4}{rgb}{0.647058823529412,0.164705882352941,0.164705882352941}

\begin{axis}[
ylabel shift = 1pt,
legend columns=1, 
legend cell align={left},
legend style={
	fill opacity=0.8,
	draw opacity=1,
	draw=none,
	text opacity=1,
	at={(1.45,1)},
	line width=3pt,
	anchor=north,
	/tikz/column 2/.style={
		column sep=10pt,
	},
	font=\LARGE
},
legend image post style={line width =3pt},
every tick label/.append style={font=\LARGE}, 
label style={font=\LARGE},
tick align=outside,
tick pos=left,
x grid style={white!69.0196078431373!black},
xlabel={number of candidates},
xmin=1.5, xmax=78.5,
xtick style={color=black},
y grid style={white!69.0196078431373!black},
ylabel={normalized swap distance},
ymin=0.015492012012012, ymax=0.0444263063063063,
ytick style={color=black}
]
\addplot [semithick, red!54.5098039215686!black]
table {%
5 0.034
10 0.0431111111111111
25 0.0402666666666667
50 0.0357714285714286
75 0.0332720720720721
};
\addlegendentry{Seq.-Lo. vs Seq.-Wi.}
\addplot [semithick, color0]
table {%
5 0.022
10 0.024
25 0.0210333333333333
50 0.0181061224489796
75 0.016990990990991
};
\addlegendentry{Seq.-Lo. vs Score}
\addplot [semithick, color1]
table {%
5 0.018
10 0.0213333333333333
25 0.0204333333333333
50 0.0184
75 0.0168072072072072
};
\addlegendentry{Seq.-Wi. vs Score}
\addplot [semithick, color2, dashed]
table {%
5 0.029
10 0.032
25 0.0373666666666667
50 0.0372816326530612
75 0.0341873873873874
};
\addlegendentry{Kemeny vs Seq.-Lo.}
\addplot [semithick, color3, dashed]
table {%
5 0.035
10 0.0382222222222222
25 0.0402
50 0.0376244897959184
75 0.0347675675675676
};
\addlegendentry{Kemeny vs Score}
\addplot [semithick, color4, dashed]
table {%
5 0.025
10 0.0337777777777778
25 0.0391
50 0.0369877551020408
75 0.0345873873873874
};
\addlegendentry{Kemeny vs Seq.-Wi.}
\end{axis}

\end{tikzpicture}}
		\caption{Borda}\label{fig:m-borda}
	\end{subfigure}
	\caption{Pairwise average normalized swap distance between the Kemeny ranking and rankings produced by our three scoring-based methods on profiles sampled from Mallows model with $\normphi=0.8$ with $100$ voters and a varying number of candidates.}\label{fig:m-comparision}
\end{figure*}

\paragraph{Varying the number of candidates}
We now turn to analyzing the influence of the number of candidates. 
In \Cref{fig:m-comparision}, we depict the pairwise distances between our methods for Plurality (\Cref{fig:m-plural}) and Borda (\Cref{fig:m-borda}) in profiles with $100$ voters and a varying number of candidates.
The profiles are generated using the Mallows model with $\normphi=0.8$. For the Euclidean model with $d=10$, the results are similar, so we omit them. 
For each $m\in \{5,10,25,50,75\}$, we generated $100$ profiles. 
Note that our use of \emph{normalized} distances is particularly convenient when comparing results with differing numbers of candidates.

We start by examining the results for Plurality (\Cref{fig:m-plural}). 
Here, as the number of candidates increases, the average number of Plurality points per candidate decreases and in particular more candidates get a Plurality score of zero. 
This leads to more ties in the execution of  \seqLoAb{Plurality} and \Score{Plurality} when determining the ordering of candidates at the bottom of the output ranking.
Because both methods break ties using the same tie-breaking order $\succ_{\text{tie}}$, \seqLoAb{Plurality} and \Score{Plurality} become more similar as the number of candidates increases. 
On the flip side, when comparing  \seqWiAb{Plurality} to \seqLoAb{Plurality} (or when comparing \seqWiAb{Plurality} to \Score{Plurality}), they become less similar as the number of candidates increases. 
This is because \seqWiAb{Plurality} is able to distinguish candidates who initially have low Plurality scores and does not need to rely on tie-breaking. Indeed, when \seqWiAb{Plurality} starts to rank the weak candidates, all other candidates have already been deleted and thus the average Plurality score of the weak candidates is higher and more informative. 

Comparing our three scoring-based methods to the Kemeny ranking, we see that the distance increases with more candidates. 
This is because decreasing the average number of Plurality points per candidates makes it harder for our three scoring-based methods to distinguish the strengths of candidates; intuitively speaking, as we increase the number of candidates, the information provided by only examining the first position decreases (but, as discussed above, this effect is smaller for \seqWiAb{Plurality}).

For Borda (\Cref{fig:m-borda}), the general trend is reversed. With small exceptions, the higher the number of candidates, the more similar are the rankings produced by the different methods.
One effect that potentially contributes to this is that for a higher number of candidates the range of awarded points increases thereby allowing for a clearer distinction of the candidates.
Nevertheless, for both Plurality and Borda, the ordering of the pairs of methods remains largely unaffected by changing the number of candidates.

\medskip

Overall, the results from this section suggest that while the size of the profile in question influences the level of similarity of the different methods, the general trends observed in the previous sections hold mostly independent of the size of the profile.

\begin{figure*}[t!]  
	\centering                  
	\begin{subfigure}[t]{0.475\textwidth} 
		\centering
		\resizebox{0.9\textwidth}{!}{
\begin{tikzpicture}[every plot/.append style={line width=3.5pt}]

\definecolor{color0}{rgb}{1,0.549019607843137,0}
\definecolor{color1}{rgb}{0.133333333333333,0.545098039215686,0.133333333333333}
\definecolor{color2}{rgb}{0.117647058823529,0.564705882352941,1}
\definecolor{color3}{rgb}{0.580392156862745,0,0.827450980392157}
\definecolor{color4}{rgb}{0.647058823529412,0.164705882352941,0.164705882352941}

\begin{axis}[
ylabel shift = 1pt,
legend columns=1, 
legend cell align={left},
legend style={
	fill opacity=0.8,
	draw opacity=1,
	draw=none,
	text opacity=1,
	at={(1.45,1)},
	line width=3pt,
	anchor=north,
	/tikz/column 2/.style={
		column sep=10pt,
	},
	font=\LARGE
},
legend image post style={line width =3pt},
every tick label/.append style={font=\LARGE}, 
label style={font=\LARGE},
tick align=outside,
tick pos=left,
x grid style={white!69.0196078431373!black},
xlabel={normalized dispersion parameter},
xmin=-0.05, xmax=1.05,
xtick style={color=black},
xtick={-0.2,0,0.2,0.4,0.6,0.8,1,1.2},
xticklabels={−0.2,0.0,0.2,0.4,0.6,0.8,1.0,1.2},
y grid style={white!69.0196078431373!black},
ylabel={normalized swap distance},
ymin=-0.0200004444444444, ymax=0.420009333333333,
ytick style={color=black},
ytick={-0.05,0,0.05,0.1,0.15,0.2,0.25,0.3,0.35,0.4,0.45},
yticklabels={−0.05,0.00,0.05,0.10,0.15,0.20,0.25,0.30,0.35,0.40,0.45}
]
\addplot [semithick, red!54.5098039215686!black]
table {%
0 0.400008888888889
0.1 0.195264444444444
0.2 0.126728888888889
0.3 0.0895466666666667
0.4 0.07388
0.5 0.0716511111111111
0.6 0.0799222222222222
0.7 0.100095555555556
0.8 0.137942222222222
0.9 0.202384444444444
1 0.254024444444444
};
\addlegendentry{Seq.-Lo. vs Seq.-Wi.}
\addplot [semithick, color0]
table {%
0 0.400008888888889
0.1 0.195264444444444
0.2 0.12678
0.3 0.0898844444444444
0.4 0.0745911111111111
0.5 0.0722222222222222
0.6 0.0785977777777778
0.7 0.09506
0.8 0.123095555555556
0.9 0.168168888888889
1 0.200193333333333
};
\addlegendentry{Seq.-Lo. vs Score}
\addplot [semithick, color1]
table {%
0 0
0.1 1.33333333333333e-05
0.2 0.000228888888888889
0.3 0.00117777777777778
0.4 0.00407555555555556
0.5 0.0103133333333333
0.6 0.0206977777777778
0.7 0.0386355555555556
0.8 0.0687755555555556
0.9 0.114882222222222
1 0.147404444444444
};
\addlegendentry{Seq.-Wi. vs Score}
\addplot [semithick, color2, dashed]
table {%
0 0
0.1 0
0.2 0
0.3 0.000117777777777778
0.4 0.00154888888888889
0.5 0.00724666666666667
0.6 0.0198933333333333
0.7 0.04306
0.8 0.0871622222222222
0.9 0.176882222222222
1 0.26872
};
\addlegendentry{Kemeny vs Seq.-Lo.}
\addplot [semithick, color3, dashed]
table {%
0 0.400008888888889
0.1 0.195264444444444
0.2 0.12678
0.3 0.0899222222222222
0.4 0.0750866666666667
0.5 0.0744333333333333
0.6 0.0856066666666667
0.7 0.11108
0.8 0.159364444444444
0.9 0.251846666666667
1 0.336415555555556
};
\addlegendentry{Kemeny vs Score}
\addplot [semithick, color4, dashed]
table {%
0 0.400008888888889
0.1 0.195264444444444
0.2 0.126728888888889
0.3 0.0895488888888889
0.4 0.07382
0.5 0.0708044444444444
0.6 0.0777311111111111
0.7 0.0958888888888889
0.8 0.131971111111111
0.9 0.211675555555556
1 0.294557777777778
};
\addlegendentry{Kemeny vs Seq.-Wi.}
\end{axis}

\end{tikzpicture}}
		\caption{Veto}\label{fig:veto}
	\end{subfigure}%
	\qquad
	\begin{subfigure}[t]{0.475\textwidth}
		\centering 
		\resizebox{0.9\textwidth}{!}{
\begin{tikzpicture}[every plot/.append style={line width=3.5pt}]

\definecolor{color0}{rgb}{1,0.549019607843137,0}
\definecolor{color1}{rgb}{0.133333333333333,0.545098039215686,0.133333333333333}
\definecolor{color2}{rgb}{0.117647058823529,0.564705882352941,1}
\definecolor{color3}{rgb}{0.580392156862745,0,0.827450980392157}
\definecolor{color4}{rgb}{0.647058823529412,0.164705882352941,0.164705882352941}

\begin{axis}[
ylabel shift = 1pt,
legend columns=1, 
legend cell align={left},
legend style={
	fill opacity=0.8,
	draw opacity=1,
	draw=none,
	text opacity=1,
	at={(1.45,1)},
	line width=3pt,
	anchor=north,
	/tikz/column 2/.style={
		column sep=10pt,
	},
	font=\LARGE
},
legend image post style={line width =3pt},
every tick label/.append style={font=\LARGE}, 
label style={font=\LARGE},
tick align=outside,
tick pos=left,
x grid style={white!69.0196078431373!black},
xlabel={normalized dispersion parameter},
xmin=-0.05, xmax=1.05,
xtick style={color=black},
xtick={-0.2,0,0.2,0.4,0.6,0.8,1,1.2},
xticklabels={−0.2,0.0,0.2,0.4,0.6,0.8,1.0,1.2},
y grid style={white!69.0196078431373!black},
ylabel={normalized swap distance},
ymin=-0.006824, ymax=0.348295111111111,
ytick style={color=black},
ytick={-0.05,0,0.05,0.1,0.15,0.2,0.25,0.3,0.35},
yticklabels={−0.05,0.00,0.05,0.10,0.15,0.20,0.25,0.30,0.35}
]
\addplot [semithick, red!54.5098039215686!black]
table {%
0 0.332153333333333
0.1 0.0791466666666667
0.2 0.0324488888888889
0.3 0.0228066666666667
0.4 0.0227066666666667
0.5 0.0280755555555556
0.6 0.03816
0.7 0.0536977777777778
0.8 0.0813444444444444
0.9 0.140146666666667
1 0.208308888888889
};
\addlegendentry{Seq.-Lo. vs Seq.-Wi.}
\addplot [semithick, color0]
table {%
0 0.216391111111111
0.1 0.0511111111111111
0.2 0.0219711111111111
0.3 0.0157577777777778
0.4 0.0165355555555556
0.5 0.0213666666666667
0.6 0.02948
0.7 0.0422244444444444
0.8 0.0640755555555556
0.9 0.109271111111111
1 0.164742222222222
};
\addlegendentry{Seq.-Lo. vs Score}
\addplot [semithick, color1]
table {%
0 0.187575555555556
0.1 0.0315377777777778
0.2 0.0131666666666667
0.3 0.00996
0.4 0.0112866666666667
0.5 0.0157533333333333
0.6 0.0228311111111111
0.7 0.0332911111111111
0.8 0.0537711111111111
0.9 0.0977466666666667
1 0.149926666666667
};
\addlegendentry{Seq.-Wi. vs Score}
\addplot [semithick, color2, dashed]
table {%
0 0.266906666666667
0.1 0.0489244444444444
0.2 0.0198688888888889
0.3 0.0139466666666667
0.4 0.0143288888888889
0.5 0.0186466666666667
0.6 0.0262755555555556
0.7 0.0398311111111111
0.8 0.0630111111111111
0.9 0.112473333333333
1 0.178855555555556
};
\addlegendentry{Kemeny vs Seq.-Lo.}
\addplot [semithick, color3, dashed]
table {%
0 0.22168
0.1 0.0441066666666667
0.2 0.0186355555555556
0.3 0.0131755555555556
0.4 0.0138955555555556
0.5 0.0187955555555556
0.6 0.0283377777777778
0.7 0.04246
0.8 0.0690333333333333
0.9 0.122655555555556
1 0.193593333333333
};
\addlegendentry{Kemeny vs Score}
\addplot [semithick, color4, dashed]
table {%
0 0.210393333333333
0.1 0.03072
0.2 0.0127133333333333
0.3 0.00931777777777778
0.4 0.0100088888888889
0.5 0.0142111111111111
0.6 0.0220355555555556
0.7 0.0347066666666667
0.8 0.0583422222222222
0.9 0.107962222222222
1 0.173777777777778
};
\addlegendentry{Kemeny vs Seq.-Wi.}
\end{axis}

\end{tikzpicture}}
		\caption{Half}\label{fig:half}
	\end{subfigure}
	\caption{Pairwise average normalized swap distance between the Kemeny ranking and rankings produced by our three scoring-based methods for two different scoring systems on Mallows profiles with $100$ voters and $10$ candidates.}\label{fig:furtherrules}
\end{figure*}
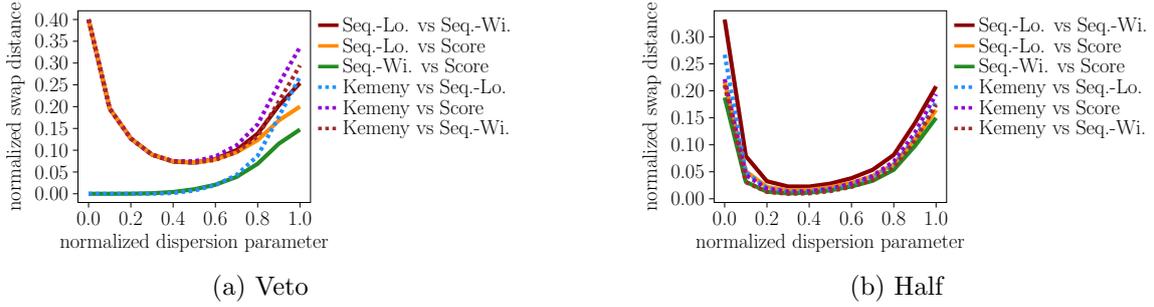

\subsection{Further Voting Rules}
In this section, we briefly examine two additional scoring systems. 
In accordance with our theoretical analysis, we start by examining ranking methods based on the Veto scoring system with scoring vector $(0,\dots, 0,-1)$. 
\Cref{fig:veto} shows the results of our experiment for the Mallows model (again we generated $\num{10000}$ profiles for each $\normphi\in \{0,0.1,\dots, 0.9,1\}$). 
Comparing \Cref{fig:veto} to the analogous plots for Plurality (see \Cref{fig:general-comparsion} and \Cref{fig:Kemeny-comparison}), we see that the results are very similar to each other after swapping the roles of \seqWiAbS and \seqLoAbS: 
For Plurality, \seqLoAb{Plurality} and \Score{Plurality} are closely related, while \seqWiAb{Plurality} produces results most similar to the Kemeny ranking. 
In contrast, for Veto, \seqWiAb{Veto} and \Score{Veto} are closely related and  \seqLoAb{Veto} is most similar to Kemeny. 
The switched role of \seqWiAbS and \seqLoAbS is to be expected, recalling our equivalence \Cref{lem:equiv}. 
(For profiles sampled using the Euclidean model, similar conclusions apply.)

Recalling that for Veto, \seqWiAbS is similar to \ScoreS, whereas for Plurality \seqLoAbS is similar to \ScoreS, we wanted to check the behavior of our rules in between these two extremes. 
Thus, we introduce a new scoring system, which we call \emph{Half}. This scoring system assigns one point to the first  $\floor*{\frac{m}{2}}$ candidates and zero points to all other candidates, where $m$ is the number of candidates.
We depict the results for profiles sampled from the Mallows model in \Cref{fig:half} (again we generated $\num{10000}$ profiles for each $\normphi\in \{0,0.1,\dots, 0.9,1\}$). 
Indeed, in this case \seqWiAbS and \seqLoAbS are both at the same distance to \ScoreS and at the same distance to the Kemeny ranking.
Nevertheless, naturally the rankings produced by \seqWiAb{Half} and \seqLoAb{Half} are still different. 
Remarkably, our three scoring-based ranking methods based on Half produce rankings that are closer to the Kemeny ranking then the rankings produced by any of our methods based on Plurality or Veto. 
This indicates the advantage of allowing voters to distinguish between many candidates to identify candidate strength.

\section{Additional Material for \Cref{sec:complexity-loser}}

\restatehere{parameterm}
\begin{proof}
It remains to prove the correctness of the algorithm described in the main body. 
For this, it is sufficient to prove that the recurrence relation is correct. 
To this end, assume a subset $C'\subseteq C$ of candidates is an elimination set as witnessed by the selected ranking $\succ$ and let $c:=\cand(\succ,|C'|)$. 
Then, as witnessed by $\succ$, $C'\setminus \{c\}$ is an elimination set, and no candidate has a higher $\vs$-score than $c$ after deleting all candidates from $C'\setminus \{c\}$. 
If $T[C']$ is set to true because $C'\setminus \{c\}$ is an elimination set (as witnessed by $\succ$) for some candidate $c\in C'$, then $C'$ is clearly an elimination set as we can eliminate candidates from $C'\setminus c$ in the first $|C'|-1$ rounds (breaking ties according to $\succ$) and $c$ in round $|C'|$.
\end{proof}

\subsection{Plurality}

Before we present our ETH-based lower bound for the parameter $m$, we state some relevant results.

\begin{proposition}
	[\citealp{amiri2021note}, extending results about subcubic vertex cover by \citealp{JS99} and \citealp{Kom18} ]
	\label{prop:eth-cubic-vertex-cover}
	If the Exponential Time Hypothesis (ETH) is true, there does not exist an algorithm solving \textsc{Cubic Vertex Cover} in time $2^{o(n)} \cdot  \mathrm{poly}(n)$, where $n$ is the number of vertices.
\end{proposition}

\begin{corollary}
	\label{prop:eth-regular-clique}
	If the ETH is true, there does not exist an algorithm solving \textsc{Regular Clique} (the \textsc{Clique} problem restricted to graphs where every vertex has the same degree) in time $2^{o(n)} \cdot  \mathrm{poly}(n)$, where $n$ is the number of vertices.
\end{corollary}
\begin{proof}
	Suppose there was such an algorithm. Let $G = (V,E)$ be a cubic graph, and $k$ be a target size for a vertex cover. Then the complement graph $\overline{G}$ is regular. Recall that a set $T \subseteq V$ is a vertex cover in $G$ if and only if $V \setminus T$ is an independent set in $G$ if and only if $V \setminus T$ is a clique in $\overline{G}$. Hence by applying the algorithm to find a clique of size at least $n - k$ in $\overline{G}$, we can find a vertex cover of size at most $k$ in $G$ in time $2^{o(n)} \cdot  \mathrm{poly}(n)$, contradicting \Cref{prop:eth-cubic-vertex-cover}.
\end{proof}

\restatehere{eth}
\begin{proof}
 We reduce from \textsc{Cubic Vertex Cover} 
 (given a graph where each vertex has degree three and an integer~$t$,
 the problem asks whether there is a \emph{vertex cover} of size~$t$,
 that is, a set of $t$~vertices such that each edge is incident to at
 least one of these vertices).
 
 Given a cubic graph~$G$ with $n$~vertices and $3n/2$~edges, and some integer~$t$,
 we will create an equivalent instance of \Winner for \seqLoAb{Plurality}.
 Since our reduction will only use $2n+3n/2+3$ candidates, this implies the claimed ETH-result by \Cref{prop:eth-cubic-vertex-cover}.
 
 \paragraph{Candidates} For each vertex~$i$, we create two \emph{vertex candidates}~$v_i$ and $v'_i$.
 For each edge~$j$, we create one \emph{edge candidate}~$e_j$.
 Moreover, we reference the edge candidates corresponding to edges incident to vertex~$i$ by~$e^i_1$, $e^i_2$, and $e^i_3$.
 Finally we have the candidates~$d$, $w$, and~$q$.
 
 \paragraph{Idea} The idea is as follows.
 In the first $n$~rounds, for each $i \in [n]$ either~$v_i$ or $v'_i$ is eliminated.
 Eliminating~$v_i$ will correspond to selecting the $i$th vertex as part of the vertex cover.
 In round~$n+1$, in order to make candidate~$d$ win the election,
 we must eliminate candidate~$w$, because this is the latest point where
 it has smaller Plurality score than~$d$.
 Eliminating any other candidate that can be eliminated in round~$n+1$ will result in~$w$
 having score larger than the score of~$d$ and this can never change afterwards.
 To be able to eliminate $w$~in round~$n+1$, however, we must ensure that the candidates~$e_j$, $j \in [3n/2]$,
 and candidate~$q$ have at least the same score as~$w$.
 This can only happen if the set $\{i \mid \text{$v_i$ is eliminated before round $n+1$}\}$ is a vertex cover.
 (Whenever some $v_i$~is eliminated, the edge candidate corresponding to edges incident to vertex~$i$ gain
  one point and reach at least the same score as~$w$.)
 Moreover the set $\{i \mid \text{$v'_i$ is eliminated before round $n+1$}\}$ must be of size at least~$n-t$, ensuring that we have only selected $t$ vertices to be part of the vertex cover.
 (Whenever some $v'_i$~is eliminated, candidate~$q$ gains three points and needs in total at least~$3(n-t)$
  additional points to reach at least the same score as~$w$ in round~$n+1$.)
 
 \paragraph{Voters} \quad We have the following voters.
 \begin{alignat*}{4}
  \text{$105n$ voters} \quad & d \succ w \succ \dots & \\
  \text{$99n$ voters}  \quad & w \succ d \succ \dots & \\
  \text{$99n -1$ voters} \quad & e_j \succ w \succ d \succ \dots && \forall j\in [m] \\
  \text{$99n -3(n-t)$ voters} \quad & q \succ w \succ d \succ \dots & \\
  \text{$60n-3$ voters} \quad & v_i \succ v'_i \succ w \succ d \succ \dots && \forall i \in [n] \\
  \text{$1$ voter} \quad & v_i \succ e^i_1 \succ w \succ d \succ \dots && \forall i \in [n] \\
  \text{$1$ voter} \quad & v_i \succ e^i_2 \succ w \succ d \succ \dots && \forall i \in [n] \\
  \text{$1$ voter} \quad & v_i \succ e^i_3 \succ w \succ d \succ \dots && \forall i \in [n] \\
  \text{$60n-3$ voters} \quad & v'_i \succ v_i \succ w \succ d \succ \dots && \forall i \in [n] \\
  \text{$3$ voters} \quad & v'_i \succ q \succ w \succ d \succ \dots &\:\:& \forall i \in [n]
 \end{alignat*}
 
 This completes the construction.
 
\paragraph{Key Observations}
 Observe that in round~$\ell, \ell \in [n]$ either $v_i$ or $v'_i$ for some $i \in [n]$
 must be eliminated, since these candidates have score~$60n$, while every other candidate has score
 at least $99n-3(n-t)>96n$.
 If some vertex candidate is eliminated, then the respective other vertex candidate gains
 $60n-3$~additional points and, hence, will not be eliminated before all edge candidates.
 In round~$n+1$ either some edge candidate~$e_j$ (with score between $99n-1$ and $99n+1$),
 candidate~$q$ (with score between $96n$ and $102n$),
 or candidate~$w$ (with score $99n$) is eliminated.
 If any candidate different from~$w$ is eliminated in round~$n+1$,
 than $w$~gains more than~$95n$ additional points and will finally win the election.
 If candidate~$w$ is eliminated in round~$n+1$, then candidate~$d$ wins the election.
 Independent of whether we have eliminated candidate $w$ in round $n+1$ or we are still in round $n+1$, next, all edge candidates and candidate~$q$ will be eliminated.
 Candidate~$w$ or, if $w$~is eliminated than $d$~receives the votes of all candidates eliminated after round~$n$.
 If $w$~is not eliminated, than $d$~has still score~$105n$ and will be eliminated next.
 Then, the remaining vertex candidates are eliminated and finally either~$w$ or $d$~wins.
 
\paragraph{Correctness}
 We show that $d$~is a winner for \seqLoAb{Plurality} of the constructed profile if and only if
 graph~$G$ contains a vertex cover of size~$t$.
 
 For the ``if''-part, assume that there is a vertex cover of size~$t$.
 Without loss of generality, let the first~$t$ vertices denote such a vertex cover.
 To see that~$d$ is a winner of the election, consider the following elimination order.
 In round $\ell, \ell \in [t]$, eliminate candidate~$v_\ell$.
 In round $\ell, t+1 \le \ell \le n$, eliminate candidate~$v'_\ell$.
 Now, each edge candidate has score at least~$99n$ and also candidate~$q$ has score $99n$.
 Thus, we next eliminate~$w$, then~$q$ and then the edge candidates in an arbitrary order
 that is consistent with their scores (some may have score~$99n$ being covered once,
 some have score~$99n+1$ being covered twice).
 Finally eliminate the remaining vertex candidates (with scores between $120n-6$ and $120n-3$)
 so that only candidate~$d$ remains and wins.
 It is easy to verify that this elimination ordering is indeed consistent with the Plurality
 scores in the respective rounds.
 
 For the ``only if''-part, recall the idea and key observations.
 Assuming~$d$ wins, it must be that~$w$ is eliminated in round~$n+1$.
 To do this, $V^*=\{i \mid \text{$v_i$ is eliminated before round $n+1$}\}$ must be a vertex cover,
 since each edge candidate must have gained at least one additional point in the first $n$ rounds to have a score of $99n$ in round $n+1$ where $w$~is eliminated.
 Moreover, since candidate~$q$ also needs score at least~$99n$ in round $n+1$ to allow $w$~to be eliminated in this round,
 it must hold that $\{i \mid \text{$v'_i$ is eliminated before round $n+1$}\}$ is of size at least~$n-t$.
 Thus, $V^*$~is a vertex cover of size at most~$t$.
\end{proof}

\subsection{Veto}

\restatehere{CoombsNP}
\begin{proof}
	We reduce from \textsc{Regular Clique}, i.e. \textsc{Clique} restricted to regular graphs (where all vertices have the same degree). Our reduction will imply NP-hardness and also the claimed ETH-result by \Cref{prop:eth-regular-clique}. Let $(G, k)$ be an instance of \textsc{Regular Clique}, where $G = (V,E)$ is regular and each vertex has degree $r$. We construct a profile on candidate set $C = \{d, w\} \cup V \cup \{ s_v : v \in V \}$ (the candidates from $\{ s_v : v \in V \}$ act as dummy candidates). The question is whether $d$ is a Coombs winner.
	\begin{alignat*}{4}
		\text{$k(k-2) + r + 1$ voters} &\quad & \cdots \succ d & \\
		\text{$r + 1$ voters} &\quad &  \cdots \succ w & \\
		\text{$k(k-2)$ voters} &\quad & \cdots \succ s_v \succ v & \quad\text{for each $v \in V$} \\
		\text{1 voter} &\quad & \cdots \succ d \succ v & \quad\text{for each $v \in V$} \\
		\text{1 voter} &\quad & \cdots \succ w \succ u \succ v & \quad\text{for each $v \in V$ and each $u \in V$ with $\{v,u\} \in E$}
	\end{alignat*}
	In these votes, ``$\cdots$'' is replaced by all unmentioned candidates according to some common canonical order in which $d$ is ranked first. To avoid doubt, for each edge $\{v,u\} \in E$, we introduce two votes of the bottom type, one with $v \succ u$ and one with $u \succ v$. 
	For convenience (to avoid talking about negative numbers), we say that the \emph{bottom count} of a candidate is the number of times the candidate is ranked in last position. Thus the bottom count is the negative of the veto score, and Coombs proceeds by eliminating candidates with the highest bottom count. 
	Throughout the proof, for a vertex $v \in V$, we write $\text{Nghbhd}(v) = \{u \in V : \{u,v\} \in E\}$ for the neighborhood of $v$ in $G$.
	\\
	
	The intuition behind the construction is that to avoid eliminating $d$ (so as to make $d$ a winner), we need to first eliminate $w$, despite its initially low bottom count. The way to increase its count is to eliminate vertices forming a dense subgraph, since then many `edge' votes (of the bottom type) are transferred to $w$. The rule will start by eliminating vertices. Once some are eliminated, it is then only possible to eliminate vertices which are adjacent to all previously eliminated vertices. Thus, elimination sequences encode cliques. \\
	
	We start by proving the forward direction. Suppose $G$ contains a clique $T = \{ v_1, \dots, v_k \}$ of size $k$. Note that in the initial profile, the bottom counts are:
	\begin{itemize}[itemsep=0pt]
		\item $d$ has count $k(k-2) + r + 1$,
		\item $w$ has count $r + 1$,
		\item every $v \in V$ has count $k(k-2) + r + 1$,
		\item every $s_v$, $v \in V$, has count $0$.
	\end{itemize}
	Thus, we can eliminate $v_1$. After that,
	\begin{itemize}[itemsep=0pt]
		\item $d$ has count $k(k-2) + r + 2$,
		\item $w$ has count $r + 1$,
		\item every $v \in \text{Nghbhd}(v_1)$ has count $k(k-2) + r + 2$,
		\item every $v \not\in \text{Nghbhd}(v_1)$ has count $k(k-2) + r + 1$,
		\item $s_{v_1}$ has count $k(k-2)$,
		\item every $s_v$, $v \neq v_1$, has count $0$.
	\end{itemize}
	Thus, we can eliminate $v_2$ (since it is a neighbor of $v_1$). After this, since we have eliminated both endpoints of the edge $\{ v_1, v_2 \}$, the votes of the two voters corresponding to that edge are transferred to $w$. 
	
	In fact, we can eliminate candidates in the order $v_1, \dots, v_k$. After we have eliminated candidates $v_1, \dots, v_p$, $p \le k$, we have the following situation.
	\begin{itemize}[itemsep=0pt]
		\item $d$ has count $k(k-2) + r + 1 + p$,
		\item $w$ has count $r + 1 + p(p-1)$,
		\item every $v \in V$ has count $k(k-2) + r + 1 + |\{v_1, \dots, v_p\} \cap \text{Nghbhd}(v)|$,
		\item $s_{v_1}, \dots, s_{v_p}$ have count $k(k-2)$,
		\item every other $s_v$ has count $0$.
	\end{itemize}
	In particular, if $p < k$, then $v_{p+1}$ can be eliminated in the next step, since its bottom count is $k(k-2) + r + 1 + p$. After all members of the clique have been eliminated, we reach bottom counts
	\begin{itemize}[itemsep=0pt]
		\item $d$ has count $k(k-2) + r + 1 + k = k(k-1) + r + 1$,
		\item $w$ has count $r + 1 + k(k-1)$,
		\item every $v \in V$ has count $k(k-2) + r + 1 + |\{v_1, \dots, v_k\} \cap \text{Nghbhd}(v)| \le k(k-1) + r + 1$,
		\item $s_{v_1}, \dots, s_{v_k}$ have count $k(k-2)$,
		\item every other $s_v$ has count $0$.
	\end{itemize}
	Hence, $w$ is a candidate with highest bottom count, and thus now be eliminated. All its votes are transferred to the last (uneliminated) candidate in the canonical order, call it $x$. Because previously $w$ had maximum bottom count, $x$ now has maximum bottom count and can be eliminated. This argument applies repeatedly, and hence from now on the remaining candidates are eliminated in the canonical order, finishing with $d$ (which was placed first in the canonical order), so $d$ is a Coombs winner. \\
	
	Conversely, suppose $d$ is a Coombs winner. Thus, there is a way to eliminate candidates so that $d$ is eliminated last. Now, at each of the first $k$ elimination steps (assuming that $d$ is eliminated last) the set of eliminateable candidates is a subset of $\{d\} \cup V$. To see this by induction, note that it holds at the first step. Further, after we have eliminated some candidates $W \subseteq V$ with $|W| = p < k$, writing $e_W$ for the number of edges between vertices in $W$, the resulting bottom counts are
	\begin{itemize}[itemsep=0pt]
		\item $d$ has count $k(k-2) + r + 1 + p$,
		\item $w$ has count $ r + 1 + 2\cdot e_W$,
		\item every $v \in V$ has count $k(k-2) + r + 1 + |W \cap \text{Nghbhd}(v)| $,
		\item for $v \in W$, $s_v$ has count $k(k-2)$,
		\item every other $s_v$ has count $0$.
	\end{itemize}
	Since $2\cdot e_W \le 2 \cdot \binom{p}{2} = p(p-1) < k(k-2) + p$ (since $p < k$), the bottom count of $w$ is smaller than the score of $d$, so the candidate $w$ cannot be eliminated in round $p$; also clearly no candidate $s_v$ can be eliminated. Thus, only a subset of $\{d\} \cup V$ can be eliminated at this step. In fact, we see that a vertex $v \in V \setminus W$ can be eliminated at this step if and only if $v$ is adjacent to all candidates from $W$. It follows that if during the first $k$ steps we do not eliminate $d$, then we eliminate $k$ vertices. These vertices must form a clique in $G$.
\end{proof}

\restatehere{CoombsWhard}
\begin{proof}
	We prove this statement by proving the hardness of an equivalent problem for \seqWiAb{Plurality}: 
	Specifically, we show that it is W[1]-hard parameterized by $n$ to decide whether given a ranking profile $P$ and a designated candidate $d$ there is a selected ranking by \seqWiAb{Plurality} where $d$ is ranked last. 
	Then by applying \Cref{lem:equiv} the theorem follows.

	In some round, we call the elimination of a candidate $c$ \emph{valid} if $c$ is a Plurality winner in the election from this round.
	
	We say that a candidate is \emph{present} in some round if the candidate has not been deleted in some previous round.
	
	For a set $S$ and an element $s\in S$, we write $S-s$ as a shorthand notation for $S\setminus \{s\}$
	
	\paragraph{Construction}
	We prove hardness by a reduction from the W[1]-hard \textsc{Multicolored Independent Set} problem parameterized  by the solution size~$\ell$.
	In \textsc{Multicolored Independent Set}, we are given a $\ell$-partite graph $(V^1 \cupdot V^2 \cupdot \dots \cupdot V^\ell, E)$  and the question is whether there is an independent set~$X$ of size $\ell$ with $X\cap V^j\neq \emptyset$ for 
	all $j\in [\ell]$.
	To simplify notation, we assume that $V^j = \{v^j_1, \dots, v^j_\nu\}$ for all~$j\in [\ell]$.
	We refer to the elements of $[\ell]$ as \emph{colors} and say that a vertex $v$ has
	color~$j\in [\ell]$ if~$v\in V^j$. 
	
	For each $j\in [\ell]$ and $i\in [\nu]$, we introduce two \emph{vertex candidates} $c_i^j$ and $q_i^j$. 
	Moreover, for each edge $e\in E$, we introduce an \emph{edge candidate} $f_e$. Let $F=\cup_{e\in E} f_e$. 
	For $j\in [\ell]$ and $i\in [\nu]$, let $F^j_{i}$ be the set of all edge candidates corresponding to edges incident to $v_i^j$. 
	Moreover, we introduce $\ell$ \emph{blocker candidates} $B=\{b^1,\dots, b^{\ell}\}$. 
	Lastly, we add \emph{dummy candidates} $G=\{g_1,\dots, g_4 \}$ and $T=\{t_1,\dots, t_4\}$, and add the designated candidate $d$. 
	
	For a subset $C'$ of candidates, let $[C']$ be the lexicographic strict ordering of the candidates in $C'$. 
	In particular let $[T]$ be $t_1\succ t_2\succ t_3\succ t_4$, $[B]$ be $b^1\succ \dots \succ b^{\ell}$, and $[B-b^j]$ be $b^1\succ \dots \succ b^{j-1}\succ b^{j+1}\succ \dots \succ  b^{\ell}$.

	We now describe the ranking profile.
	We complete each ranking by appending the remaining candidates in an arbitrary order. 
	First for each color $j\in [\ell]$, we introduce two \emph{color rankings} and refer to them as the \emph{first} and \emph{second} color ranking: 
	\begin{align*}
	c^j_1\succ [F^j_1]\succ q^j_1\succ c^j_2\succ [F^j_2]\succ q^j_2\succ \dots \succ\\ c^j_{\nu}\succ [F^j_{\nu}]\succ q^j_{\nu}\succ b^{j}\succ [B- b^{j}]\succ d\\
	q^j_{\nu}\succ[F^j_{\nu}]\succ c^j_{\nu}\succ q^j_{\nu-1}\succ[F^j_{\nu-1}]\succ c^j_{\nu-1}\succ \dots \succ\\ q^j_{1}\succ[F^j_{1}]\succ c^j_{1}\succ b^{j}\succ [B- b^{j}]\succ d
	\end{align*}
	
	Moreover, we introduce four \emph{global rankings} as depicted in \Cref{fig:Coombs-NP} (note that the last three rankings only ``differ'' in the beginning).
	
	\begin{figure*}
		\resizebox{\linewidth}{!}{
			\begin{minipage}{\linewidth}
				\begin{align*}
					g_1\succ g_2\succ g_3\succ g_4\succ  [T] \succ c_\nu^1\succ c_{\nu-1}^1 \succ \dots \succ c_{1}^1\succ
					c_{\nu}^2\succ \dots \succ  c_{1}^2\succ \dots \succ  c_{\nu}^{\ell}\succ \dots \succ c_{1}^{\ell} \succ [B] \succ d \\
					g_2\succ g_3\succ g_4\succ [T] \succ c_1^1\succ q_1^1\succ c_2^1 \succ q_2^1 \succ \dots \succ c_{\nu}^1\succ q_{\nu}^1\succ c_1^2\succ   q_1^2\succ \dots \succ  c_{\nu}^2\succ q_{\nu}^2\succ \dots \succ  c_1^{\ell}\succ q_1^{\ell}\succ \dots \succ  c_{\nu}^{\ell} \succ q_{\nu}^{\ell} \succ [B]\succ d  \\
					g_3\succ g_4\succ [T] \succ c_1^1\succ q_1^1\succ c_2^1 \succ q_2^1 \succ \dots \succ c_{\nu}^1\succ q_{\nu}^1\succ c_1^2\succ   q_1^2\succ \dots \succ   c_{\nu}^2\succ q_{\nu}^2\succ \dots \succ  c_1^{\ell}\succ q_1^{\ell}\succ \dots \succ  c_{\nu}^{\ell} \succ q_{\nu}^{\ell} \succ [B]\succ d  \\
					g_4\succ [T] \succ c_1^1\succ q_1^1\succ c_2^1 \succ q_2^1 \succ \dots \succ c_{\nu}^1\succ q_{\nu}^1\succ c_1^2\succ   q_1^2\succ \dots \succ c_{\nu}^2\succ q_{\nu}^2\succ \dots \succ  c_1^{\ell}\succ q_1^{\ell}\succ \dots \succ  c_{\nu}^{\ell} \succ q_{\nu}^{\ell} \succ [B]\succ d
				\end{align*}%
			\end{minipage}%
		}
		\caption{Global rankings from the construction for \Cref{th:Coombs-NP}.}\label{fig:Coombs-NP}
	\end{figure*}

	Lastly, we add five \emph{dummy rankings} (these will be the rankings that contribute to the Plurality score of $d$ at some point): 
	\begin{align*}
	&t_1\succ t_2\succ t_3\succ t_4\succ d  \\
	&t_2\succ t_3\succ t_4\succ d  \\
	&t_3\succ t_4\succ d \\
	&t_4\succ d\\
	& d
	\end{align*}

	All rankings are completed arbitrarily.
	The general intuition is that we need to eliminate $g_1$ in some ``early'' round where all candidates still have a Plurality score of at most one, as $g_1$ is only ranked before $d$ in one ranking. 
	The elimination of $g_1$ then triggers immediately some uniquely determined follow-up eliminations (specifically of all candidates in $G\cup T$). 
	After that $d$ has a Plurality score of five. 
	As $d$ needs to be eliminated last, all candidates that are still present after this need to have at least Plurality score five when they are eliminated. 
	As each edge candidate appears only in four rankings before $d$, this implies that all edge candidates must have already been deleted. 
	Moreover, in Statement 5 of \Cref{claim:back}, we prove that for each color $j\in [\ell]$,  there is some $i\in [\nu]$ such that $c_{i_j}^j$ and $q_{i_j}^j$ are still present in the round where $g_1$ is eliminated. 
	As all edge candidates must have been deleted before, from this one can conclude that $\{v^1_{i_1}, \dots, v^{\ell}_{i_{\ell}}\}$ is an independent set. 
	
	\paragraph{Forward Direction}
	For the forward direction, assume that we are given a multicolored independent set $V'=\{v_{i_1}^1,\dots, v_{i_{\ell}}^{\ell}\}$. 
	
	We now describe a valid elimination order of the candidates for which $d$ is eliminated last.
	We proceed in four phases.
	
	In the first phase every candidate has Plurality score at most one in each round and it is thus valid to delete a candidate if it is ranked first in at least one ranking. 
	For each $j\in [\ell]$, we do the following:
	We delete all candidates that are ranked before $c_{i_j}^j$ in the first color ranking for color $j$ in the order in which they appear in this ranking. 
	After that we delete all candidates that are ranked before $q_{i_j}^j$ in the second color ranking for color $j$ in the order in which they appear in this ranking.

	As $V'$ is an independent set, after this, for each $j\in [\ell]$ the color rankings for color $j$ are: 
	\begin{align*}
	&c^j_{i_j}\succ  q^j_{i_j}\succ b^{j}\succ [B- b^{j}]\succ d\\
	&  q^j_{i_j}\succ  c^j_{i_j}\succ b^{j}\succ [B- b^{j}]\succ d
	\end{align*}
	At the end of phase one, every candidate has a Plurality score of at most one. 
	
	Subsequently in the second phase, we start by eliminating candidate $g_1$. 
	After that, for the next few rounds the elimination order is unique because one candidate is the unique Plurality winner. 
	Specifically, afterwards we eliminate $g_2$, then $g_3$, then $g_4$, then $t_1$, then $t_2$, then $t_3$, and lastly $t_4$. 
	Afterwards, $d$ is ranked in the first position in the five dummy rankings. 
	The global rankings became: 
	$$c^1_{i_1}\succ c^2_{i_2} \succ \dots\succ  c^{\ell}_{i_{\ell}}\succ [B]\succ d,$$ 
	and three times: 
	$$c^1_{i_1}\succ q^1_{i_1}\succ c^2_{i_2}\succ q^2_{i_2} \succ \dots\succ c^{\ell}_{i_{\ell}}\succ q^{\ell}_{i_{\ell}}\succ [B]\succ d.$$

	In the third phase, for $j=1$ to $j=\ell$, we first eliminate $c^{j}_{i_{j}}$ and then $q^{j}_{i_{j}}$ (and subsequently increment $j$ by one).
	To argue why all these eliminations are valid observe that in each round in this phase the Plurality score of all candidates is at most five:
	For $d$ this follows directly from the fact that $d$ appears after the blocker candidates in all but five rankings. 
	For the blocker candidates, observe that in each round in the third phase at least one vertex candidate is present. 
	Thus, no blocker candidate is ever ranked in first place in one of the last three global rankings. 
	Ignoring the last three global rankings and all rankings where the blocker candidates appear after $d$, each blocker candidate appears in at most three rankings before all other blocker candidates. 
	Thus, each blocker candidate has at most a Plurality score of three in some round in this phase. 
	To see why vertex candidates have a Plurality score of at most give, fix some $j\in [\ell]$.
	For candidate  $q^{j}_{i_{j}}$ the statement clearly holds, as $q^{j}_{i_{j}}$ only appears in five rankings before $d$. 
	For candidate $c^j_{i_j}$, observe that $c^j_{i_j}$ appears before $d$ in six rankings; however, in one of these rankings $c^j_{i_j}$ is ranked after $q^{j}_{i_{j}}$. 
	As $q^{j}_{i_{j}}$ is eliminated after $c^j_{i_j}$, the Plurality score of $c^j_{i_j}$ is at most five in each round.  
	
	Using that the Plurality score of all candidates in each round in the third phase is upper bounded by five, we now argue why all eliminations are valid. 
	For this, let us examine the situation for $j=1$ (so the first round in the third phase): 
	As $c^1_{i_1}$ is ranked first by the four global rankings and the first color ranking for color $1$, $c^1_{i_1}$ has Plurality score five and eliminating it is valid. 
	Subsequently, $q^1_{i_1}$ is ranked first by the last three global rankings and both color rankings for color $1$. 
	Thus, $q^1_{i_1}$ has Plurality score five and eliminating it is valid. 
	The same argument also applies for increasing $j$, establishing the validity of this phase. 
	After the third phase only blocker candidates and $d$ remain.
	
	In the fourth phase, we eliminate $b^i$ for $i=1$ to $i=\ell$.
	The designated candidate $d$ has a Plurality score of five in each round of the fourth phase. 
	Note that after the third phase, $b^1$ has six Plurality points and is thus the unique Plurality winner. 
	After eliminating $b^1$, $b^2$ has eight Plurality points and is thus the unique Plurality winner. 
	The same reasoning applies until $i=\ell$.
	Afterwards, only $d$ is left, which completes the argument.

	\paragraph{Backward direction}
	For the backward direction, assume that there is an execution of \seqWiAb{Plurality} such that $d$ is eliminated in the last round. 
	We will now reason about the elimination order of the candidates in this execution of \seqWiAb{Plurality}. 
	Let $x$ be the round in which candidate $g_1$ is eliminated. 
	We will now prove a series of claims that are the cornerstone of the proof of correctness: 
	\begin{claim} \label{claim:back}
		Assume that $d$ is eliminated last and let $x$ be the round in which $g_1$ is eliminated. Then, 
		\begin{enumerate}
			\item Every candidate has at most Plurality score one in round $x$. 
			\item All candidates from $G\cup T$ are present in round $x$. Every candidate which is not part of $G\cup T$ and that is present in round $x$ has a Plurality score of at least five in the round in which it is eliminated. 
			\item All candidates from $F$ have been eliminated before round $x$. 
			\item For each $j\in [\ell]$, there are $i,i'\in [\nu]$ such that $c_{i}^j$ and $q_{i'}^j$ are present in round $x$. 
			\item For each $j\in [\ell]$, there is some $i\in [\nu]$ such that both $c_{i}^j$ and $q_{i}^j$ are present in round $x$. 
		\end{enumerate}
	\end{claim}
	\begin{claimproof}
		
		\textbf{Proof of Statement 1.} As $g_1$ is only ranked once before $d$ and $d$ is eliminated last, $g_1$ has a Plurality score of one in round $x$. 
		Thus, all candidates have a Plurality score of at most one in round $x$.

		\textbf{Proof of Statement 2.} 
		We prove the following from which Statement 2 directly follows, as $d$ is the last candidate to be eliminated: 
		In rounds $x$ to $x+8$ exactly the candidates $G\cup T$ are eliminated. The  Plurality score of $d$ is at least five in every round after round $x+8$.

		Assume that no candidate from $G\cup T$ has been eliminated before round $x$. 
		Then, after $g_1$ is eliminated, in round $x+1$ $g_2$ has Plurality score two and all other candidates have a Plurality score of at most one. 
		Thus, $g_2$ will be eliminated. 
		Following this reasoning, $g_3$ will be eliminated in round $x+2$, $g_4$ in round $x+3$, $t_1$ in round $x+4$, $t_2$ in round $x+5$, $t_3$ in round $x+6$, and $t_4$ in round $x+7$. 
		Afterwards $d$ is ranked first in the five dummy rankings. 
		
		For the sake of contradiction, assume that some candidate from $G\cup T$ is eliminated before round $x$. 
		Let $h$ be the first candidate from $G\cup T$ that is eliminated. 
		Then in the round where $h$ is eliminated $h$ has a Plurality score of one, as in all rankings where $h$ is not ranked first and appears before $d$ it is ranked after some candidates from $G\cup T$. 
		As the elimination of $h$ distributes one Plurality point to some other candidate from $G\cup T$, the elimination of $h$ triggers the above described elimination procedure from this candidate onwards ultimately leading to $d$ gaining at least one Plurality point. 
		Accordingly, $d$ has a Plurality score of at least two in round $x$ in this case, a contradiction to Statement 1.

		\textbf{Proof of Statement 3.}
		Assume that for some $e\in E$, $f_e$ is eliminated after round $x-1$. 
		Observe that each edge candidate appears only in four rankings before $d$ (the four color rankings of colors of its endpoints). 
		Thus, before $d$ is eliminated the Plurality score of $f_e$ is at most four, a contradiction to Statement 2 and $d$ being eliminated last.

		\textbf{Proof of Statement 4.} 
		Fix some $j\in [\ell]$.
		We prove the statement in three steps by first excluding that no vertex candidate for color $j$ is present, then excluding that only vertex candidates of the form $q_i^j$ are present, and finally excluding that only vertex candidates of the form $c_i^j$ are present. 
		From these three parts the statement directly follows.
		
		First, assume for the sake of contradiction that for all $i\in [\nu]$, $c_{i}^j$ and $q_{i}^j$ have been eliminated before round $x$. 
		By this and Statement 3, it follows that either some blocker candidate or $d$ is ranked first in the two color rankings for color $j$ in round $x$. 
		However, this implies that this candidate has Plurality score at least two in round $x$, a contradiction to Statement 1. 
		
		Second, assume for the sake of contradiction that for all $i\in [\nu]$ $c_{i}^j$ has been eliminated before round $x$ but that there is some $i'\in [\nu]$ such that $q_{i'}^j$ is present in round $x$.
		
		We make a case distinction based on whether $q_{i'}^j$ is the only present vertex candidate for this color or not: 
		We claim that if there is an $i''\in [\nu]-i'$ such that $q_{i''}^j$ is present in round $x$, then the Plurality score of both $q_{i'}^j$ and $q_{i''}^j$ is at most four in any round before $d$ is eliminated. 
		To see this note that there are only five rankings in which  $q_{i'}^j$ and $q_{i''}^j$ are ranked before $d$ and that in fact these are the same five rankings for both. 
		However, for the two color rankings for color $j$ it holds that in one of them $q_{i'}^j$ is ranked before $q_{i''}^j$ and in the other $q_{i''}^j$ is ranked before $q_{i'}^j$. 
		Thus, as long as both $q_{i'}^j$ and $q_{i''}^j$ are not eliminated, both of them can have at most four Plurality points. 
		As $q_{i'}^j$ and $q_{i''}^j$ are present in round $x$, using Statement 2 it follows that none of $q_{i'}^j$ and $q_{i''}^j$ can reach a Plurality score of five before $d$'s elimination, a contradiction to $d$ being eliminated last. 
		
		Otherwise for all $i'' \in [\nu]-i'$ $q_{i''}^j$ has been eliminated before round $x$ (and by our initial assumption for all $i\in [\nu]$ $c_{i}^j$ has also been eliminated before round $x$). Using Statement 3, it follows that $q_{i'}^j$ is ranked first in both color rankings for color $j$ in round $x$, a contradiction to Statement 1.

		Third, assume for the sake of contradiction that for all $i\in [\nu]$ $q_{i}^j$ has been eliminated before round $x$ but that there is some $i'\in [\nu]$ such that $c_{i'}^j$ has not been eliminated before round $x$.
		We make a case distinction similar as above. 
		We claim that if there is a $i''\in [\nu]-i'$ such that $c_{i''}^j$ is present in round $x$, then both $c_{i'}^j$ and $c_{i''}^j$ have at most four Plurality points in any round before $d$ is eliminated.
		Assume without loss of generality hat $i'<i''$.
		Both $c_{i'}^j$ and $c_{i''}^j$ are only ranked before $d$ in the two color rankings for color $j$ and the four global rankings. 
		However, in four of these six rankings $c_{i'}^j$ is ranked before $c_{i''}^j$ (i.e., the first color ranking for color $j$ and the last three global rankings), while in the other two $c_{i''}^j$ is ranked before $c_{i'}^j$. 
		Thus, as long as $c_{i'}^j$, $c_{i''}^j$, and $d$ are present, both $c_{i'}^j$ and $c_{i''}^j$ have at most four Plurality points. 
		With the help of Statement 2, we again reach a contradiction to $d$ being eliminated last.

		Otherwise for all $i'' \in [\nu]-i'$, $c_{i''}^j$ has been eliminated before round $x$. Using Statement 2, it follows that $c_{i'}^j$ is ranked first in both color rankings for color $j$ in round $x$, a contradiction to Statement 1.

		\textbf{Proof of Statement 5.} 
		Fix some $j\in [\ell]$. 
		Let $i'$ be the smallest $i$ such that $c_i^j$ is present in round $x$ (form Statement 4 we know that such an $i$ needs to exist). 
		Assume for the sake of contradiction that there is some $i''< i'$ such that $q_{i''}^j$ is present in round $x$. 
		We argue that in this case neither $c_{i'}^j$ nor $q_{i''}^j$ can ever reach five Plurality points before $d$ is eliminated, a contradiction to Statement 2 and $d$ being eliminated last: 
		The only rankings in which one of $c_{i'}^j$ and $q_{i''}^j$ appears before $d$ are the two color rankings for color $j$ and the four global rankings. 
		However, in the first color ranking and the last three global rankings $q_{i''}^j$ appears before $c_{i'}^j$, whereas in the other two rankings $c_{i'}^j$ appears before $q_{i''}^j$.
		Thus, as long as the other candidate is present none of the two can get a Plurality score of five. 
		
		Finally, assume for the sake of contradiction that $q_{i'}^j$ is not present in round $x$, implying that it was eliminated in round $y$ for some $y<x$. 
		Together with our above observation that excludes the presence of $q_{i'''}^j$ for all $i'''< i'$ in round $x$ and Statement 4, this implies that there is some $i''>i'$ such that $q_{i''}^j$ is present in round $x$. 
		Clearly, $q_{i'}^j$ was ranked in first position by at least one ranking in round $y$. 
		As all candidates from $G\cup T$ are still present in round $x$ (Statement 2), $q_{i'}^j$ is not ranked first in one of the global rankings in round $y$. 
		The only remaining two rankings where $q_{i'}^j$ appears before $d$ are the two color rankings for color $j$. 
		However, as $c_{i'}^j$ and $q_{i''}^j$ are present in round $y$ and $q_{i'}^j$ is ranked behind $c_{i'}^j$ in the first color ranking and behind $q_{i''}^j$ in the second color ranking for 
		color $j$, $q_{i'}^j$ is also not ranked first in these rankings in round $y$, a contradiction. 
	\end{claimproof}
	
	By Statement $5$ of \Cref{claim:back}, we get that for each color $j\in [\ell]$, there is some $i_j\in [\nu]$ such that $c^j_{i_j}$ and $q^j_{i_j}$ are present in round $x$. 
	We claim that $V'=\{v^1_{i_1},\dots, v^{\ell}_{i_{\ell}}\}$ is an independent set in the given graph. 
	Assume for the sake of contradiction that there are two colors $j\neq j'\in [\ell]$ such that $e=\{v^j_{i_j}, v^{j'}_{i_{j'}}\}\in E$. 
	From Statement $3$ of \Cref{claim:back}, we get that $f_e$ has been eliminated before round $x$. 
	The only rankings where $f_e$ is ranked before $d$ are the four  color rankings for colors $j$ and $j'$. 
	In each of these rankings $f_e$ is ranked between $c^j_{i_j}$ and $q^j_{i_j}$ or between $c^{j'}_{i_{j'}}$ and $q^{j'}_{i_{j'}}$. 
	As all these four candidates are still present in round $x$, it follows that $f_e$ was never ranked in the first position in some ranking before round $x$, a contradiction to $f_e$ being eliminated before round $x$. 
	Thus, $V'$ is an independent set, which is clearly multicolored. 
\end{proof}

\restatehere{vetoXP}
\begin{proof}
	For a ranking profile $P'=(\succ'_1,\dots, \succ'_{n'})$ over $m'$ candidates,  the \emph{bottom-list} $(\cand(\succ_1,m'), \dots ,\cand(\succ_{n'},m'))$ of $P'$ contains for each voter its least preferred candidate.
		
	Assume we are given a set $C$ of $m$ candidate and a ranking profile $P=(\succ_i,\dots, \succ_n)$ of $n$ voters. 
	Knowing for each round in an execution of \seqLo{Veto} the  bottom-list is clearly sufficient for us to reconstruct the selected ranking. 
	Even more, if we only know the bottom-list in some round, then we can  reconstruct which candidates have been deleted in previous rounds. These are exactly the candidates that appear behind the currently bottom candidate of a voter in its original vote (note that they clearly need to be deleted and that not more candidates could have been deleted because they have never appeared in the last place and thus never had the lowest Veto score).
	To formalize this, 
	for a tuple $\mathbf{x}=(x_1,\dots, x_n)\in C^n$, 
	let $D(\mathbf{x})$ be the set of candidates that are ranked behind $x_i$ for $i\in [n]$ in $\succ_i$, i.e., $D(\mathbf{x})=\{c\in C \mid \exists i\in [n]: \rank(x_i,\succ_i)<\rank(c,\succ_i)\}$.
	Intuitively speaking, we need to delete all candidates from $D(\mathbf{x})$ to make $\mathbf{x}$ a bottom-list of our profile. 
	However, in case that $D(\mathbf{x})$ contains some of the candidates appearing in $\mathbf{x}$,  $\mathbf{x}$ will actually not be the bottom list of the resulting profile. 
	Accordingly, we call $\mathbf{x}$ \emph{valid} if $\mathbf{x}$ is the bottom list of $P|_{C\setminus D(\mathbf{x})}$.

	Using this notation, we solve the problem via dynamic programming. 
	For this we introduce a table $T[i,c,c_1,\dots,c_n]$ for $i\in [k]$ and $c,c_1,\dots, c_n\in C$.
	Moreover, we add a dummy cell $T[0,\emptyset, \cand(\succ_1,m), \dots ,\cand(\succ_n,m)]$, which we set to true.
	An entry $T[i,c,c_1,\dots,c_n]$ is true if there is an execution of  \seqLo{Veto} resulting in ranking $\succ$ such that $c$ is ranked in position $i$ in $\succ$ and the bottom list of the profile after round $i$ is  $(c_1,\dots, c_n)$.
	
	For increasing $i\in [k]$, we fill table $T$ by setting $T[i,c,c_1,\dots,c_n]$ to true if $(c_1,\dots, c_n)$  is valid and 	
	there exists candidates $c',c'_1,\dots, c'_n \in C$ such that 
	\begin{itemize}
		\item $D((c_1,\dots, c_n))=D((c'_1,\dots, c'_n))\cup \{c\}$ 
		\item no candidate appears more often than $c$ among $c'_1,\dots, c'_n$, and 
		\item $T[i-1,c',c'_1,\dots, c'_n]$ is true.
	\end{itemize}

	After $T$ is filled, we can simply check whether there are candidates $c_1,\dots,c_n\in C$ such that $T[k,d,c_1,\dots, c_n]$ is true in which case we return yes and no otherwise. 	
\end{proof}

\subsection{Borda}
We now prove that \Winner for \seqLo{Borda} is NP-hard, even for a constant number of voters. We begin with a useful observation for reasoning about \seqLo{Borda}.

\begin{remark}
	[Weighted majority graph, C2-Borda scores]
	\label{rem:c2-borda}
	Every ranking profile induces a \emph{weighted majority graph} (aka the C2-graph) which is an edge-weighted directed graph whose vertex set is the set of candidates, and for $c,d \in C$, the weight of the edge $c \to d$ is $w_{cd} = |\{i \in N : c \succ_i d\}| - |\{i \in N : d \succ_i c\}|$.
	Given a ranking profile and its induced weighted majority graph, the \emph{C2-Borda score} of an alternative $c \in C$ is $\sum_{d \in C \setminus \{c\}} w_{cd}$. 
	It is well-known that the Borda score of an alternative is an affine transformation of its C2-Borda score (indeed, the C2-Borda score is equivalent to the difference between the candidates' Borda score and the average Borda score of candidates).
	Hence, to obtain the output of \seqLo{Borda} at a given profile, we only need to know the profile's weighted majority graph.
\end{remark}

In the construction of our reduction we build a weighted majority graph to reason about the scores of candidates.
Using well-known results of \citet{McGa53a} and \citet{Debo87a}, any arc-weighted digraph (in which all arc weights have the same parity; in our reductions we only ever use weights 0 and 2) can be realized as the weighted majority graph of a ranking profile, and this profile can be constructed in polynomial time.
In some of our reductions, we will need to prove that a particular digraph can be constructed using a small number of rankings. For this, we use the following lemma: 
\begin{lemma}[\citealp{ErMo64a}, \citealp{bachmeier2019}]
	\label{lem:bilevel}
	We call an edge-weighted digraph $G = (V,A)$ \emph{bilevel} if we can partition its vertex set as $V = (C_1 \cup \cdots \cup C_s) \cup (D_1 \cup \cdots \cup D_s)$, where all subsets are pairwise disjoint but some of them may be empty, such that
	\[
	A = (C_1 \times D_1) \cup \cdots \cup (C_s \times D_s)
	\]
	with all arcs having weight 2 (see \Cref{fig:bilevel}).
	
	If $G$ is bilevel, then it can be induced as the weighted majority graph of a 2-voter profile.
\end{lemma}
\begin{figure}[t]
	\centering
	\begin{tikzpicture}
	[vertex/.style={circle, fill=black, inner sep=1.5pt},
	bracepath/.style={decorate, decoration = {brace}},
	mirror/.style={decoration = {mirror}},
	setname/.style={transform shape, yscale=0.8, xscale=1.23},
	xscale=0.65]
	\node [vertex] at (0.5,0) (a) {};
	\node [vertex] at (1.5,0) (b) {};
	\node [vertex] at (2.5,0) (c) {};
	
	\node [vertex] at (4,0) (e) {};
	\node [vertex] at (5,0) (f) {};
	
	\node [vertex] at (6.5,0) (g) {};
	
	\node [vertex] at (8,0) {};
	\node [vertex] at (9,0) {};
	\node [vertex] at (10,0) {};
	
	\node [vertex] at (0,-1) (1) {};
	\node [vertex] at (1,-1) (2) {};
	\node [vertex] at (2,-1) (3) {};
	\node [vertex] at (3,-1) (4) {};
	
	\node [vertex] at (4.5,-1) (5) {};
	
	\node [vertex] at (6.5,-1) (6) {};
	
	\draw[-latex, thin]
	(a) edge (1) edge (2) edge (3) edge (4)
	(b) edge (1) edge (2) edge (3) edge (4)
	(c) edge (1) edge (2) edge (3) edge (4)
	
	(e) edge (5)
	(f) edge (5)
	
	(g) edge (6)
	;
	
	\node [setname] at (1.5, 0.6) {$C_1$};
	\node [setname] at (4.5, 0.6) {$C_2$};
	\node [setname] at (6.5, 0.6) {$C_3$};
	\node [setname] at (9.0, 0.6) {$C_4$};
	
	\draw [bracepath] (0.2, 0.25) -- (2.8, 0.25);
	\draw [bracepath] (3.7, 0.25) -- (5.3, 0.25);
	\draw [bracepath] (6.2, 0.25) -- (6.8, 0.25);
	\draw [bracepath] (7.7, 0.25) -- (10.3, 0.25);
	
	\node [setname] at (1.5, -1.6) {$D_1$};
	\node [setname] at (4.5, -1.6) {$D_2$};
	\node [setname] at (6.5, -1.6) {$D_3$};
	\node [setname] at (9.0, -1.6) {$D_4 = \emptyset$};
	
	\draw [bracepath,mirror] (-0.3, -1.25) -- (3.3, -1.25);
	\draw [bracepath,mirror] (4.2, -1.25) -- (4.8, -1.25);
	\draw [bracepath,mirror] (6.2, -1.25) -- (6.8, -1.25);
	\end{tikzpicture}
	\caption{A bilevel graph as described in \Cref{lem:bilevel}.}
	\label{fig:bilevel}
\end{figure}
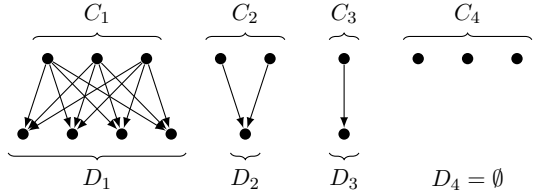
\begin{proof}
	Consider the sets $C_1, \dots, C_s, D_1, \dots, D_s$ as sets that are linearly ordered in some arbitrary way. Then construct the following two rankings $\succ_1$ and $\succ_2$:
	\begin{align*}
	&C_1 \succ_1 D_1 \succ_1 C_2 \succ_1 D_2 \succ_1 \dots \succ_1 C_s \succ_1 D_s, \\
	&\rev(C_s) \succ_2 \rev(D_s) \succ_2 \dots \succ_ 2 \rev(C_1) \succ_2 \rev(D_1). 
	\end{align*}
	It is easy to check that these two voters induce $D$ as their weighted majority graph.
\end{proof} 

We are now ready to prove the theorem.

\restatehere{baldwinNP}
\begin{proof}
	We will give the proof for the case $n = 8$. It can be extended to any larger even number of voters by repeatedly adding 2 opposite rankings (one the reverse of the other) to the profile constructed in the reduction. Adding opposite rankings does not change the induced weighted majority scores. Thus, the C2-Borda scores of the candidates do not change, and hence this does not change the result of \seqLo{Borda}.
	
	We reduce from \textsc{Cubic Vertex Cover}. (Then the ETH-based claim follows using \Cref{prop:eth-cubic-vertex-cover}.) Let $G=(V,E)$ be a graph with $q$ vertices where each vertex $v \in V$ is incident to exactly 3 edges, and let $t$ be the target vertex cover size. 
	We construct an instance of the \Winner problem as follows.
	
	The candidate set consists of one candidate for each vertex, one candidate for each edge, a designated candidate $d$, and dummy candidates $B = \{b_1, b_2, b_3, b_4\}$, $F = \{f_1, \dots, f_{q-t+8}\}$, $G = \{g_1\}$, $H = \{h_1, h_2,h_3, h_4\}$, and $K = \{k_1, \dots, k_5\}$.
	Let $C = V \cup E \cup \{d\} \cup B \cup F \cup G\cup H \cup K$. 

	The ranking profile will be constructed so as to induce a desired weighted majority graph, where all arcs will have weight 2. The arcs are as follows:
	\begin{align*}
		A = &\{ (e, v) \in E \times V : v\in e \}  \cup (B \times V) \cup (G \times E) \cup  {}  \\
		& (H \times E) \cup (F \times B) \cup (K \times H) \cup (H \times \{d\})
	\end{align*}
	
	\begin{figure}
		\centering
		\begin{tikzpicture}
			[multinode/.style={circle,draw=black,inner sep=1.5pt,minimum width=10pt},
			element/.style={circle,draw=black,inner sep=1.5pt,minimum width=4pt},
			set/.style={rectangle,draw=black,inner sep=3pt,minimum width=4pt},
			every edge quotes/.style={fill=white,inner sep=0.5pt},
			every label/.style={color=red!60!black}]
			\node [label={[label distance=-3pt]90:$-6$}] (e) {$e$};
			\node [below left =0.6cm and 0.6cm of e] (v1) {$v_1$};
			\node [label={[label distance=-6pt]30:$-14$},below right=0.6cm and 0.6cm of e] (v2) {$v_2$};
			\node [multinode, label={[label distance=-2pt]20:$2t-16$}, above=0.8cm of e] (B) {$B$};
			\node [multinode, label={[label distance=-2pt]20:$4$}, above=of B] (F) {$F$};
			\node [multinode, label={[label distance=-2pt]20:$2|E|$},left=1.5cm of B] (G) {$G$};
			\node [multinode, label={[label distance=-2pt]20:$2|E|-8$},right=1.5cm of B] (H) {$H$};
			\node [label={[label distance=-4pt]20:$-8$},right= 2.5cm of e] (d) {$d$};
			\node [multinode, label={[label distance=-2pt]20:$8$},above=of H] (K) {$K$};
			
			\draw[-latex]
			(e) edge (v1)
			edge (v2)
			(B) edge [bend right=10] (v1)
			edge [bend left=10] (v2)
			(F) edge (B)
			(G) edge (e)
			(H) edge (e)
			(H) edge (d)
			(K) edge (H)
			;
		\end{tikzpicture}
		\caption{An illustration of the reduction of \Cref{th:Baldwin-NP}. All arcs have weight 2. Red superscripts denote the difference between the weight of outgoing and ingoing arcs for a candidate.} \label{fi:Baldwin-NP}
	\end{figure}
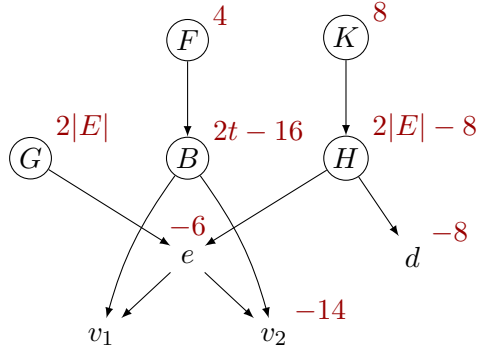

	The constructed weighted majority graph is depicted in \Cref{fi:Baldwin-NP}. 
	We now describe how to write the weighted majority graph $(C,A)$ as a sum of 4 bilevel graphs (as defined in \Cref{lem:bilevel}). 
	For each vertex $v$, we  label the three edges incident to it arbitrarily as $e_v^1$, $e_v^2$, $e_v^3$. Consider the following arc sets:
	\begin{align*}
		A_1 &= (B \times V) \cup ((H \cup G) \times E), \\
		A_2 &= \{ (e_v^1, v) : v \in V \} \cup  (F \times B), \\
		A_3 &= \{ (e_v^2, v) : v \in V \} \cup (K \times H),\\
		A_4 &= \{ (e_v^3, v) : v \in V \} \cup (H \times \{d\}).
	\end{align*}
	It is clear that each of these sets describe bilevel graphs, that they are pairwise disjoint, and that $A = A_1 \cup A_2 \cup A_3 \cup A_4$.
	By invoking \Cref{lem:bilevel}, we get a profile $P$ containing $8$ voters with the depicted weighted majority graph.  
	
	The C2-Borda scores (see \Cref{rem:c2-borda}) in this profile are:
	\begin{itemize}
		\item $d$ has score $-8$ (since it is beaten by 4 $H$-candidates)
		\item each $b \in B$ has score $2t-16$
		\item each $v \in V$ has score $-14$ (since $v$ is beaten by 3 edge candidates, and by 4 $B$-candidates)
		\item each $e \in E$ has score $-6$ (since $e$ beats 2 vertex candidates but is beaten by 4 $H$-candidates and 1 $G$-candidate).
		\item each $h \in H$ has score $2|E|-8$ (since it beats $|E|$ edge candidates and $d$, but is beaten by 5 $K$-candidates)
		\item Candidates in $F$, $G$, and $K$ will always have non-negative scores (since they are beaten by no candidates)
	\end{itemize}
	
	Suppose that $T = \{v_1, \dots, v_t\}$ is a vertex cover of $G$.
	Then the following is a valid elimination ordering according to \seqLoAb{Borda}, where $d$ is eliminated last and thus ranked first in the corresponding selected ranking.
	\begin{itemize}
		\item Eliminate $T$: In the first $t$ rounds, vertex candidates have the lowest C2-Borda score ($-14$), so we can eliminate $T$ in an arbitrary ordering. Each time we eliminate a vertex candidate, the score of the $B$-candidates goes down by $2$, and the score of all incident edge candidates also goes down by $2$.
		\item Eliminate $B$: Starting in round $t+1$, candidates in $B$ have the uniquely lowest C2-Borda score of $-16$, so we can eliminate $B$ in an arbitrary ordering. As we do so, the scores of the remaining vertex candidates go up.
		\item Eliminate $E$: The remaining vertex candidates in $V \setminus T$ currently have score $-6$ because $B$ has been eliminated. Edge candidates have a C2-Borda score of either $-8$ or $-10$ depending on whether $T$ contains one or both of the endpoints of the edge (note that for each edge at least one endpoint has been deleted because $T$ is a vertex cover). These are the lowest C2-Borda scores, so we can eliminate all edge candidates (in an arbitrary ordering except that the $-10$ edges get eliminated first). As we eliminate edge candidates, the scores of $H$-candidates go down.
		\item Eliminate $H$: Just after we finish eliminating $E$, the score of $H$ has dropped to $-8$, which is the lowest C2-Borda score. So we can eliminate $H$ in an arbitrary order.
		\item Eliminate $F \cup G \cup K$: At this point, all remaining candidates have C2-Borda score $0$ (and they have no arcs between them, so eliminating candidates does not change the scores). So we can eliminate all candidate except $d$ in an arbitrary order.
		\item Eliminate $d$.
	\end{itemize}
	
	Conversely, suppose there is a ranking selected by \seqWiAb{Borda} where $d$ is eliminated last. In the first $t-1$ rounds, only vertex candidates can be eliminated since only they have the lowest C2-Borda score of $-14$. But then, in round $t$, the $B$-candidates also have score $-14$. We distinguish two cases: whether another vertex is eliminated, or whether a $B$-candidate is eliminated.
	\begin{itemize}
		\item Case 1: In round $t$, a vertex candidate is eliminated. Define $T$ to be the set of vertices whose candidates were eliminated in the first $t$ rounds. Starting in round $t+1$, the $B$-candidates have score $-16$, which is uniquely lowest, so they will then all be consecutively eliminated.
		\item Case 2: In round $t$, a $B$-candidate is eliminated. Then the C2-Borda score of the remaining vertex candidates goes up to $-12$ while the remaining $B$-candidates have score $-14$, which makes them the candidates with the uniquely lowest C2-Borda score. Thus, in the following rounds, all the $B$-candidates will be consecutively eliminated. Define $T$ to be the set of vertices whose candidates were eliminated in the first $t-1$ rounds.
	\end{itemize}
	In either of the two cases, we have now reached a stage where for a set $T$ of at most $t$ vertices (to be precise, either $t$ or $t-1$ vertices), the corresponding vertex candidates have been eliminated, followed by the elimination of set $B$. We will prove that $T$ is a vertex cover.
	
	After the above-described eliminations, since $B$ is eliminated, the remaining vertex candidates in $V \setminus T$ have score $-6$. Candidate $d$ has score $-8$. Edge candidates have score $-6$ if neither of their endpoints were contained in $T$, and otherwise they have a score of $-8$ or $-10$. All other candidates have score $-6$ or higher. Thus, in the next rounds, edges that were covered have their candidates eliminated. While this happens, the score of $H$-candidates goes down, but they do not become eliminateable before all covered edge candidates are eliminated. If in this way all edge candidates are eliminated, then $T$ was a vertex cover and we are done. If there is an edge that is not covered by $T$, then after all covered edge candidates are eliminated, we end up in a situation where $d$ has score $-8$, but all other candidates have score $-6$ or higher, so $d$ would need to be eliminated next, a contradiction.
\end{proof}

\section{Additional Material for \Cref{sec:complexity-winner}} \label{app:complexity-winner}

We start by observing that for each scoring system  $\vs$ our general problem \PositionK is in XP, as we can simply guess which candidates are ranked on the first $k$ positions (in which ordering) in the selected ranking  and then verify whether it gives rise to a valid execution of \seqWi{$\vs$}.
\begin{observation}
	For every scoring system $\vs$, \PositionK for \seqWi{$\vs$} is in XP. 
\end{observation}

\subsection{Plurality}

In this section, we consider \seqWiAb{Plurality}. 
We start by showing that \TopK is NP-hard and W[1]-hard with respect to $k$. 

\begin{proposition}
	\TopK for \seqWi{Plurality} is NP-hard and W[1]-hard with respect to $k$.
\end{proposition}
\begin{proof}
	We reduce from \textsc{Independent Set}, which is W[1]-hard when parameterized by the solution size.
	Given a graph $G=(V,E)$ with $|V|=\nu$, and an integer $\ell$, \textsc{Independent Set} asks whether there is an independent set of size $\ell$ in $G$ (we assume without loss of generality that $\nu>2$, $\ell>2$, $|E|>1$, and $\ell<\nu$). 
	From an instance of \textsc{Independent Set}, we construct an instance of \TopK as follows.
	We add a candidate $c_v$ for each vertex $v\in V$. 
	Moreover, we introduce $2\nu^3(\ell-1)+\nu^3$ dummy candidates, a designated candidate $d$, a blocker candidate $b$, and an edge candidate $e$. 
	
	We now turn to the description of the ranking profile.
	We first add the following rankings:
	\begin{alignat*}{4}
	2\nu^2\text{ voters} \quad & c_v \succ b \succ d  \succ \dots  && \forall v\in V \\
	2\nu^2\text{ voters} \quad & c_v \succ e \succ \dots && \forall v\in V \\
	1 \text{ voter} \quad & c_v \succ b \succ  c_w \succ e \succ \dots \quad && \forall \{v,w\}\in E 
	\end{alignat*}
	
	Moreover, for each $v\in V$, as long as $c_v$ has less than $2\nu^2(\ell-1)+\nu^2$ Plurality points in the current profile, we add a ranking where $c_v$ is ranked first and some so far never in second place appearing dummy candidate is ranked second. 
	We set $k:=\ell+2$. 
	
	Note that dummy candidates will clearly not be eliminated in the first $k$ rounds so we can ignore them.

	Assume that $V=\{v_1,\dots, v_{\ell}\}$ is an independent set. 
	Then, we eliminate in the first $\ell-1$ rounds candidates $c_{v_1}, \dots, c_{v_{\ell-1}}$. 
	Note that by this eliminations only the Plurality scores of dummy candidates, of $b$, and of $e$ have changed. 
	After round $\ell-1$, $e$ has score $2\nu^2(\ell-1)$. 
	Moreover, $b$ has score at most	$2\nu^2(\ell-1)+(\ell-1)\cdot \nu$. 
	Thus, the vertex candidates still have the highest score and we can eliminate $c_{v_{\ell}}$. 
	After round $\ell$, candidate $b$ clearly has the highest Plurality score, so we eliminate $b$. 
	The elimination of $b$ redistributes $2\ell\nu^2$ points to $d$, at most $\nu$ points to each vertex candidates, and no point to $e$ (because $V$ is an independent set and thus in each of the voters of the third type at least one vertex candidate ranked before $e$ is still present).
	Thus, we can eliminate $d$ in the next round $\ell+2$. 
	
	Conversely, assume that there is an execution of \seqWiAb{Plurality} such that $d$ is eliminated in round $\ell+2$ or before. 
	As argued before clearly in the first $\ell$ rounds vertex candidates need to be deleted. 
	Let $V'\subseteq V$ be the subset of vertices corresponding to these vertex candidates. 
	We claim that $V'$ needs to be an independent set. 
	As argued above, $b$ will be eliminated in round $\ell+1$. 
	Thus, $d$ has a Plurality score of $2\ell\nu^2$ in round $\ell+2$. 
	In round $\ell+2$, candidate $e$ is also ranked in the first position in the  $2\ell\nu^2$
	votes of the second type. 
	Moreover, if there is an edge $\{v,u\}\in E$ with $v,u\in V'$, then we have eliminated all candidates ranked before $e$ in the corresponding vote, giving $e$ a Plurality score of at least  $2\ell\nu^2+1$. 
	This leads to a contradiction, as this implies that $e$ has a higher score than $d$ in round $\ell+2$.
\end{proof}

Moreover, also when we parameterized by the number of $n$ of voters we still get W[1]-hardness by reusing some ideas for the reduction for \Winner for \seqLoAb{Veto} from \Cref{th:Coombs-NP}.

\begin{theorem}
	\TopK for \seqWi{Plurality} is W[1]-h. with respect to the number $n$ of voters. 
\end{theorem}
\begin{proof}
	We prove hardness by a reduction from the W[1]-hard \textsc{Multicolored Independent Set} problem parameterized  by the solution size~$\ell$.
	\paragraph{Construction} 
	In \textsc{Multicolored Independent Set}, we are given an $\ell$-partite graph $(V^1 \cupdot V^2 \cupdot \dots \cupdot V^{\ell}, E)$  and the question is whether there is an independent set~$X$ of size $\ell$ with $X\cap V^j\neq \emptyset$ for 
	all $j\in [\ell]$.
	To simplify notation, we assume that $V^j = \{v^j_1, \dots, v^j_\nu\}$ for all~$j\in [\ell]$.
	We refer to the elements of $[\ell]$ as \emph{colors} and say that a vertex $v$ has
	color~$j\in [\ell]$ if~$v\in V^j$. 
	Moreover, let $|E|=\mu$.
	
	We construct an instance of our problem by setting $k:=\mu+\ell\cdot (\nu-1)+1$. 
	We start by describing the candidate set. 
	For each $j\in [\ell]$ and $i\in [\nu+1]$, we introduce a \emph{vertex candidates} $c_i^j$ (notably candidates $c_{\nu+1}^j$ do not correspond to a vertex but act more like a dummy candidate). 
	Moreover, for each edge $e\in E$, we introduce an \emph{edge candidate} $f_e$. 
	For $j\in [\ell]$ and $i\in [\nu]$, let $F^j_{i}$ be the set of all edge candidates corresponding to edges incident to $v_i^j$. 
	Further, for $j\neq j' \in [\ell]$, let $F^{j,j'}$ be the set of all edge candidates corresponding to edges between a vertex of color $j$ and a vertex of color $j'$ (note that there are no edges between vertices of the same color). 
	Moreover, we introduce $k$ \emph{blocker candidates} $B=\{b_1,\dots, b_{k}\}$. 
	Lastly, we add our designated candidate $d$. 
	For a subset $C'$ of candidates, let $[C']$ be the lexicographic strict ordering of the candidates in $C'$. 
	
	Turning to the input rankings, we introduce for each color $j\in [\ell]$, $\binom{\ell}{2}$ copies of two types of \emph{color rankings}: 
		\begin{align*}
			&c^j_1\succ [F^j_1]\succ c^j_2\succ [F^j_2]\succ \dots \succ c^j_{\nu}\succ [F^j_{\nu}]\succ c^j_{\nu+1}\succ [B]\succ d,
			&&\forall t\in [\textstyle\binom{\ell}{2}] \\
			&c^j_{\nu+1}\succ [F^j_{\nu}]  \succ c^j_{\nu} \succ [F^j_{\nu-1}]  \succ c^j_{\nu-1} \succ \dots \succ [F^j_1] \succ c_j^1 \succ [B]\succ d. 
			&&\forall t\in [\textstyle\binom{\ell}{2}]
	\end{align*}
	
	Moreover, for each pair of colors $j\neq j'\in[\ell]$, we introduce the following \emph{critical rankings}: 
	\begin{align*}
	[F^{j,j'}]\succ d\succ [B].
	\end{align*}
	We complete all rankings arbitrarily. 
	\paragraph{Forward Direction}
	For the forward direction, assume that we are given a multicolored independent set $V'=\{v_{i_1}^1,\dots, v_{i_\ell}^{\ell}\}$. 
	
	We now describe a valid elimination order of the candidates for which $d$ is eliminated in round $k$.
	For $j=1$ to $j=\ell$, we do the following:
	We delete all candidates that are ranked before $c_{i_j}^j$ in the first type of color rankings for color $j$ in the order in which they appear in this type of ranking.
	Afterwards, we delete all candidates that are ranked before $c_{i_j+1}^j$ in the second type of color ranking for color $j$ in the order in which they appear in this type of ranking. 
	Thus, all candidates ranked before the blocker candidates in the color rankings, except $c_{i_j}^j$, $c_{i_j+1}^j$, and candidates from $F_{i_j}^j$ got deleted. 
	Thus, as $V'$ is an independent set, all edge candidates got deleted.
	As for each color two vertex candidates remain, this implies that $\mu+\ell\cdot (\nu-1)=k-1$ candidates got deleted.
	The resulting profile looks as follows: 
	\begin{align*}
	&c^j_{i_j}\succ  c^j_{i_j+1}\succ  [B]\succ d, \quad \forall j\in [\ell], t\in [\textstyle\binom{\ell}{2}]\\
	&c^j_{i_j+1}\succ  c^j_{i_j}\succ  [B]\succ d, \quad \forall j\in [\ell],t\in [\textstyle\binom{\ell}{2}]\\
	&d\succ [B], \quad \forall j\neq j' \in [\ell]
	\end{align*}
	Notably, $d$ is a Plurality winner in this profile so we can eliminate $d$ in round $k$. 
	
	\paragraph{Backward Direction}
	For the backward direction, assume that there is an execution of \seqWiAb{Plurality} for which $d$ is ranked in one of the first $k$ positions in the selected ranking.
	Let $k^*\leq k$ be the round in which $d$ is eliminated.
	Notably, as there exist $\binom{\ell}{2}$ identical rankings, the Plurality winner in each round, needs to have Plurality score at least $\binom{\ell}{2}$. 
	Moreover, as $d$ only appears in the $\binom{\ell}{2}$ critical rankings in one of the first $k^*$ positions, this implies that $d$ needs to be ranked first in all critical rankings in round $k^*$ and that all candidates have a Plurality of score at most $\binom{\ell}{2}$ in round $k^*$.
	From the first part of this observation it follows that all edge candidates need to be deleted before round $k^*$. 
	Using this and that all candidates have Plurality score at most $\binom{\ell}{2}$ in round $k^*$, we get that for each color $j\in [\ell]$, at least two vertex candidates corresponding to vertices of this color are present in round $k^*$: 
	If there is a color with no remaining vertex candidates in round $k^*$, then some blocker candidate will be ranked in the first position in the $2\binom{\ell}{2}$ rankings corresponding to this color in round $k^*$. 
	If there is a color with only one vertex candidate from this color remaining, then this candidate is ranked first in the first position in the $2\binom{\ell}{2}$ rankings corresponding to this color in round $k^*$.
	
	Now, for each color $j\in [\ell]$, let $i_j$ be the smallest $i$ such that $c^j_{i}$ is present in round $k^*$  and let $t_j\neq i_j$ be some other index such that $c^j_{t_j}$ is present in round $k^*$ (by our above observation both of these need to exist).
	We claim that $\{v^1_{i_1}, \dots , v^{\ell}_{i_{\ell}}\}$ is an independent set in the given graph. 
	Assume for the sake of contradiction that there are $j\neq j'\in [\ell]$ with $\{v_{i_j}^{j}, v_{i_{j'}}^{j'}\}=e\in E$. 
	We claim that in this case $f_e$ has not been eliminated before round $k^*$: 
	Recall that a candidate can only get eliminated if its Plurality score is at least $\binom{\ell}{2}$. 
	Moreover, as $f_e$ only appears in one of the first $k^*$ positions in one of the critical rankings, this means that in the round in which $f_e$ is eliminated it needs to be ranked in the first position in one of the color rankings for either color $j$ or $j'$. 
	However, note that in all of these rankings $f_e$ is either ranked between $c^j_{i_j}$ and $c^j_{t_j}$ or between $c^{j'}_{i_{j'}}$ and $c^{j'}_{t_{j'}}$. As all these four candidates are still present in round $k^*$ (and thus also in all previous rounds), it follows that $f_e$ was never ranked in the first position in one of the color rankings, a contradiction to all edge candidates being deleted before round $k^*$.  
\end{proof}

From \Cref{th:Lo-Ve-K-n}, using the equivalence between \seqLoAbS and \seqWiAbS from \Cref{lem:equiv}, we can directly conclude the following:
\begin{corollary}
	\PositionK for \seqWi{Plurality} is in XP with respect to the number $n$ of voters. 
\end{corollary}

Moreover, if we combine the two parameter $n$ and $k$, for each of which we have proven W[1]-hardness, then we can obtain fixed-parameter tractability. 
The core observation here is that in the first $k$ rounds of an execution of \seqWiAb{Plurality} at most $n\cdot k$ candidates can have a non-zero Plurality in one of the rounds (these are the candidates which are ranked in one of the first $k$ positions of some voter). 
All other candidates will be eliminated immediately (in some arbitrary order). 
Subsequently we can apply our algorithm from \Cref{th:parameter-m}. 
\begin{observation}
\PositionK for \seqWi{Plurality} is solvable in $\mathcal{O}(2^{nk}\cdot nm^2)$ time.
\end{observation}

\subsection{Veto}
We now turn to \seqWiAb{Veto} and start by proving that \TopK is NP-hard and W[2]-hard with respect to $k$ for this rule. 

\begin{proposition}
	\TopK for \seqWi{Veto} is NP-hard and W[2]-hard with respect to $k$. 
\end{proposition}
\begin{proof}
	We reduce from \textsc{Hitting Set}, where given a universe $U$ and a family of sets $\mathcal{S}$ and an integer $\ell$, the question is whether there is an $\ell$-subset of the universe containing at least one element from each set from $\mathcal{S}$, i.e., $U'\subseteq U$ with $|U'|=\ell$ and $S\cap U'\neq \emptyset$ for all $S\in \mathcal{S}$.
	\textsc{Hitting Set} is W[2]-hard when parameterized by $\ell$. 
	Let $\nu:=|U|$ and $\mu:=|\mathcal{S}|$. 
	For an element $u\in U$, let $\mathcal{S}_u$ denote the family of sets in which $u$ appears. 
	That is, $\mathcal{S}_u=\{S\in \mathcal{S}\mid u\in S\}$. 
	
	We construct an instance of \TopK as follows. 
	For each element $u\in U$, we introduce an \emph{element candidate} $c_u$. 
	For each set $S\in \mathcal{S}$, we introduce a \emph{set candidate} $e_S$. 
	Lastly, we introduce a \emph{blocker candidate} $b$ and the designated candidate $d$. 
	We set $k:=\ell+1$. 
	
	For convenience (to avoid talking about negative numbers), we say that the \emph{bottom count} of a candidate is the number of times the candidate is ranked in last position. Thus the bottom count is the negative of the veto score, and \seqWi{Veto} proceeds by eliminating candidates with the lowest bottom count. 
	
	We now turn to the description of the ranking profile. 
	We first add the following rankings: 
	\begin{align*}
	\dots \succ b \succ e_S \succ c_u, \qquad &\forall u\in U \text{ and } S\in \mathcal{S}_u \\
	\dots \succ b \succ c_{u'} \succ c_u, \qquad &\forall u\in U \text{ and } u'\in U \setminus \{u\}\\
	\dots \succ b \succ e_S, \qquad & \forall S \in \mathcal{S} \text{ and } i\in [\nu+\mu+\ell-1]\\
	\dots \succ d, \qquad & \forall i\in [\nu+\mu+\ell]\\
	\dots \succ b, \qquad & \forall i\in [\nu+\mu+\ell+1]
	\end{align*}
	
	Note that with these rankings, for each element $u\in U$, the candidate $c_u$ has bottom count at most $\nu+\mu$. For each element $u\in U$, we add several copies of the ranking $\dots \succ b \succ c_u$ until $c_u$ has bottom count exactly $\nu+\mu$. 
	Thus, the bottom counts of the candidates in the initial profile are as follows: 
	\begin{itemize}
		\item For each $u\in U$, $c_u$ has a count of $\nu+\mu$. 
		\item For each $S\in \mathcal{S}$, $e_S$ has a count of $\nu+\mu+\ell-1$. 
		\item Candidate $b$ has a count of $\nu+\mu+\ell+1$.
		\item The designated candidate $d$ has a count of $\nu+\mu+\ell$. 
	\end{itemize}
	Initially, all element candidates $c_u$ are Veto winners, having the lowest bottom count. 
	Note that the bottom count of uneliminated candidates can only increase over time. 
	Thus, for $d$ to be a Veto winner in some round, each element candidate needs to be either deleted or ranked last in at least $\ell$ additional rankings, and also each set candidate needs to be either deleted or ranked last in at least one additional vote.

	Let $U'=\{u_1, \dots, u_\ell\}$ be a hitting set of size $\ell$. 
	Then, in the first round we eliminate $c_{u_1}$, increasing the count of each other element candidate by $1$, and increasing the score of each set candidate corresponding to a set from $\mathcal{S}_{u_1}$ by $1$. 
	All remaining element candidates are still veto winners and we continue eliminating $c_{u_i}$ for $i = 2, \dots, \ell$.
	After round $\ell$, each remaining element candidate has count $\nu+\mu+\ell$. 
	Moreover, as $U'$ is a hitting set, each set candidate also has count at least $\nu+\mu+\ell$. 
	This means that $d$ is a Veto winner in round $\ell+1=k$ and we can eliminate it.
	
	Conversely, assume that there is an execution of \seqWiAb{Veto} such that $d$ is eliminated in round $\ell+1$ or earlier. 
	As each element candidate either needs to be ranked last in $\ell$ additional rankings or deleted, in the first $\ell$ rounds element candidates need to be eliminated. 
	Let $U'\subseteq U$ be the subset of elements that correspond to the eliminated element candidates. 
	Then in case $U'$ does not form a hitting set, there is a set candidate that still has count $\nu+\mu+\ell-1$ in round $\ell+1$, and thus in particular a lower count than the designated candidate $d$, a contradiction. 
\end{proof}

By applying \Cref{lem:equiv} to \Cref{ob:STV-n}, we get that \PositionK for \seqWiAb{Veto} is fixed-parameter tractable with respect to $n$. 
\begin{corollary}
	\PositionK for \seqWi{Veto} is solvable in $\mathcal{O}(2^n\cdot nm^2)$ time. 
\end{corollary}

\subsection{Borda}
We conclude by studying \seqWiAb{Borda}. Recall \Cref{rem:c2-borda} which showed that it suffices to reason about the weighted majority graph induced by a profile.

 \begin{theorem}\label{th:seqWi-Borda-n}
 	\TopK for \seqWi{Borda} is NP-complete for $n = 8$.
 \end{theorem}
 
 \begin{proof}
 	We reduce from \textsc{Cubic Vertex Cover}. Let $G=(V,E)$ be a graph with $q$ vertices where each vertex $v \in V$ is incident to exactly 3 edges, and let $t$ be the target vertex cover size. 
 	We construct an instance of the \TopK problem as follows.
 	
 	The candidate set consists of one candidate for each vertex, one candidate for each edge, a designated candidate $d$, and dummy candidates $f$, candidates $B = \{b_1, \dots, b_6\}$, and candidates $H = \{h_1, \dots, h_{q+3-t}\}$.
 	Let $C = V \cup E \cup \{d\} \cup \{f\} \cup B \cup H$. 
 	We set $k = t + 7$.
 	The ranking profile will be constructed so as to induce a desired weighted majority graph, where all arcs will have weight 2. The arcs are as follows:
 	\begin{align*}
 		A = &\{ (e, v) \in E \times V : e \in E, v\in e \}  \cup {} \\
 		&\{ (d,f) \} \cup (V\times B) \cup (B \times H).
 	\end{align*}
 	
 	\begin{figure}
 		\centering
 	\begin{tikzpicture}
 		[multinode/.style={circle,draw=black,inner sep=1.5pt,minimum width=10pt},
 		element/.style={circle,draw=black,inner sep=1.5pt,minimum width=4pt},
 		set/.style={rectangle,draw=black,inner sep=3pt,minimum width=4pt},
 		every edge quotes/.style={fill=white,inner sep=0.5pt},
 		every label/.style={color=red!60!black}]
 		\node [label={[label distance=-5pt]30:$4$}] (e) {$e$};
 		\node [below left =0.6cm and 0.4cm of e] (v1) {$v_1$};
 		\node [label={[label distance=-6pt]30:$6$},below right=0.6cm and 0.4cm of e] (v2) {$v_2$};
 		\node [multinode, label={[label distance=4pt]20:$6-2t$}, below=1.8cm of e] (B) {$B$};
 		\node [multinode, label={[label distance=-2pt]20:$-12$}, below=0.6cm of B] (H) {$H$};
 		\node [label={[label distance=-3pt]30:$2$},right=1.2cm of v2] (d) {$d$};
 		\node [label={[label distance=-6pt]30:$-2$},below of=d] (f) {$f$};
 		\draw[-latex]
 		(e) edge (v1)
 		edge (v2)
 		(v1) edge (B)
 		(v2) edge (B)
 		(d) edge (f)
 		(B) edge (H);
 	\end{tikzpicture}
 	\caption{An illustration of the reduction of \Cref{th:seqWi-Borda-n}. All arcs have weight 2. Red superscripts denote the difference between the weight of outgoing and ingoing arcs for a candidate.} \label{fi:seqWi-Borda-n}
 \end{figure}
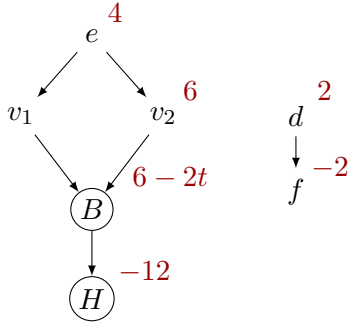

 	The constructed weighted majority graph is depicted in \Cref{fi:seqWi-Borda-n}. 
 	We now describe how to write the weighted majority graph $(C,A)$ as a sum of 4 bilevel graphs (as defined in \Cref{lem:bilevel}). 
 	For each vertex $v$, we  label the three edges incident to it arbitrarily as $e_v^1$, $e_v^2$, $e_v^3$. Consider the following arc sets:
 	\begin{align*}
 		A_1 &= (V \times B), \\
 		A_2 &= \{ (e_v^1, v) : v \in V \} \cup  (B \times H), \\
 		A_3 &= \{ (e_v^2, v) : v \in V \} \cup \{(d,f)\}, \\
 		A_4 &= \{ (e_v^3, v) : v \in V \}.
 	\end{align*}
 	It is clear that each of these sets describe bilevel graphs, that they are pairwise disjoint, and that $A = A_1 \cup A_2 \cup A_3 \cup A_4$.
 	By invoking \Cref{lem:bilevel}, we get a profile $P$ containing $8$ voters with the depicted weighted majority graph.  
 	
 	The C2-Borda scores (for a definition, see \Cref{rem:c2-borda}) in this profile are:
 	\begin{itemize}
 		\item $d$ has score 2
 		\item each $b \in B$ has score $6-2t$
 		\item each $v \in V$ has score 6 (since $v$ is incident to 3 edges, and beats 6 $b$-candidates)
 		\item each $e \in E$ has score 4 (since $e$ is incident to 2 vertices).
 		\item each $h \in H$ as well as $f$ have negative scores (and can only have non-positive scores throughout the elimination process because they do not beat any candidates) and will not be selectable in the first $k$ rounds.
 	\end{itemize}
 	
 	Suppose that $T = \{v_1, \dots, v_t\}$ is a vertex cover of $G$.
 	Then the following is a valid start of an elimination ordering, with $d$ eliminated in round $k = t + 7$: 
 	First, we eliminate all candidates from $T$ in some arbitrary ordering, then all candidates from $B$ in some arbitrary ordering, and then $d$.
 	Explicitly, in the first $t$ rounds, the maximum Borda score of a candidate is $6$ and all vertex candidates have a Borda score of $6$ (no other candidates have a Borda score of $6$ in these rounds). Thus we can select members of the vertex cover $T$ in each of these rounds. In round $t+1$, the remaining vertex candidates and the candidates in $B$ have the maximum Borda score of $6$. Thus, we can eliminate candidates in $B$ (while doing so, the Borda scores of vertex candidates decrease, so candidates in $B$ continue having the maximum Borda score). After all $6$ candidates in $B$ are eliminated, we are in the following situation with respect to the remaining candidates' C2-Borda score:
 	\begin{itemize}
 		\item $d$ has score $2$
 		\item each $e$ has score at most $2$ (since we have eliminated a vertex cover, and thus have eliminated at least one candidate that $e$ beats)
 		\item each remaining $v \in V$ has score $-6$; each $h \in H$ and $f$ have non-positive scores.
 	\end{itemize}
 	Hence, at this point, candidate $d$ has the highest Borda score and can be eliminated.
 	
 	Conversely, suppose there is a ranking selected by \seqWi{Borda} where $d$ is eliminated in round $k = t+7$ or earlier. As observed above, in the first $t$ rounds, only vertex candidates can be eliminated. Let $T = \{v_1, \dots, v_t\}$ be the set of vertices whose candidates are eliminated in these rounds. From round $t+1$ until $t+6$, all candidates in $B$ have score at least $6$ (it cannot go lower because the candidates that $B$ beats, namely $H$, cannot be eliminated). Because $d$ has score only $2$,  all the $6$ candidates in $B$ are eliminated before $d$. This brings us to round $t + 7$ where by assumption $d$ is eliminated. Hence at this point, no candidate has score higher than $2$. In particular, for every edge $e \in E$, its score is less than $4$. This can only have happened if at least one of the vertices incident to $e$ has been eliminated and is thus part of $T$. It follows that $T$ is a vertex cover.
 \end{proof}

 \begin{theorem}\label{th:seqWi-Borda-k}
	\TopK for \seqWi{Borda} is W[2]-hard with respect to $k$.
\end{theorem}

\begin{proof}
	We reduce from \textsc{Hitting Set}, using a similar construction as in \Cref{th:seqWi-Borda-n}. Let $U$ be a given universe of elements and let $\mathcal{S}$ be a given family of subsets of $U$. We are also given an integer $t$, and the question is whether there is a $t$-subset of the universe containing at least one element from each set from $\mathcal{S}$, i.e., $U'\subseteq U$ with $|U'|=t$ and $S\cap U'\neq \emptyset$ for all $S\in \mathcal{S}$. 
	Let $q = |U|$. 
	
	We construct an instance of the \TopK problem as follows.
	
	We first give an incomplete description of the constructed instance. Later, we will add some dummy candidates that have no influence except that they increase the Borda scores of some of the candidates to a desired level. The candidate set consists of one candidate for each element, one candidate for each set $S \in \mathcal{S}$, a designated candidate $d$, and a set $B = \{b_1, b_2\}$ of two blocking candidates.
	Let $C_{\text{base}} = V \cup E \cup \{d\} \cup B$ (again, we will add to this set later). 
	The ranking profile will be constructed so as to induce a desired weighted majority graph, where all arcs will have weight 2 (using standard arguments; \citealp{McGa53a,Debo87a}). The arcs are as follows:
	\begin{align*}
		A_{\text{base}} = &\{ (S, e) \in \mathcal{S} \times U : e\in S \}  \cup (U\times B)
	\end{align*}
	
	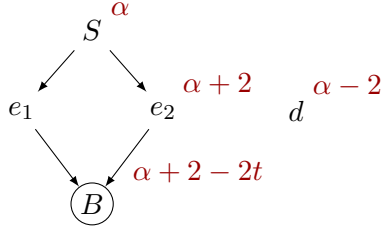
\begin{figure}
		\centering
		\begin{tikzpicture}
			[multinode/.style={circle,draw=black,inner sep=1.5pt,minimum width=10pt},
			element/.style={circle,draw=black,inner sep=1.5pt,minimum width=4pt},
			set/.style={rectangle,draw=black,inner sep=3pt,minimum width=4pt},
			every edge quotes/.style={fill=white,inner sep=0.5pt},
			every label/.style={color=red!60!black}]
			\node [label={[label distance=-5pt]30:$\alpha$}] (S) {$S$};
			\node [below left =0.6cm and 0.4cm of e] (e1) {$e_1$};
			\node [label={[label distance=-6pt]30:$\alpha+2$},below right=0.6cm and 0.4cm of e] (e2) {$e_2$};
			\node [multinode, label={[label distance=4pt]20:$\alpha+2-2t$}, below=1.8cm of e] (B) {$B$};
			\node [label={[label distance=-5pt]30:$\alpha-2$},right=1.2cm of v2] (d) {$d$};
			\draw[-latex]
			(S) edge (e1)
			edge (e2)
			(e1) edge (B)
			(e2) edge (B);
		\end{tikzpicture}
		\caption{An illustration of the reduction of \Cref{th:seqWi-Borda-k}. All arcs have weight 2. Red superscripts denote the C2-Borda score of the candidates (after adding dummy candidates).} \label{fi:seqWi-Borda-k}
	\end{figure}
	
	The constructed weighted majority graph is depicted in \Cref{fi:seqWi-Borda-k}.
	
	The C2-Borda scores (for a definition, see \Cref{rem:c2-borda}) in this profile are:
	\begin{itemize}
		\item $d$ has score 0
		\item each $S \in \mathcal{S}$ has score $2|S|$
		\item each $e \in U$ has score $2|\{S : e \in S\}| - 4$
		\item each $b \in B$ has score $-2|U|$.
	\end{itemize}
	Choose $\alpha$ to be an even integer such that the number $\alpha + 2 - 2t$ is larger than all the Borda scores just mentioned. (Clearly we can take an $\alpha$ that is polynomial size.)
	
	Now we go through every candidate $c \in C_{\text{base}}$ and add dummy candidates $D_c$ and arcs $\{c\} \times D_c$ to increase the C2-Borda score so that we now have the following scores (see \Cref{fi:seqWi-Borda-k}):
	\begin{itemize}
		\item $d$ has score $\alpha-2$
		\item each $S \in \mathcal{S}$ has score $\alpha$
		\item each $e \in U$ has score $\alpha+2$
		\item each $b \in B$ has score $\alpha +2 -2t$.
	\end{itemize}
	Explicitly, if a candidate's target C2-Borda score is $R$ and its current score is $S < R$, then we need to add $|D_c| = (R-S)/2$ dummy candidates.
	Thus, the complete candidate set is $C = C_{\text{base}} \cup \bigcup_{c\in C_{\text{base}}} D_c$ and the final arc set is $A = A_{\text{base}} \cup \bigcup_{c\in C_{\text{base}}} (\{c\} \times D_c)$. This weighted majority graph can be realized by a polynomial-size ranking profile \citep{McGa53a}.%
	\footnote{Note that we are not claiming that this weighted majority graph can be realized by constantly many voters. The reason that a construction like in \Cref{th:seqWi-Borda-n} does not translate is that we do not have a constant upper bound on the number of occurrences of each element. Such bounds typically cannot be imposed while retaining W[1]-hardness.}
	We set $k = t + 3$.
	
	We now prove that our reduction is correct. Note that throughout the argument, all the dummy candidates have non-positive C2-Borda score and it will never be possible to eliminate any of them, so they can essentially be ignored.
	
	Suppose that $T = \{e_1, \dots, e_t\}$ is a hitting set.
	Then the following is a valid start of an elimination ordering, with $d$ eliminated in round $k = t + 3$: 
	First, we eliminate all candidates from $T$ in some arbitrary ordering, then the two candidates from $B$ in some arbitrary ordering, and then $d$.
	Explicitly, in the first $t$ rounds, the maximum Borda score of a candidate is $\alpha + 2$ and all element candidates have a Borda score of $\alpha + 2$ (no other candidates have a Borda score of $\alpha + 2$ in these rounds). Thus we can select members of the hitting set $T$ in each of these rounds. In round $t+1$, the remaining element candidates and the candidates in $B$ have the maximum Borda score of $\alpha + 2$. Thus, we can eliminate candidates in $B$ (while doing so, the Borda scores of element candidates decrease, so candidates in $B$ continue having the maximum Borda score). After all $2$ candidates in $B$ are eliminated, we are in the following situation with respect to the remaining candidates' C2-Borda score:
	\begin{itemize}
		\item $d$ has score $\alpha-2$
		\item each $S \in \mathcal{S}$ has score at most $\alpha-2$ (since we have eliminated a hitting set, and thus have eliminated at least one candidate that $S$ beats)
		\item each remaining $e \in U$ has score $\alpha-2$
		\item dummy candidates have non-positive score
	\end{itemize}
	Hence, at this point, candidate $d$ has the highest Borda score and can be eliminated.
	
	Conversely, suppose there is a ranking selected by \seqWi{Borda} where $d$ is eliminated in round $k = t+3$ or earlier. As observed above, in the first $t$ rounds, only element candidates can be eliminated. Let $T = \{e_1, \dots, e_t\}$ be the set of elements whose candidates are eliminated in these rounds. In rounds $t+1$ and $t+2$, all candidates in $B$ have score at least $\alpha+2$ (it cannot go lower because the candidates that $B$ beats are all dummy candidates which cannot be eliminated). Because $d$ has score only $\alpha-2$,  the two candidates in $B$ must be eliminated before $d$. This brings us to round $t + 3$ where by assumption $d$ is eliminated. Hence at this point, no candidate has score higher than $\alpha-2$. In particular, for every set $S \in \mathcal{S}$, its score is less than $\alpha$. This can only have happened if at least one of the elements of $S$ has been eliminated and is thus part of $T$. It follows that $T$ is a hitting set.
\end{proof}

\end{document}